\documentclass[journal, romanappendices]{IEEEtran}

	\usepackage[utf8]{inputenc} 
	\usepackage[T1]{fontenc}
	\usepackage{url}
	\usepackage{ifthen}
	\usepackage{cite}
	\usepackage[cmex10]{amsmath} 
\usepackage{amsthm}

	\usepackage{pgfplots}
	\usepackage{bm}
	\usepackage{amsbsy}
	\usepackage{epsfig} 
	\usepackage{latexsym} 
	\usepackage{amssymb}
	\usepackage{wasysym}
	\usepackage{placeins}
	\usepackage{color}
	\usepackage{rotating}
	\usepackage{dsfont}
	\usepackage{caption}
	\usepackage{subcaption}
	\usepackage{algorithm}
	\usepackage{graphicx}
	\usepackage[shortlabels]{enumitem}
	\usepackage{epstopdf}
	\usepackage{cancel}
	\usepackage{algpseudocode}
	\usepackage{units}
	\usepackage{multirow}
	\usepackage{cases}
	\usepackage{tabularx}
\usepackage[ddmmyyyy]{datetime}

\usepackage{csquotes}
\usepackage{tikz}

\usepackage[percent]{overpic}

	\usepackage[final]{microtype} 
	
\usepackage{hyperref}
	
\usepackage{subfiles} 
	
\usetikzlibrary{shapes,arrows,snakes,positioning}
\usetikzlibrary{decorations.markings}
\usetikzlibrary{calc}
\usetikzlibrary{fit}
\tikzset{%
	highlight/.style={rectangle,rounded corners,fill=red!15,draw,fill opacity=0.5,thick,inner sep=0pt}
}
\newcommand{\tikzmark}[1]{\tikz[overlay,remember picture] \node (#1) {};}

\newcommand{\tikzdrawbox}[3][(0pt,0pt)]{%
    \tikz[overlay,remember picture]{
    \draw[#3]
      ($(left#2)+(-0.7em,1em) + #1$) rectangle
      ($(right#2)+(-2.3em,-0.6em) - #1$);}
}

\usepackage{mathtools}
\usepackage[mathscr]{euscript}

\newcommand{\ellop}{\bar{\ell}}
\newcommand{\bZero}{\boldsymbol{0}}

\newcommand{\expM}{{^{(M)}}}

\newcommand{\PfracK}{P}
\newcommand{\eqmo}[1]{\begin{equation}\begin{aligned}#1\end{aligned}\end{equation}}

\newcommand{\varleq}{\leq\kern-0.6em{\raisebox{.5ex}{\text{\tiny{{v}}}}}\kern0.60em}

\newcommand{\zopi}{\zop\expo_{\inv}}
\newcommand{\zzop}{z\expo_n}

\newcommand\itemEq[1][]{%
  \ifx\relax#1\relax  \item \else \item[#1] \fi
  \abovedisplayskip=0pt\abovedisplayshortskip=0pt~\vspace*{-\baselineskip}}
	
\newcommand{\dd}{\mathop{}\!\mathrm{d}}

\newcommand{\itb}{\begin{itemize}}
\newcommand{\ite}{\end{itemize}}
\newcommand{\enb}{\begin{enumerate}}
\newcommand{\ene}{\end{enumerate}}
\newcommand{\eqm}[1]{\begin{align}#1\end{align}}

\newcommand{\limpf}{\lim_{P\rightarrow\infty}}
\newcommand{\setcomp}{{\mathsf{c}}}	
\newcommand{\Trans}{{\mathrm{T}}}

\newcommand{\orthP}[1]{\bP_{\hv^{\bot}_{#1}}}

\newcommand{\LB}{\left(}
\newcommand{\RB}{\right)}
\newcommand{\LSB}{\left[}
\newcommand{\RSB}{\right]}

\newcommand{\cmtblu}[1]{{\color{blue!70!black}{#1}}}

\newcommand{\abs}[1]{\left|#1\right|}
\newcommand{\absn}[1]{|#1|}

\newcommand{\bXh}{\hat{\bX}}

\newcommand{\h}{{\mathrm{h}}}

\newcommand{\bhh}{\hat{\hv}}

\newcommand{\bHH}{{{\hat{\bH}}}}

\newcommand{\thv}{\tilde{\hv}}

\newcommand{\zop}{\breve{z}}

\newcommand{\iop}{{\bar{i}}}
\newcommand{\expj}{^{(j)}}
\newcommand{\expo}{^{(1)}}
\newcommand{\expt}{^{(2)}}

\newcommand{\Pb}{{\bar{P}}}

\newcommand{\DeltaR}{\Delta\mathrm{R}}

\newcommand{\ExpB}[1]{\mathbb{E}\!\LSB#1\RSB}	
\newcommand{\ExpBs}[2]{\mathbb{E}_{#1}\!\LSB#2\RSB}	
\newcommand{\ExpHC}[1]{\mathbb{E}\!\LSB#1\RSB}

\DeclareMathAlphabet{\mathbit}{OML}{cmr}{bx}{it}
\DeclareMathAlphabet{\mathsf}{OT1}{cmss}{m}{n}
\DeclareMathAlphabet{\mathTXf}{OT1}{cmss}{bx}{it}
\DeclareMathOperator{\cov}{cov}

\DeclareMathOperator{\rank}{rank}

\DeclareMathOperator*{\argmax}{argmax}

\DeclareMathOperator{\Real}{Re}
\DeclareMathOperator{\Imag}{Im}
\DeclareMathOperator{\Exp}{{\mathbb{E}}}

\DeclareMathOperator{\CN}{\mathcal{N}_{\mathbb{C}}}

\DeclareMathOperator{\inv}{inv}

\DeclareMathOperator{\DoF}{DoF}

\DeclareMathOperator{\RX}{RX}

\DeclareMathOperator{\SNR}{SNR}

\newcommand{\Nb}{{{\mathbb{N}}}}
\newcommand{\Zb}{{{\mathbb{Z}}}}

\newcommand{\Rb}{{{\mathbb{R}}}}
\newcommand{\Cb}{{{\mathbb{C}}}}

\newcommand{\Eb}{{{\mathbb{E}}}}

\newcommand{\an}{\mathrm{a}}
\newcommand{\bbn}{\mathrm{b}}

\newcommand{\iin}{\mathrm{i}}

\newcommand{\sn}{\mathrm{s}}

\newcommand{\yn}{\mathrm{y}}

\newcommand{\dv}{\mathbf{d}}
\newcommand{\ev}{\mathbf{e}}

\newcommand{\hv}{\mathbf{h}}

\newcommand{\nv}{\mathbf{n}}

\newcommand{\sv}{\mathbf{s}}
\newcommand{\tv}{\mathbf{t}}
\newcommand{\uv}{\mathbf{u}}
\newcommand{\vv}{\mathbf{v}}
\newcommand{\wv}{\mathbf{w}}
\newcommand{\xv}{\mathbf{x}}
\newcommand{\yv}{\mathbf{y}}

\newcommand{\Dc}{{{\mathcal{D}}}}

\newcommand{\Fc}{{{\mathcal{F}}}}

\newcommand{\Hc}{{\mathcal{H}}}
\newcommand{\Ic}{{{\mathcal{I}}}}

\newcommand{\Lc}{{{\mathcal{L}}}}

\newcommand{\Nc}{{{\mathcal{N}}}}
\newcommand{\Oc}{{{\mathcal{O}}}}
\newcommand{\Pc}{{{\mathcal{P}}}}
\newcommand{\Qc}{{{\mathcal{Q}}}}
\newcommand{\Rc}{{{\mathcal{R}}}}

\newcommand{\Vc}{{{\mathcal{V}}}}

\newcommand{\pset}{\ensuremath{\Pc_{\!\!o}}}

\newcommand{\bA}{\mathbf{A}}

\newcommand{\bG}{\mathbf{G}}
\newcommand{\bH}{\mathbf{H}}
\newcommand{\bI}{\mathbf{I}}

\newcommand{\bP}{\mathbf{P}}

\newcommand{\bT}{\mathbf{T}}
\newcommand{\bU}{\mathbf{U}}
\newcommand{\bV}{\mathbf{V}}
\newcommand{\bW}{\mathbf{W}}
\newcommand{\bX}{\mathbf{X}}

\newcommand{\bs}{\bm{s}}

\newcommand{\balpha}{{\boldsymbol{\alpha}}}
\newcommand{\bdelta}{\bm{\delta}}
\newcommand{\bDelta}{\boldsymbol{\Delta}}

\newcommand{\bphi}{\ensuremath{\boldsymbol{\phi}}}

\newcommand{\Z}{\mathbf{0}} 
\newcommand{\U}{\mathbf{1}} 
\newcommand{\norm}[1]{\lVert{#1}\rVert}

\newcommand{\E}{{\mathbb{E}}}

\newcommand{\He}{{{\mathrm{H}}}}

\newtheorem{theorem}{Theorem}
\newtheorem{definition}{Definition}
\newtheorem{proposition}{Proposition}
\newtheorem{lemma}{Lemma}
\newtheorem{corollary}{Corollary}

\newtheorem{remark}{Remark}

\renewcommand{\cmtblu}[1]{#1}

\begin{document}
	\title{Asymptotically Achieving Centralized Rate on the Decentralized Network MISO Channel}
	\author{Antonio Bazco-Nogueras,~\IEEEmembership{Member,~IEEE,} 
					Paul de Kerret,~\IEEEmembership{Member,~IEEE,}\\ ~~~~
					David Gesbert,~\IEEEmembership{Fellow,~IEEE,} and 
					Nicolas Gresset,~\IEEEmembership{Senior Member,~IEEE}%
		\thanks{This work was in part supported by the European Research Council under the European Union’s Horizon 2020 research and innovation program (Agreement no.~670896). (Corresponding author: Antonio Bazco-Nogueras).}
		\thanks{A. Bazco-Nogueras was with the Mitsubishi Electric R\&D Centre Europe, Rennes, France, and also with the Communication Systems Department, EURECOM, Sophia-Antipolis, France. He is now with the IMDEA Networks Institute, Madrid, Spain  (e-mail:antonio.bazco@imdea.org).}
		\thanks{P. de Kerret was with the Communication Systems Department, EURECOM, Sophia-Antipolis, France; he is now with Greenly.}
		\thanks{D. Gesbert is with the Communication Systems Department, EURECOM, Sophia-Antipolis, France (e-mail: gesbert@eurecom.fr). }
		\thanks{N. Gresset is with the Mitsubishi Electric R\&D 
		Centre Europe, Rennes, France (e-mail: n.gresset@fr.merce.mee.com).}
}

	\maketitle

	\begin{abstract} 
		In this paper, we analyze the high-SNR regime of the $M\times K$ Network MISO channel in which each transmitter has access to a different channel estimate, possibly with different precision. 
		It has been recently shown that, for some regimes, this setting attains the same Degrees-of-Freedom as the ideal centralized setting with perfect Channel State Information (CSI) sharing, in which all the transmitters are endowed with the best estimate available at any transmitter. 
		This result is restricted by the limitations of the Degrees-of-Freedom metric, as it only provides information about the slope of growth of the capacity as a  function of the SNR, without any insight about the possible performance at a given SNR. 
		\cmtblu{In order to overcome this limitation, we analyze the affine approximation of the rate on the high-SNR regime for this decentralized Network MISO setting for the antenna configurations in which it achieves the Degrees-of-Freedom of the centralized setting. We show that, for a regime of antenna configurations, it is possible to asymptotically attain the same achievable rate as in the ideal centralized scenario. 
		Consequently, it is possible to achieve the beamforming gain of the ideal perfect-CSI-sharing setting even if only a subset of transmitters is endowed with precise CSI, which can be exploited in scenarios such as distributed massive MIMO where the number of transmit antennas is much bigger than the number of served users. 		
		This outcome is a consequence of the synergistic compromise between CSI precision at the transmitters and consistency between the locally-computed precoders, which is an inherent trade-off of decentralized settings that does not exist in the centralized CSI configuration.}
		We propose a precoding scheme achieving the previous result, which is built on an uneven structure in which some transmitters reduce the precision of their own precoding vector for the sake of using transmission parameters that can be more easily predicted by the other transmitters. 
	\end{abstract} 
	
	
		\section{Introduction}\label{se:intro}
	
			\subsection{Cooperative Transmission}\label{se:cooperative}
			Multi-user cooperative networks and the extend of their theoretical capabilities have been thoroughly analyzed in the literature\cite{Costa1987,Caire2003,Lozano2005,Hassibi2007}. 
			Initially, it was shown that cooperation can bring multiplicative gains under ideal assumptions on the channel knowledge available at the communicating nodes. 
			As a matter of example, it is known that, under the assumption of perfect Channel State Information (CSI), the setting in which $M$ transmitters (TXs) jointly serve $K$ single-antenna users (the so-called Network MISO channel) achieves a rate that scales as  $\min(N_T,K)$ times the rate of the single-antenna point-to-point channel\cite{Jindal2005}, where $N_T$ is the total number of transmit antennas. 
			Conversely, the $K\times K$ interfering channel can attain a multiplexing gain of $K/2$\cite{Cadambe2008}. This perfect CSI scenario has been profoundly studied\cite{Viswanath2003,Jindal2005,Weingarten2006,Etkin2008}. 
			Unfortunately, the assumption of perfect information is not practical in most of the current wireless networks. 
			Motivated by the infeasibility of the previous assumption, the community has investigated settings in which the information available at the communicating nodes does not meet the perfect CSI assumption, such as scenarios where the information available is imperfect	\cite{Kountouris2006,Motahari2009,Huh2012,Piovano2017,Davoodi2016_TIT_DoF}, partial~\cite{Tandon2012b,Lashgari2016}, 
			or delayed	\cite{MaddahAli2012,Xu2012,Yang2013,Vahid2016}. 
			
			Even though the aforementioned works considered an imperfect acquisition or estimation of the CSI, all the cooperating nodes are assumed to \emph{perfectly} share the same \emph{imperfect} information. 
			Yet, current and upcoming wireless networks characteristics make this assumption of \emph{perfect sharing} also impractical for many applications. 
			This is due to, for example, the proliferation of heterogeneous networks for which some of the nodes have a wireless, fluctuating, or limited backhaul\cite{Simeone2009,Jaber2016}, hierarchical distributed networks (with remote radio-heads of  different capabilities), 
			or Ultra-Reliable Low-Latency Communication (URLLC) applications\cite{Popovski2018,Bennis2018}, in which the perfect sharing of the information would result in an intolerable  
			delay. 
			Settings in which simple devices with low capabilities aim to communicate in a dense environment, as in IoT applications, also fall into the use cases in which the sharing of channel information is indispensable but challenging.
			This evolution of different use cases boosts the interest of distributed information settings, in which the information available at the communicating nodes is not only imperfect but different from one node to another. 
			This type of settings can be formalized as the so-called Team Decision problems\cite{Radner1962}, in which different agents aiming for the same goal attempt to cooperate in the absence of perfect communication among~them. 
			\subsection{Distributed CSIT Setting}\label{se:distr_csi}
			There exists a great number of different distributed settings \cite{Grandhi1994,Ng2008,Dimakis2010,Aggarwal2012,Vahid2016_erasure, lampiris2020resolving}. 
			In particular, we focus on the scenario where $M$ TXs jointly serve $K$ different single-antenna users (RXs), which is often referred to as the Network MISO setting, under the assumption of distributed CSI at the TXs (CSIT). 
			In such scenario,  every TX has access to the information symbols of every RX, but it does not share the same CSI with the other TXs.  
			This setting arises in situations in which the data can be buffered or cached, but the CSI needs to be available with very small delay; for example, for the transmission of high-popularity content that can be cached at the network edge before the transmission occurs\cite{lampiris2020resolving,MaddahAli2014,Zhang2017}, or in high mobility scenarios and IoT (or V2X) networks with fast varying channels but low data rate\cite{Hail2015,Wang2017}. 
			
			The Network MISO with distributed CSIT has been analyzed in recent works. 
			Initially, it was shown in~\cite{dekerret2012_TIT} that the $2\times 2$ single-antenna scenario in which one TX has better knowledge of the full channel matrix that the other TX achieves the same Degrees-of-Freedom (DoF) as the ideal centralized case in which  both TXs are endowed with the best CSI. 
			The DoF  metric, which will be presented in the following section, is defined as the pre-logarithmic factor of the capacity as function of the SNR\cite{Caire2003,Cadambe2008,Etkin2008}, and it is also known as multiplexing gain. 
			The outcome of\cite{dekerret2012_TIT} derives from a precoding with a master-slave structure, named Active-Passive Zero-Forcing (AP-ZF), in which the TX whose CSI is less precise transmits with a fixed precoder. 
			Although this result could depend on the asymmetric structure of the setting, it has been extended to more general settings. 
			Indeed, it was shown in\cite{Bazco2018_WCL} that the Generalized DoF\footnote{Generalized DoF refers to the DoF analysis under the assumption that the difference of channel strengths between the links does not vanish at high SNR. See\cite{Etkin2008} for more details.}  
			\cmtblu{of the centralized setting with perfect CSIT sharing are attained in the distributed setting, no matter which TX has the best estimation of every single channel coefficient.} 
			By way of explanation, the pre-logarithmic factor of the centralized $2\times 2$ Network MISO is preserved as long as the estimate of a certain link is available at one of the TXs. 
			This comprises for example cases in which TXs have only local CSI or in which each TX knows better the channel towards a certain RX. 
			The DoF analysis has been also applied to the  $K\times K$ scenario, in which two main insights can be outlined: 
			First, the optimal DoF of the centralized setting with perfect CSIT sharing is also achieved for some regimes of the $K\times K$ scenario\cite{Bazco2018_TIT}. 
			Second, the quantization of the information available at a certain TX can be beneficial in distributed settings, inasmuch as it helps to transform the setting into a hierarchical configuration in which the structure of the CSI allocation can be used to increase the DoF\cite{dekerret2016_asilomar}. 
			
			These results provide some understanding on the resilience of cooperating settings under information mismatches between different nodes. 
			However, the DoF metric is a limited metric because it only provides information about the pre-log factor, not offering any information about the achievable rate at any given SNR. 
			For that reason, it is necessary to take a step beyond and analyze the affine approximation of the rate at the high-SNR regime. 
		
			\subsection{Linear Approximation of Rate at High-SNR}\label{se:linear_approx}
				Finding the fundamental limits of communication in distributed settings is a challenging problem.  
				Indeed, these limits have remained open  even for several important centralized cases. 
				Nevertheless, it is possible to find significant insights on those settings through rate approximations, which help us to move towards a better understanding of the behavior of complex wireless networks. 
				
				\begin{figure}%
					\centering
%
%
\definecolor{mycolor1}{rgb}{0.00000,0.44700,0.74100}%
	\begin{tikzpicture}
		\begin{axis}[%
					width=1.7*0.5\columnwidth,
					height=1.7*.42\columnwidth,
					at={(0\columnwidth,0\columnwidth)},
					scale only axis,
					xmin=0,
					xmax=100,
					xlabel={SNR [dB]},
					x label style={at={(axis description cs:0.5,0.05)},anchor=north},
					ymin=-3,
					ymax=300,
					ylabel={Rate [bits/s/Hz]},
					y label style={at={(axis description cs:0.1,.45)},rotate=0,anchor=south},						
					axis background/.style={fill=white},
					title style={font=\bfseries},
					axis x line*=bottom,
					axis y line*=left,
					yticklabels={,,},
					xticklabels={,,},
		]
			\addplot[samples=201,domain=0:35,blue,line width=1pt] {6/70*x^2 + 35} node[above,pos=0.9]{};								
			\addplot[samples=201,domain=35:100,blue,line width=1pt] {(6*x - 100) + 30} node[above,pos=0.9]{};								
			\addplot[samples=201,domain=70/6:100,blue,dashed,line width=0.5pt] {(6*x - 105) + 30} node[above,pos=0.9]{};				
			\addplot[samples=201,domain=69.9:100,orange!80!blue,line width=1pt] {6*x - 274.4878} node[above,pos=0.9]{};									
			\addplot[samples=201,domain=20:69.9, orange!80!blue,line width=1pt] {6*exp((x-10)/20)-5  + 30} node[above,pos=0.9]{};								
			\addplot[samples=201,domain=0:20, orange!80!blue,line width=1pt] {0.015*x^2 - 0.1054*x + 1  + 30} node[above,pos=0.9]{};		
			\addplot[samples=201,domain=45:100,orange!80!blue,dashed,line width=0.5pt] {6*x - 279.4878} node[above,pos=0.9]{};											
										%
										%

			\addplot[samples=201,domain=0:46.15, ->, orange!80!blue,densely dotted,line width=0.5pt] {5} node[above,pos=0.7]{$\small\Lc^B_\infty$};		
			\addplot[samples=201,domain=0:12, ->, blue,densely dotted,line width=0.5pt] {2} node[above,pos=0.5]{$\small\Lc^A_\infty$};		
			
			\draw[densely dotted, blue] (44,195) node[anchor=north]{$$}
				-- (54,255) node[anchor=north]{$$}
				-- (54,195) node[anchor=south]{\phantom{zef}~\quad$\small\DoF_{\!A}$}
				-- cycle;					
			\draw[densely dotted, orange!80!blue] (71.5,160) node[anchor=north]{$$}
				-- (81.5,220) node[anchor=north]{$$}
				-- (81.5,160) node[anchor=south]{\phantom{zef}~\quad $\small\DoF_{\!B}$}
				-- cycle;					
														
			\node at (25,270) {\fbox{$\DoF_{\!A} = \DoF_{\!B}$}};

		\end{axis}
	\end{tikzpicture}%
%
						\caption{Qualitative illustration of the affine approximation of two different setting with the same DoF (slope) but different rate offset $\Rc_\infty = \DoF\Lc_\infty$, and hence different achieved rate. }\label{fig:affine}%
				\end{figure}
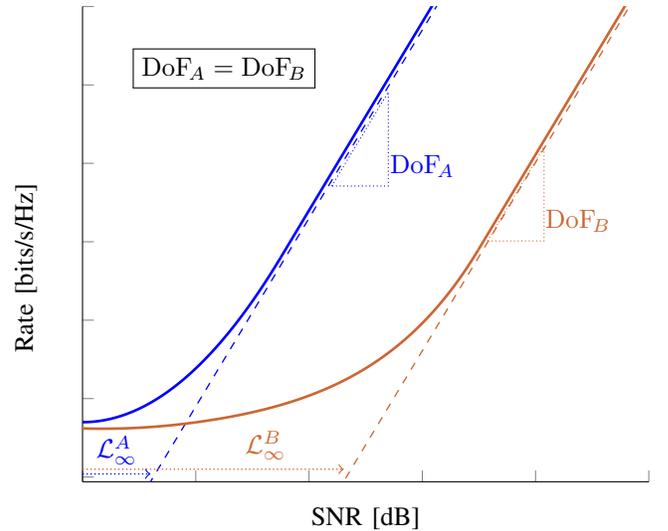
				A very useful metric that has been applied in the literature is the affine approximation of the rate at high SNR, introduced in~\cite{Shamai2001_TIT}. 
				Following this linear approximation, the rate can be written as
					\eqm{
							R = \DoF\log_2(P) - \Rc_\infty + o(1),  \label{eq:lin_approx} 
					}
				where $P$ denotes the SNR, $\DoF$ is the pre-logarithmic factor (or Degrees-of-Freedom), and $\Rc_\infty$ is the \emph{rate offset} (or vertical offset). 
				The approximation in~\eqref{eq:lin_approx} can also be written in terms of the \emph{power offset} (horizontal offset) $\Lc_\infty$, where $\Rc_\infty = \DoF\Lc_\infty$. 
				An illustrative representation is shown in Fig.~\ref{fig:affine}. 
				The term $\Lc_\infty$ represents the zero-order  term with respect to a reference setting with the same slope but whose affine approximation intersects the origin. These terms are defined as
					\eqm{
							\DoF ~&\triangleq \limpf \frac{R(P)}{\log_2(P)} \\
							\Rc_\infty ~& \triangleq \limpf \DoF\log_2(P) - R(P),
					}
				where $R(P)$ represents the rate as function of the SNR $P$. 
				This measure 
				has already proved instrumental in several findings. 
				In\cite{Lozano2005}, Lozano \emph{et al.} analyzed the multiple-antenna point-to-point scenario, revealing that some system features that do not impact the DoF (as antenna correlation, fading, etc.) do considerably impact the zero-order term, affecting the performance of the system at any possible SNR. 
				In addition to exposing some limitations of the DoF metric, \cite{Lozano2005} also revealed that the affine expansion offers appreciably tight approximations also at medium-to-low SNR. 
				This characterization has been also established for the Broadcast Channel (BC) with perfect CSIT using Dirty-Paper Coding and linear precoding \cite{Lee2007}, and for the BC with imperfect CSIT \cite{Jindal2006}. 
				In \cite{Jindal2006}, the BC setting with quantized feedback	was studied under the assumption of Zero-Forcing (ZF) schemes, showing that the CSIT error variance must be proportional to $\SNR^{-\alpha}$ in order to attain a DoF per user of $\DoF_{\RX i} = \alpha$. 
				Furthermore, having a CSIT error variance scaling with $\SNR^{-\alpha}$ was shown to be equivalent to obtaining a quantized feedback of $\alpha\log_2(\SNR)$ bits from the RX, which could be attained if the feedback resources scale proportionally to~$\log_2(\SNR)$. 
				
				The aforementioned works are yet focused on the centralized scenario. 
				For the best of our knowledge, this affine approximation was analyzed in the Distributed CSIT setting for first time  in\cite{Bazco2019_ISIT}. 
				In that work, it was shown that the $2\times 2$ Network MISO with Distributed CSIT and quantized feedback achieves the same rate offset as the ideal centralized scenario of \cite{Jindal2006}, in which the best CSIT is available at both TXs. Therefore, both settings asymptotically achieve the same rate. This result strengthens the previous results on DoF because it shows that there is no fundamental gap between both centralized and distributed CSIT settings, which could not be inferred from DoF analysis. 
				Besides that, it was also shown that the loss of performance at practical SNR values can be dramatically reduced with the correct design of the transmission schemes. 				
				
			\subsection{Main Contributions}\label{se:contributions}
				Motivated by the result of\cite{Bazco2019_ISIT} for the $2\times 2$ setting, we aim to characterize the high-SNR achievable throughput of ZF precoding techniques \cmtblu{for the $M\times K$ Network MISO setting with Distributed CSIT, and we compare it with the ideal centralized CSIT setting. This comparison is only meaningful for the antenna configurations that allow the distributed setting to achieve the same DoF as its centralized counterpart, since, otherwise, the rate gap between both settings will become unboundedly big as the SNR increases. Therefore, we restrict ourselves to such regime.} Our main contributions write as follows:
					\itb
						\item We obtain the affine expansion of the rate achieved with ZF precoding on the  $M\times K$ Network MISO setting with Distributed CSIT  
						\cmtblu{for the antenna configurations in which the centralized DoF can be achieved in the Distributed CSIT setting. }
						We prove that it is possible to asymptotically reach the same rate as in the ideal centralized setting in which the best estimate available in the network is shared by all the TXs. 
						
						This result is especially relevant when the number of transmit antennas ($N_t$) is considerably greater than the number of RXs ($K$), such as in distributed massive MIMO settings\cite{truong2013viability}, since we can achieve the beamforming gain of the centralized CSIT setting even if we only have precise CSIT  at $K$ transmit antennas. 

						\item \cmtblu{We present a discussion on a dilemma that is intrinsic to distributed settings and does not exist in the conventional centralized case: the interplay between local precision and global consistency.} In particular, we present an achievable scheme that achieves the previous result by means of capitalizing on the compromise between precoder precision and consistency among transmitters. We demonstrate that decreasing the precision of the precoding at certain transmitters improves the average performance as it helps to enhance the predictability of the transmission from the other TXs. 						
					\ite
				In addition to that, the techniques and approaches employed for the design of the transmission scheme are believed to be worthwhile by themselves for general decentralized problems, since they deal with the interplay between local precision and global consistency, which is an inherent aspect of decentralized and Team Decision problems. 
			
			\paragraph*{Notations} 
					$\Nb_N$ stands for $\Nb_N \triangleq \{1,2,\dots,N\}$. 
					We follow the asymptotic notation presented in\cite{Knuth1976}, based on the prevalent Bachmann–Landau notation\cite{Hardy2008}. 
					In consequence, 
							$f(x) = o(g(x))$ implies that $\lim_{x\rightarrow\infty}\!\frac{\!f(x)\!}{g(x)}=0$, 
							$f(x) = \Oc(g(x))$ implies than $\limsup_{x\rightarrow\infty}\!\frac{|f(x)|}{g(x)}<\infty$, 
							and 
							$f(x) = \Theta(g(x))$ implies than $\lim_{x\rightarrow\infty}\!\frac{|f(x)|}{g(x)}=c$, $0<c<\infty$.  
					For any expected value $E$ and event $X$,  $E_{|X}$ denotes  the conditional expectation given $X$. 
					$\Pr(X)$ stands for the probability of the event $X$ and $X^\setcomp$ represents the complementary event to $X$.  
					$\bA_{i,k}$ or $(\bA)_{i,k}$ represent the element of the matrix $A$ located in the $i$-th row and the $k$-th column. 
					$\bI_n$ writes for the identity matrix of size $n\times n$. 
					$\U_{M\times N}$ (resp. $\Z_{M\times N}$) represents the all-ones (resp. all-zeros) matrix of size $M\times N$. 
					$\norm{\bA}$ denotes the Frobenius norm of the matrix $\bA$. 

		\section{Problem Formulation}\label{se:sysmod}
	
			\subsection{Transmission and System Model}\label{se:sys_mod}
										
			We consider the Network MISO setting in which $M$ multi-antenna TXs jointly serve $K$ single-antenna RXs.  
			We suppose that TX~$j$ has $N_j$ antennas, and we denote the total number of transmit antennas as $N_T = \sum_{j=1}^M N_j$. 
			The received signal is given by
				\eqm{
						\yv \triangleq \sqrt{P}\bH\bW\sv + \nv,
				}
			where~$P$ is the transmit power,~$\yv \triangleq [\yn_1,\dots, \yn_K]^\Trans$ is the received signal vector, and $\yn_i$ is the received signal at RX~$i$. $\nv$ stands for the Additive White Gaussian Noise (AWGN) distributed as~$\CN(0,1)$, where $\CN(0,\Gamma)$ stands for the circularly-symmetric complex normal distribution with covariance matrix (or variance) $\Gamma$. The vector $\sv \triangleq [\sn_1,\dots, \sn_K]^\Trans$ is the vector of independent and identically distributed (i.i.d.) information symbols, where $s_i$ is the message to RX~$i$. The vector of information symbols   satisfies $\Exp[{\bs\bs^{H}}] = \bI$. 
			The channel matrix is given by 
				\eqm{
						 \bH \triangleq 
										\begin{bmatrix}
												\hv_1 \\
												\vdots \\
												\hv_K 
										\end{bmatrix}
								 \triangleq 
										\begin{bmatrix}
												\hv_{1,1} & \dots  & \hv_{1,M} \\
												\vdots 				& \ddots & \vdots \\
												\hv_{K,1} & \dots  & \hv_{K,M} 
										\end{bmatrix}	
								\in \Cb^{K\times N_T}. 
				}
			Hence, $\hv_i\in \Cb^{1\times N_T}$ denotes the global channel vector towards RX~$i$ 
			 and $\hv_{i,j}\in \Cb^{1\times N_j}$ is the channel vector from TX~$j$ to RX~$i$. 
			Note that we have defined the row vectors as $\hv_i$ and $\hv_{i,j}$ in place of the usual Hermitian notation $\hv^\He_i$ and $\hv^\He_{i,j}$. 
			This is done to ease the notation for the remainder of the document.  
			The channel coefficients are assumed to be i.i.d. as $\CN(0,1)$  
			such that all the channel sub-matrices are full rank with probability one. 
			The precoding matrix is given by
				\eqm{
						 \bW \triangleq 
										\begin{bmatrix}
												\bT_1 \\
												\vdots \\
												\bT_M 
										\end{bmatrix}
								 &\triangleq 
										\mu
										\begin{bmatrix}
												\wv_1 & \dots & \wv_K 
										\end{bmatrix} \\ 
								 &\triangleq 
										\mu
										\begin{bmatrix}
												\wv_{1,1} & \dots 	& \wv_{K,1} \\
												\vdots 		& \ddots  & \vdots \\
												\wv_{1,M} & \dots 	&\wv_{K,M} 
										\end{bmatrix}	
								\in \Cb^{N_T\times K}. \label{eq:setting_def}
				}
			Hence,  $\bT_j\in \Cb^{N_j\times K} $ represents the precoding matrix applied at TX~$j$, $\mu\wv_i\in \Cb^{N_T\times 1} $ denotes the global precoding vector for the information symbol of RX~$i$ ($s_i$), and $\mu\wv_{i,j}\in \Cb^{N_j\times 1} $ is the precoding vector applied at TX~$j$ for~$s_i$. 
			We further define $\wv_{i,j,n}$ as the coefficient at the $n$-th antenna of TX~$j$. 
			The parameter $0<\mu\leq 1$ is a power correction value that will be detailed afterwards. 
			We define $\bT_{j,n}\in \Cb^{1\times K}$ as the precoding vector applied at the $n$-th antenna of TX~$j$ for the vector $\bs$, with $n\in\Nb_{N_j}$. 
			We assume that the precoder has a per-antenna instantaneous unit-norm constraint,  such that 
				\eqm{
						\norm{\bT_{j,n}} \leq 1. \label{eq:power_contraint}
				}
			The results presented here also hold under the assumption of per-TX instantaneous constraint ($\norm{\bT_{j}} \leq 1$). 
			Note that, even if we set $\norm{\bT_{j}} = 1$, the transmit power varies over the time as the power of the information symbols~$s_i$ varies. 
			For sake of concision, we refer hereinafter to the unit-norm constraint of~\eqref{eq:power_contraint} as \emph{instantaneous power constraint}, although strictly speaking it is an \emph{instantaneous power constraint on the precoding vector} averaged over the information symbols. 
			
			\cmtblu{This is done in opposition to the average power constraint on the precoder ($\ExpB{\norm{\bT_{j}}^2} \leq 1$) which has been assumed in other works \cite{Bazco2018_WCL,Bazco2018_TIT}. This average power constraint, which results in a power normalization factor of $\nicefrac{1}{\sqrt{\Exp[{\norm{\wv_i}^2}]}}$ for each $\wv_i$, it is known to be less harmful for distributed CSIT settings. This fact follows because, under the instantaneous power constraint, the precoding must be carefully designed to avoid exceeding the power budget at every time, and this power normalization must be locally done at each TX  without an exact knowledge of the actions of the other TXs. On the other hand, under the average power constraint the normalization is based on  statistics and it is constant. 
			}
			
			\cmtblu{
			Thus, we consider the instantaneous power constraint because it is the most challenging case for a distributed setting where each TX must independently compute and agree on the applied power normalization. As it will be shown in the following, the constraint on the instantaneous power is one of the key aspects that could diminish the performance in the distributed setting. 
			}
			\subsection{Network MISO Setting with Distributed CSIT}\label{se:distr_mod}
		
		  The Distributed CSIT (D-CSIT) model is characterized by the consideration that each TX is endowed with a possibly different estimate. The key particularity of this setting is that, for any channel coefficient, there exist as many estimates as TXs, each of them \emph{locally available at a single TX}. 
			By way of example, we can think of a practical scenario in which there are some TXs with more precise knowledge for some channel coefficients and some other TXs with more precise information about other part of the channel matrix, while some sort of coarse information could be shared between them. 
			
			Although the assumption of Distributed CSIT may seem contradictory with the assumption of perfect sharing of user data symbols, which is inherent to the Network MISO setting,  both aspects co-exist in many scenarios of interest. 			
			More specifically, this model is inspired by the different timescales of latency that information data and CSI may experience in a range of emerging applications. 
			Indeed, CSI sharing is constrained by the channel coherence time (which can be very short in high-mobility scenarios) and by the possibly latency-limited backhaul connection; 
			on the other hand, many data applications have delivery time restrictions which are orders of magnitude weaker, such that the data can be pre-fetched or cached at the TXs and ready to be synchronously transmitted. 
			We refer to\cite{Bazco2018_TIT} for a detailed discussion and motivation on the  joint transmission with distributed CSIT and the practical scenarios in which it arises.

				In this work, we assume that a limited cooperation between TXs occurred before the transmission phase, leading to a certain CSIT precision configuration. 
				Hence, we assume hereinafter that the average CSIT precision at each TX remains fixed for a period of time that is long enough to be considered constant in our analysis. 
				\cmtblu{The problem of studying the best strategy of CSIT sharing with constrained links and/or delays is a very interesting research problem; yet, it is also a complex problem which would require another work on its own, and thus we do not discuss the exact CSIT acquisition and sharing mechanism. }

				\cmtblu{We focus on a particular CSIT configuration where each TX knows the whole channel matrix  with the same precision. 
				This setting is denoted as the \emph{Sorted CSIT setting}, as the TXs can be sorted by level of average precision. 
				Let us denote  the estimate  of the channel matrix $\bH$ available at TX~$j$  as $\bHH^{(j)}\in\Cb^{K\times N_T}$. Then, we model the D-CSIT configuration such that the estimate $\bHH^{(j)}$ is given by
					\eqm{\label{eq:dcsit_model0}
						\bHH^{(j)}	 \triangleq \sqrt{1 - Z^{{(j)}}}\ \bH	 + \sqrt{Z^{(j)}}\ \bDelta^{(j)},
					}
				where 
				$\bDelta^{(j)} \triangleq [\bdelta^{(j)}_1, \,\dots,\,\bdelta^{(j)}_K]^\Trans$ 
				is a random matrix that contains the additive estimation noise and whose elements are  i.i.d. as $\CN(0,1)$. The matrix $\bDelta^{(j)}$ is independent of $\bH$.
				$Z^{{(j)}}$ is a deterministic value that represents the variance scaling of the estimation noise with respect to the SNR. 
				Importantly, $Z^{{(j)}}$ characterizes the average precision of the estimate at TX~$j$, which is assumed to remain constant.  
				
				From~\eqref{eq:dcsit_model0}, we can see that the variance scaling $Z^{(j)}$ at a given TX is the same for all the channel coefficients. 
				This model encloses e.g. a scenario in which a main multi- or massive-antenna base station serves a set of users with the help of some single or multi-antenna remote radio-head or simple TXs, as depicted in Fig.~\ref{fig:scenario_case_2}. 
				
						We would like to remark that this configuration is the only one known to achieve the same DoF as the centralized setting where the most precise CSIT is shared. The analysis can be however extended to the case in which the CSIT precision order is preserved but the channel towards each user is known with different precision at a particular TX. 
						As we will see later, the precoding vectors are independently computed for each RX, and only the normalization parameter depends on all the user's vectors.\footnote{Moreover, it was shown in~\cite{Bazco2019_ISIT} that this normalization parameter can be calculated with a less restricting precision without affecting the asymptotic performance.} Because of this last point, we can restrict ourselves to the case in which the precision of each TX is characterized by a single parameter $Z^{{(j)}}$ for the sake of readability and concision,   as the extension to the setting where each RX channel vector is obtained with different precision at the same TX (and the order of TXs is preserved) would follow in a similar way.	}
						
						\cmtblu{Furthermore, this Sorted CSIT setting can model other CSIT configurations. For example, a setting where each TX obtains precise CSIT from a subset of the network (e.g., the full channel vector of some users that are associated with that specific TX, or the direct channel from all users to that specific TX), which is the most usual case in practice. Then, we could apply \emph{perfect} CSIT sharing only unidirectionally towards the main TX, whereas a coarse global CSIT estimate could be broadcast back to all the TXs from this main TX. This configuration could be enforced by tight delay constraints, and it has been already considered in practical analysis of cell-free Massive MIMO and \emph{radio stripes} systems\cite{interdonato2019ubiquitous,shaik2020mmse, Miretti2021_cellfree}. 
				Then, the only difference w.r.t. our Sorted CSIT setting would be that each TX has also the best estimate for the locally acquired CSIT, but this would provide no \emph{high-SNR} gain because the performance will be limited by the other TXs.   
				}

				\begin{remark}
					In contrast to the prior work in \cite{Bazco2019_ISIT}  for the $2\times 2$ setting, in which a distributed adaptation of the Random Vector Quantization (RVQ) feedback model of \cite{Jindal2006} was supposed, we assume in this work an additive Gaussian model for the estimation noise.  				
					This modification allows us to enlarge the contribution and verify that the previous result is not dependent on the feedback model. 
					\cmtblu{Indeed, the asymptotic results are expected to hold for a broad family of estimation models, since the key parameter that characterizes the asymptotic analysis is the scaling of the variance of the CSIT error $Z^{(j)}$, which is assumed to be constant, and not in the unit-variance random matrix~$\bDelta^{(j)}$.}
				\end{remark}

					\begin{figure}[t]\centering%
							\includegraphics[width=0.47\textwidth]{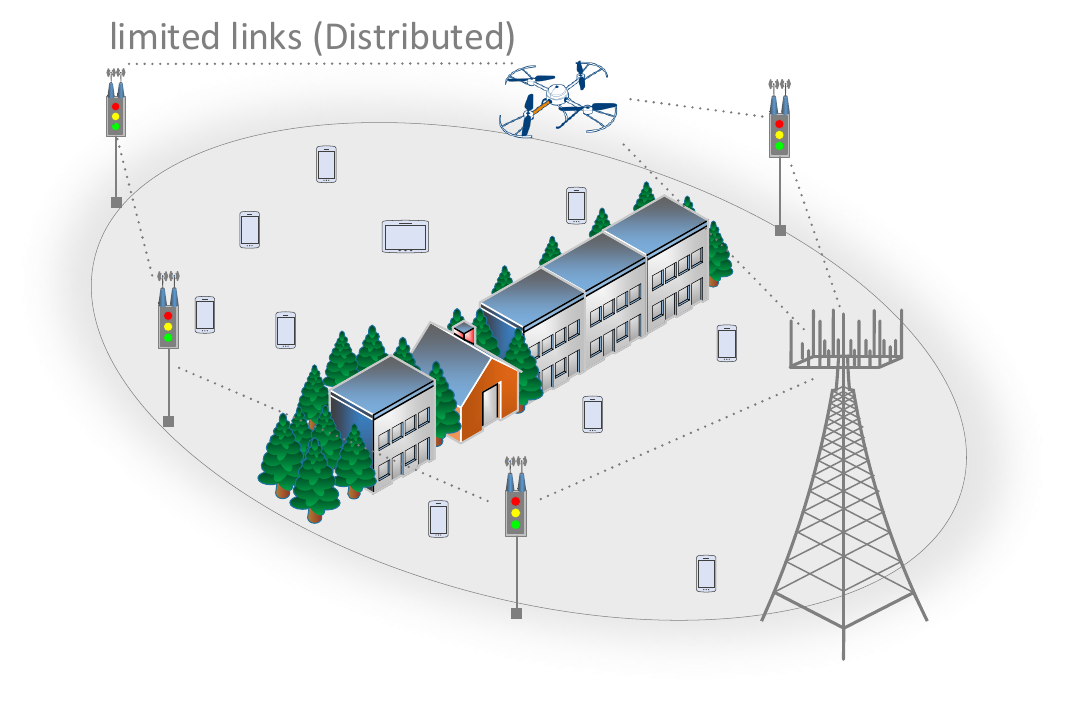}
							\caption{Master Base Station with  remote radio-heads. The Base Station obtains an estimate of the whole channel matrix, then it transmits noisy or compressed CSI to the auxiliary TXs.} 
							\label{fig:scenario_case_2}	%
					\end{figure}
		
			It is known that, for the centralized CSIT case in which all the transmit antennas share the same CSI, the variance of the estimation noise should scale as $P^{-\alpha}$ (provided that the channel estimate has unit variance) in order to obtain a multiplexing gain of $\alpha$ \cite{Jindal2006,Davoodi2018_TIT_BC}.
			As our focus is on the high-SNR regime, we consider a high-SNR modeling in which 
			the estimation error scales such that
				\eqm{\label{prop_exp_z}
						Z\expj=\Pb^{-\alpha\expj}, 
				}
			where $\Pb\triangleq\sqrt{P}$ and $0\leq\alpha\expj\leq 1$. 
			$\alpha\expj$ is the \emph{precision scaling} parameter that measures the average quality of estimation of the channel matrix at TX~$j$. 
			Hence, we can order the TXs w.l.o.g. such that
				\eqm{
						1\geq\alpha\expo>\alpha\expt\geq\dots\geq\alpha\expM\geq 0, \label{eq:order_alphas}
				}
			which implies that TX~$1$ is the \emph{best-informed} TX; in other words, the TX whose CSIT has the highest average precision. 
			\cmtblu{The reason for which the inequality $\alpha\expo>\alpha\expt$ is strict will be explained in Section~\ref{se:results_gaussians}}. 
			We define the set of precision parameters of the Distributed CSIT setting as  
				\eqm{
					\balpha_{M} = \{\alpha\expj\}_{j\in\Nb_M}. \label{eq:balpha_def_gen}
				}		
			For further use, we define the estimate at TX~$j$ for the channel of RX~$i$ as
				\eqm{
						 \bhh\expj_i\triangleq \zop\expj\hv_i + z\expj\ \bdelta\expj_i, \label{eq:dcsit}
				}		
			where $z\expj \triangleq \Pb^{-\alpha\expj}$, $\zop\expj \triangleq\sqrt{1 - (z\expj)^2}$, and $\bhh\expj_i$, $ \bdelta\expj_i$, are the $i$-th rows of the matrices $\bHH\expj$, $\bDelta\expj$, respectively. 
			As stated before, the precision parameters $\balpha_M$ are assumed to be long-term coefficients that vary slowly. Based on that, it is assumed that every TX knows the full set $\balpha_{M}$, as it only requires a sharing of few bits over a long period of time.

			\subsection{Ideal Centralized Setting}\label{se:centr_ideal}
			Finding purely distributed upper-bounds is a challenging subject that remains open, although some first results have been shown in\cite{Bazco2019_thesis}.
			However, any decentralized scenario with distributed estimates has an \emph{ideal} centralized counterpart in which a genie provides the best estimate of each parameter to every node. 
			Based on that, we define an ideal centralized scenario as follows.
				\begin{definition}[Ideal Centralized Setting] Consider the distributed setting introduced in Section~\ref{se:distr_mod}. 
				The \emph{Ideal Centralized Setting} is defined as the setting in which all the TXs are endowed with the estimate of best average precision for every channel coefficient. 
				\end{definition}
			Hereinafter, we denote the centralized channel estimate as $\bHH$ and the estimate for the channel vector of RX~$i$ as~$\bhh_i$. 
			We further denote the CSIT precision for the ideal centralized case as $\alpha^\star$. 
			Note that in the sorted setting, where $\alpha\expo>\dots>\alpha\expM$, the ideal centralized setting is equivalent to a MISO Broadcast Channel with $N_T$ transmit antennas, CSIT $\bHH$ equal to $\bHH\expo$, and $\alpha^\star = \alpha\expo$. 
			
			We will compare the rate achieved in the distributed scenario described in the previous section with the respective ideal centralized MISO BC counterpart.
			This provides us with a benchmark for the performance of ZF schemes in the Distributed CSIT setting.  
			In this way, we are able to analyze which is the impact of having distributed information or, in other words, \emph{the cost of not sharing the best CSI}. 
				\begin{remark}
					It is important to observe that the ideal centralized scenario is such that every TX owns the best estimate among the available estimates at any TX, \emph{instead of} its own estimate which by definition would have less precision.  
					This is in opposition to another genie-aided scenario, also assumed in the literature\cite{Bazco2018_WCL,Bazco2018_TIT}, in which each TX shares its CSIT with any other TX, such that every TX owns the set of M estimates of the M TXs. 
					The former model, here assumed, permits to compare the distributed scenario with the centralized counterpart. 
					The later, although it was shown in\cite{Bazco2018_TIT} that does not attain a greater DoF, would benefit from the fact that  the knowledge of $M$ estimates allows to reduce the noise variance by a factor proportional to $M$. \qed
				\end{remark}

			\subsection{Figure of Merit}\label{se:Fig_merit}
			Our figure of merit is \cmtblu{the ergodic sum rate, i.e., the expected sum rate over the joint process of fading realizations and estimates $\{\bH, \bHH\expo,\dotsc,\bHH\expM\}$\cite{Caire2007_found}. For the sake of readability, we omit hereinafter the reference to the joint process $\{\bH, \bHH\expo,\dotsc,\bHH\expM\}$ when denoting the expectation $\ExpHC{\cdot}$.} 
						
			Let us define the expected rate of RX~$i$ as $R_i \triangleq \ExpHC{r_i}$, where $r_i$ is the instantaneous rate of RX~$i$. In our $K$-user setting, $r_i$ writes as
				\begin{equation}
						 r_i\triangleq\log_2\bigg( 1+\frac{\PfracK|\hv_i^{\He}\tv_i|^2}{1+\sum_{j\neq i}\PfracK|\hv_i^{\He}\tv_{j}|^2}\bigg), 			\label{eq:fig_of_merit}
				\end{equation} 
			where $\tv_i$ denotes the precoder vector for the symbols of RX~$i$. Then, the expected sum rate is given by $R\triangleq \sum_{i=1}^K R_i$.
			Specifically, we study the linear approximation presented in~\eqref{eq:lin_approx} for the rate of the Distributed CSIT setting. 
			Let us denote the rate achieved in a Distributed CSIT setting characterized by $\balpha_{M}$ as $R(\balpha_{M})$. Hence, we aim to find the values $\DoF_d$, $\Rc^d_\infty$, such that
				\eqm{
						R(\balpha_{M}) = \DoF_d\log_2(P) - \Rc^d_\infty + o(1). 
				}

		\section{Zero-Forcing Precoding}\label{se:zf_def_exp}

			\subsection{Centralized Zero-Forcing Schemes (with ideal CSIT sharing)}\label{subse:scheme_centr}
			We restrict in this work to a  general type of ZF precoders that we rigorously characterize in the following.  
			First, in order to distinguish when the precoding vectors refer to the ideal centralized CSIT setting or the D-CSIT one, we denote the centralized precoding coefficients as $\vv_{i,k}$;  
			note that the counterpart vector for the D-CSIT setting in~\eqref{eq:setting_def} is denoted by~$\wv_{i,k}$. 
			In addition, $\bV$ and $\vv_i$ 
			are defined as the centralized counterpart of $\bW$ and $\wv_i$, respectively.  
			Hence, the  vectors $\vv_i$ are computed from any ZF precoding algorithm satisfying
				\eqm{
						1) ~\ & \bhh^{}_i\vv_{\ell} \ =\  0, \quad \forall \ell\neq i \quad\text{(Zero-Forcing condition)} \label{eq:cond_1} \tag{ZF1}\\
						2) ~\ & \E\left[{\norm{\vv_{i,j,n}}^{-1}}\right] \ =  \Theta\LB 1\RB\quad\qquad\ 
								\label{eq:cond_2} \tag{ZF2}\\
						3) ~\ & f_{\norm{\vv_{i}}} \leq f^{\max}_{\norm{\vv_{i}}} < \infty  \label{eq:cond_3} \tag{ZF3} 
				}
			Note that~\eqref{eq:cond_1} is nothing but the condition that defines ZF schemes,~\eqref{eq:cond_2} implies that the probability of precoding with a vanishing power is negligible, and~\eqref{eq:cond_3} that the precoding vector has a bounded probability density function (pdf), i.e., that it is neither predetermined nor constant. \cmtblu{This last property follows from the assumption that $\hv^{}_i$ and $\bhh^{}_i$ are drawn from Gaussian distributions and that the ZF precoding is a continuous map from estimates $\bhh^{}_i$ to vectors $\vv_{\ell}$.} 
			Hereinafter, we assume that the centralized precoding scheme satisfies \eqref{eq:cond_1}, \eqref{eq:cond_2}, \eqref{eq:cond_3}. 
			Furthermore, we assume that the precoding vectors and matrices can be expressed as a combination of summations, products, and generalized inverses\footnote{\cmtblu{The \emph{generalized inverse} of an $M\times N$~matrix $\bX$ is any matrix $\bX^{-}$ satisfying that $\bX\bX^{-}\bX = \bX$ \cite{Wiesel2008,rao1972generalized}. For any full-rank matrix $\bX$, the generalized inverse is any matrix $\bX^{-}$ that satisfies the fact that $\bX^{-}\bX = \bI$.}} of the channel estimates. 
			As an example, we can use the typical choice of the projection of the matched filters onto the null spaces of the interfered users, i.e.,
				\eqm{
						\vv_i = 
						\frac{\orthP{\iop}\bhh^\He_i}{\norm{\orthP{\iop}\bhh^\He_i}}
				}
			where $\orthP{\iop}$ is defined as 
				\eqm{
						\orthP{\iop} \triangleq \LB \bI_{N_T} - \bHH^\He_\iop(\bHH_\iop\bHH^\He_\iop)^{-1}\bHH_\iop \RB,\label{eq:centr_prec}  				
				}
			and where the matrix  $\bHH_\iop$ 
			stands for the global channel matrix with the $i$-th row removed. 
			Note that the inversion in~\eqref{eq:centr_prec} can be regularized in order to avoid degenerate cases and increase the performance at low SNR. 
			However, as conventional regularized schemes converge to their non-regularized counterpart at high SNR, we omit any reference to regularized inverses. 
			We further model the precoding scheme as a function of the CSIT, such that $\Vc$ denotes the function applied to the channel estimate:
				\eqm{
						\Vc : \Cb^{K\times N_T}\rightarrow \Cb^{N_T\times K}\qquad\text{and}\qquad \bV = \Vc\big(\bHH\big).  
				}

			\subsection{Zero-Forcing on Distributed CSIT Settings}\label{subse:centr_2_distr}
			\cmtblu{When we consider a distributed CSIT setting, we face a new challenge that does not exist in the centralized CSIT counterpart. This challenge is whether the TXs should use their CSIT even if they are among the TXs with worst CSIT precision, and to which extend they should use it. 
			Hereinafter, we will use the term ``\emph{naive} precoder'' at a certain transmitter to refer to any precoding strategy that considers that  such transmitter (naively) assumes that all the other TXs have obtained exactly its same CSIT, as if the TXs were in a centralized CSIT setting. 			
			Let us mention two simple and intuitive ZF approaches that can be applied in the distributed CSIT setting and their shortcomings.} 
				\enb
					\item	\cmtblu{\emph{Naive ZF precoder}: In this strategy, each TX assumes that all the TXs have obtained its same (imperfect) CSIT, i.e., the TXs consider that they are in a centralized CSIT setting. 			
					This approach could prove efficient at low-to-intermediate SNR for the considered CSIT error model, because the exponential relation of the CSIT error with~$P$ (cf.~\eqref{prop_exp_z})  implies that the noise variance at most transmitters will not be significantly different at such SNR ranges.
					Yet, this strategy fails at high SNR.  
					Indeed, it is known that the high-SNR performance of centralized ZF schemes collapses when they are naively applied on Distributed CSIT settings\cite{dekerret2012_TIT,Bazco2018_TIT}.  
					The main reason is that the interference cancellation achieved through~\eqref{eq:cond_1} is proportional to the worst CSIT precision among the TXs, $\alpha\expM$ in the sorted case. 
			
					Thus, the question is how to prevent the least precise TXs from harming the transmission. 
					Intuitively, the TXs should not act naively, but rather they should take into account the action of the other TXs. 
			} 
					\item \emph{Active-Passive ZF (AP-ZF):} The second intuitive idea is that the TXs whose low CSIT precision is harming the transmission should not exploit their instantaneous CSI for precoding. 
					Instead, they could transmit with a fixed or known precoder based on statistical information. 
					This solution, coined \emph{Active-Passive ZF (AP-ZF)}, achieves the centralized DoF under the less restrictive average power constraint $\ExpB{\norm{\bT_j}^2} \leq 1$ for the $2\times 2$ setting\cite{Bazco2018_WCL} and for some regimes of the $K\times K$ setting\cite{Bazco2018_TIT}. 
					\cmtblu{Under the instantaneous power constraint ${\norm{\bT_{j,n}}^2} \leq 1$ that we consider in this work, this approach also achieves the centralized DoF, but at the expense of a power back-off that allows the most precise TXs to have enough power to realign with high-enough probability the interference generated by the TXs that apply a \emph{fixed} precoder\cite{Bazco2019_ISIT,dekerret2012_TIT}. 
					The required power back-off scales as $\log_2(P)$, and thus it does not vanish at high SNR. Consequently,  it incurs  an important rate penalty that prevents the distributed scheme from achieving centralized performance.
					This case is a illustrative example of the limitation of DoF analysis, as it does not take into account such power back-off (because the penalty is $o(\log_2(P)$), and thus the DoF analysis is oblivious to the loss of performance that the power back-off creates. 
					}
				\ene
			\cmtblu{These two approaches are the opposed and extreme solutions for the question of whether the TXs should use their instantaneous CSIT: While in the first case the TXs naively consider their CSIT to be perfect, in the second case these TXs completely disregard their CSIT.} 
			In this work, we bridge the gap between the two strategies, which will prove instrumental to achieve our main results. 
			
			We present in the following several definitions that help to emphasize the two main limitations of distributed precoding. 
			
				\begin{definition}[Consistency]\label{def:consistency}
					Consider two TXs, each endowed with a different CSI. 
					Suppose that the precoder of TX~$1$ depends on the decision at TX~$2$ such that we can write $\bT_1$ as $\bT_1 = f(\bHH\expo,\ g(\bHH\expt))$, where $f,\, g$ are generic functions. 
					In a D-CSIT setting, TX~$1$ has no access to $\bHH\expt$. 
					Let us assume that TX~$1$ estimates $g(\bHH\expt)$ from $\bHH\expo$ as $\hat{g}(\bHH\expo)$. 
					Then, the computation is said to be \emph{Consistent} if and only if $\hat{g}(\bHH\expo)=g(\bHH\expt)$.
					Otherwise, it is said to be \emph{Inconsistent}.
				\end{definition}			
				\begin{definition}[Power Outage]\label{def:powerOutage}
					\cmtblu{Let $\bT_j$ denote the linear precoding matrix computed at TX~$j$. 
					For any precoder $\bT_j$, a power outage occurs if the instantaneous power constraint is violated.}	 
				\end{definition}		
			In other words, the concept of \emph{consistency} in Definition~\ref{def:consistency} stresses the fact that the TXs should precode in a coherent manner with respect to the actions of the other TXs. 
			In turn, the concept of \emph{Power Outage} reflects that even if a certain TX acquires a perfect knowledge of the precoder applied at all the other TXs, it may not be able to reduce the interference if the power required for that action overpasses the instantaneous power constraint. 
			Throughout this document, we will be interested in the cases in which the precoding is \emph{both consistent} and \emph{feasible}:
				\begin{definition}[Feasible Consistency]\label{def:feasible_consistency}
					Several TXs apply a \emph{Feasible Consistent} precoder if the precoding coefficients are \emph{Consistent} and there is no \emph{Power Outage}. 
				\end{definition}			
			One of our contributions is to show that these limitations can be overcome by encouraging \emph{consistency} among the different TXs, at the cost of reducing the precision of precoding at some~TXs.

		\section{Main Results}\label{se:results_gaussians}%
		Our main contributions rely on a novel ZF-type precoding scheme coined \emph{Consistent Decentralized ZF} (CD-ZF), which is presented in detail in Section~\ref{se:scheme_distr}. 
		In short, 
		this scheme is an adaptation to distributed scenarios of the aforementioned centralized ZF precoding, such that the precoding applied at each TX is different if the TX is the best informed one or not. 
		Let $R^{}(\balpha_M)$ be the expected sum rate for our D-CSIT setting. 
		Similarly, let $R^\star(\alpha^{(1)})$ be the expected sum rate achieved by a Zero-Forcing scheme on the ideal centralized CSIT setting as described in Section~\ref{subse:scheme_centr}. 					
			Accordingly, the rate gap between those settings is defined as   
				\eqm{
						{\DeltaR} \triangleq {R^{\star}(\alpha^{(1)})}-	{R^{}(\balpha_M)}.	\label{eq:def_deltaR}
				}
			We can now state our main result.
				\begin{theorem}[]\label{theo:rate_convergence} In the Network MISO setting with Distributed CSIT, with $N_{1}\geq K-1$ and $\alpha\expM>0$, the expected sum rate achieved by ZF-type schemes in the ideal centralized CSIT setting is asymptotically achieved, i.e.,  
						\eqm{
								\limpf {R^{\star}(\alpha^{(1)})}-	{R^{}(\balpha_M)} = 0.  \label{eq:theo_eq}
						}
				\end{theorem}
				\begin{proof} 
						The proof is relegated to Section~\ref{se:proof_theo}, and it builds on the proposed CD-ZF precoding scheme, which is presented in Section~\ref{se:scheme_distr}. 
				\end{proof}		
					\begin{corollary}[Rate-Offset under Distributed CSIT]
							It holds from Theorem~\ref{theo:rate_convergence} that the rate offset $\Rc^d_\infty\!$ (defined in~\eqref{eq:lin_approx}) of ZF with Distributed CSIT is the same as for the ideal centralized setting, whose rate offset was shown in~\cite{Jindal2006} to be constant for the case of $\alpha^\star=1$ with respect to Perfect CSIT ZF---and thus with respect to the capacity-achieving Dirty Paper Coding. 
					\end{corollary}			
				Remarkably, Theorem~\ref{theo:rate_convergence} implies\footnote{\cmtblu{We assume (see~\eqref{eq:order_alphas}) that the inequality $\alpha\expo>\alpha\expt$ is strict. This follows because, if $\alpha\expo=\alpha\expt$, since both channel realizations and noise are Gaussian, we can write $\bHH\expt$ as a noisy version of $\bHH\expo$, with a noise variance $\sqrt{2}P^{-\alpha\expo}$. Hence, the error generated when zero-forcing  the interference will have $\sqrt{2}$ times more variance than in the centralized case, which precludes the convergence presented in Theorem~\ref{theo:rate_convergence}.}  } %
				that it is possible to achieve not only the multiplexing gain but also the beamforming gain achieved by the ideal $N_T$-antenna MISO BC, even if only $N_1$ antennas are endowed with the maximum precision. 									
				\cmtblu{We would like to remark that the constraint $N_{1}\geq K-1$, i.e., that the TX with the most precise CSI has a number of antennas at least equal to the number of interfered RXs, comes from the fact that, if $N_{1} < K-1$, the use of only ZF is not enough to achieve the DoF of the centralized setting\cite{Bazco2018_TIT}, and thus $\limpf {R^{\star}(\alpha^{(1)})}-	{R^{}(\balpha_M)} = \infty$. 
				Despite that, it was shown in\cite{Bazco2018_TIT} that for certain regimes of the parameters $\alpha\expj$ it is possible to reach the ideal centralized DoF. 
				This is achieved by means of a transmission scheme which comprises interference quantization and retransmission, superposition coding at the TXs, and successive decoding at the RXs. 
				Since in this work we focus on a simple ZF transmission, we restrict the analysis to the DoF-achieving regime $N_{1} \geq K-1$, because it is the only regime in which the DoF (and thus the rate gap) of the centralized CSIT setting can be achieved. For the other cases, the distributed CSIT setting does not achieve the centralized DoF, and thus the analysis here considered is not meaningful.} 
				
				It is known that the optimal DoF of the centralized CSIT setting with precision $\alpha^{(1)}$ is $1 + (K-1)\alpha^{(1)}$\cite{Davoodi2018_TIT_BC}, which is attained by means of superposition coding where a common message is broadcast and intended to be decoded before the zero-forced messages. 
				It is remarkable that, in the regime of interest, $N_{1} \geq K-1$, the distributed CSIT setting performance still converges asymptotically to the centralized performance even if superposition coding is applied. 
				This comes from the fact that the instantaneous power applied converges to the one used in the centralized setting (as we prove at a later stage), such that the common symbol broadcast 
				can be sent with the same rate.

			\subsection{Achievability: Some Insights} \label{se:scheme_distr_intro}
			Theorem~\ref{theo:rate_convergence} evidences that the issues associated with \emph{feasible consistency} between the TXs (which are enunciated in Section~\ref{subse:centr_2_distr}) can be overcome. 
			Intuitively, the strategies mentioned in Section~\ref{subse:centr_2_distr} are extreme cases of consistency.
			Particularly, the Naive ZF represents the extreme in which consistency is not considered, whereas the Active-Passive ZF embodies the extreme with perfect consistency but limited precision and possible Power Outage, such that there may not be \emph{feasible} consistency.  
			The block diagrams of Naive ZF and Active-Passive ZF are shown in Fig.~\ref{fig:naive_zf_diagram} and Fig.~\ref{fig:apzf_zf_diagram}, respectively. 
			The main question is whether a good compromise can be found. 
			
			We can build on the idea first presented in~\cite{Bazco2019_ISIT} for the simple $2\times 2$ setting, i.e., that discretizing the decision space of the TXs helps to enforce consistency.  
			Yet, the application of this idea is  not straightforward, as in~\cite{Bazco2019_ISIT} the only source of inconsistency was a single scalar power parameter, and no beamforming was possible. 
			Specifically, the strategy in~\cite{Bazco2019_ISIT} was to design the ZF precoder such that the TX having worse channel estimate for a certain user (i.e., smaller $\alpha\expj$) precodes with a  single real value, and thus it does not consider the possible phase of the coefficient. 
			Then, the other TX (the \emph{best-informed} one) fully controls the phase tuning necessary to cancel out the interference, such that  
			both TXs need to agree only on a single value of transmit power---which must yet be agreed upon with high precision. 
			In the general $M\times K$ setting here considered, however, beamforming and phase precoding must be applied at  every TX equipped with more than one antenna, which rules out the aforementioned strategy. 
			
			Nevertheless, the main insight is still valid: By means of discretizing the decision space of the TXs that do not have the best CSI, we construct a probabilistic hierarchical setting in which the best informed TX is able to estimate correctly the action taken by the other TXs with a certain probability. 
			Interestingly, this discretization can be applied to either the available information (the channel matrix) or the output parameters (the precoding vector). 
			Both cases are illustrated in Fig.~\ref{fig:quantized_zf}.  
			This flexibility is due to the asymptotic nature of our analysis and the properties of linear systems \cmtblu{with respect to error propagation}.\footnote{\cmtblu{These properties will be detailed in Appendix~\ref{app:lemma_precoder_tx2}. There, it is shown that the error scaling of the output of a system that applies additions, products, and matrix pseudo-inverses is the same as the error scaling of the input.}}  
			It is however clear that the performance at low-to-medium SNR can importantly differ  for each of the cases. 
			In this document, we focus on the scheme that quantizes the channel matrix for the sake of a better understanding, as the proof is less devious, and because the proof for the other case (quantizing the precoder) follows the same approach and steps. 
			
			\cmtblu{Actually, we will see that this quantization is crucial for the results. Indeed, if TX~1 attempts to estimate the information used at the other TXs without any quantization, the estimation error will scale as the error variance of the other TX. Thus, TX~1 would be unable to compensate and cancel out the interference generated by the other TXs. }
			
			The key for attaining the surprising result of~Theorem~\ref{theo:rate_convergence} is the proposed precoding scheme, whose rigorous description  is presented in the following section. 
			The proof of Theorem~\ref{theo:rate_convergence} relies on a simple idea: Let $A$ be a set enclosing the \emph{feasible consistent} cases in which the precoders transmit coordinately, and $A^\setcomp$ its complementary event. 
			Hence, the rate gap ${\DeltaR} \triangleq {R^{\star}(\alpha^{(1)})}-	{R^{}(\balpha_M)}$ can be expressed as
				\eqm{
						{\DeltaR} = \DeltaR_A \Pr(A) + \DeltaR_{A^\setcomp} \Pr(A^\setcomp).
				}
			The transmission scheme has to be both feasible consistent ($\Pr(A^\setcomp)\rightarrow 0$) and, for the consistent cases, it has to be precise ($\DeltaR_A \rightarrow 0$). 
			Withal, it turns out that these two conditions are not enough to obtain Theorem~\ref{theo:rate_convergence}. 
			In particular, we need not only that the scheme is feasible consistent ($\Pr(A^\setcomp)\rightarrow 0$) but also that it attains consistency \emph{fast enough} with respect to the CSI scaling, i.e., that $\DeltaR_{A^\setcomp} \Pr(A^\setcomp)\rightarrow 0$. 
			
			As a matter of example, consider that the rate gap for the inconsistent cases ($\DeltaR_{A^\setcomp}$) scales proportionally to $\log(P)$ (as it will proven later on) such that $\DeltaR_{A^\setcomp} = \Theta(\log(P))$. 
			Then, if the probability of inconsistency $\Pr(A^\setcomp)$ approached to $0$ at a convergence rate proportional to $\frac{1}{\log(P)}$ ($\Pr(A^\setcomp)=\Theta(\frac{1}{\log(P)})$), the term $\DeltaR_{A^\setcomp} \Pr(A^\setcomp)$ would not vanish and Theorem~\ref{theo:rate_convergence} would not hold. 
			Fortunately, we can design the achievable scheme so as to reduce the probability of cases without consistency \emph{and} without power outage faster than $\Theta(\frac{1}{\log(P)})$, \emph{while providing a precise precoding and vanishing gap $\DeltaR_A$ in the consistent cases}. 
			This last condition is critical: As $P$ increases, the precision of the precoder in the consistent cases must increase faster than $\Oc(P)$. Otherwise, $\DeltaR_A$ would not vanish. 
			The key to attaining  sufficient precision for the consistent cases while reducing at the same time the probability of inconsistent cases is that they depend on different scaling of $P$ ($\Theta(\log(P))$ versus $\Theta(P)$). 
			This is rigorously shown in Section~\ref{se:proof_theo}. 
				
				\begin{figure}[t]\centering%
					\begin{subfigure}[b]{\columnwidth} 
						    \begin{tikzpicture}[>=latex',scale=.93, every node/.style={transform shape}]
        \tikzset{
					block/.style= {draw, rectangle, align=center,minimum width=2cm,minimum height=1cm},
        }
				
				
        \node   (tx2) {TX 2 :};
        \node [right = .25cm of tx2]  (h2) {$\bHH\expt$};

        \node [block, right = 0.5cm of h2]  (prec2) {						
					\begin{tabular}{c} 
								ZF\\[-0.05ex] 
							\!\!\!	Precoding \!\!\!
						\end{tabular} 
				};			

        \node [above right=-4ex and 2ex of prec2]  (out_zf2) {$\wv_1^{{{\textbf{(2)}}}}$};
        \node [below right=-4ex and 2ex of prec2]  (out_pw2) {$\wv_2^{{{\textbf{(2)}}}}$};
				
        \node [block, right=9ex of prec2, minimum width = .5cm]  (in_pw2) {						
					\begin{tabular}{c} 
								Power\\[-.05ex]
							\!\!\!	Control \!\!\!
						\end{tabular} 
				};				

        \node [right = .5cm of in_pw2]  (out2) {$\bT\expt_{2}$};

				
        \node [above = 6ex of tx2]  (tx1) {TX 1 :};
        \node [right = .25cm of tx1]  (h1) {$\bHH\expo$};

        \node [block, right = 0.5cm of h1, minimum width = .5cm]  (prec1) {						
					\begin{tabular}{c} 
								ZF\\[-.05ex]
							\!\!\!	Precoding \!\!\!
						\end{tabular} 
				};	
				
        \node [above right=-4ex and 2ex of prec1]  (out_zf1) {$\wv_1^{{{\textbf{(1)}}}}$};
        \node [below right=-4ex and 2ex of prec1]  (out_pw1) {$\wv_2^{{{\textbf{(1)}}}}$};
				
        \node [block, right=9ex of prec1, minimum width = .5cm]  (in_pw1) {						
					\begin{tabular}{c} 
								Power\\[-0.05ex]
							\!\!\!	Control \!\!\!
						\end{tabular} 
				};			

        \node [right = .5cm of in_pw1]  (out1) {$\bT\expo_{1}$};
								
	

	%


        \path[draw, ->]
            (h1) edge (prec1)
            (prec1) edge (out_zf1)
            (prec1) edge (out_pw1)

            (h2) edge (prec2)
            (prec2) edge (out_zf2)
            (prec2) edge (out_pw2)

						(out_zf2) edge (in_pw2)
            (out_pw2) edge (in_pw2)		
						
						(out_zf1) edge (in_pw1)
            (out_pw1) edge (in_pw1)														

            (in_pw1) edge (out1)														
            (in_pw2) edge (out2)														
                    ;
										
    \end{tikzpicture}\vspace{-2ex}
						\caption{Block diagram of conventional ZF applied naively in the $2\times2$ Distributed CSIT scenario (No Consistency).}\vspace{2ex}
						\label{fig:naive_zf_diagram}
					\end{subfigure}
					\hspace{1ex} 
					\begin{subfigure}[b]{0.47\textwidth} 
						    \begin{tikzpicture}[>=latex',scale=0.92, every node/.style={transform shape}]
        \tikzset{
					block/.style= {draw, rectangle, align=center,minimum width=2cm,minimum height=1cm},
        }
				
				
        \node   (tx2) {TX 2 :};
        \node [above right =-1.5ex and .25cm of tx2]  (h2) {$\bHH\expt$};
        \node [below right =-1.5ex and .25cm of tx2]  (t2) {$\alpha\expo, \alpha\expt$};

        \node [right = 3.3ex of h2] (prec2){\!\!\!{\color{red}{\textbf{x}}}};
				
        \node [block, right=22ex of tx2, minimum width = .05cm]  (in_pw2) {						
					\begin{tabular}{c} 
							\!\!\!\!\!	Power 	Control \!\!\! \!\!\! \!\!\! \!\!\! \\[-0.05ex]
							\!\!\!\!\!\!\!	(Constant) \!\!\!		\!\!\!		\!\!\!					

						\end{tabular} 
				};				
				
        \node [coordinate,  below left = -2.0ex and 0cm of in_pw2] (a2pw2){};

        \node [right = 3.3ex of in_pw2]  (out2) {$\bT_{2}$};

				
        \node [above = 6ex of tx2]  (tx1) {TX 1 :};
        \node [right = .25cm of tx1]  (h1) {$\bHH\expo$};

        \node [block, right = 3.3ex of h1]  (prec1) {						
					\begin{tabular}{c} 
								ZF\\[-0.05ex] 
							\!\!\!	Precoding \!\!\!
						\end{tabular} 
				};
				
        \node [above right=-4ex and 2ex of prec1]  (out_zf1) {$\wv_1^{{{\textbf{(1)}}}}$};
        \node [below right=-4ex and 2ex of prec1]  (out_pw1) {$\wv_2^{{{\textbf{(1)}}}}$};
				
        \node [block, right=9ex of prec1, minimum width = .25cm]  (in_pw1) {						
					\begin{tabular}{c} 
								Power\\[-0.05ex]
							\!\!\!	Control \!\!\!
						\end{tabular} 
				};				

        \node [right = 3.3ex of in_pw1]  (out1) {$\bT\expo_{1}$};
								
	

				



				
        \path[draw, ->]
            (h1) edge (prec1)
            (prec1) edge (out_zf1)
            (prec1) edge (out_pw1)

            (h2) edge (prec2)
						
            (t2) edge (a2pw2)		
						
						(out_zf1) edge (in_pw1)
            (out_pw1) edge (in_pw1)														

            (in_pw1) edge (out1)														
            (in_pw2) edge (out2)														
                    ;
										
    \end{tikzpicture}\vspace{-2ex}
						\caption{Block diagram of AP-ZF applied in the $2\times2$ Distributed CSIT scenario (Full Consistency but possible infeasibility).}
						\label{fig:apzf_zf_diagram}							
					\end{subfigure}
					\caption{Simple strategies for distributed precoding. }
						\label{fig:diagrams}
				\end{figure}
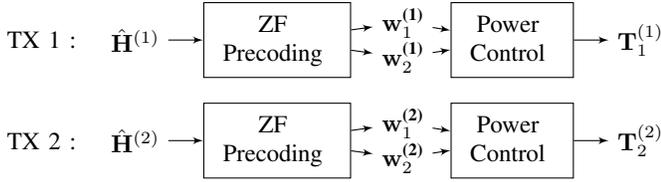
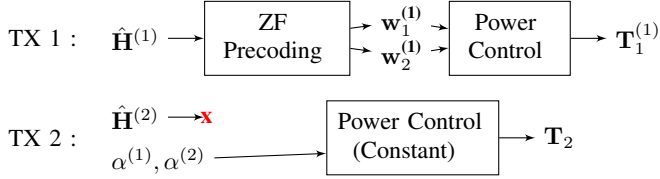
			
			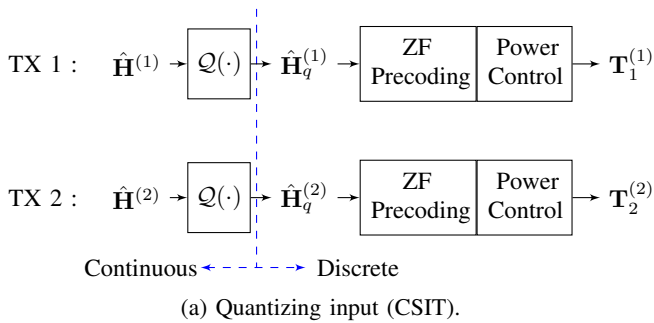
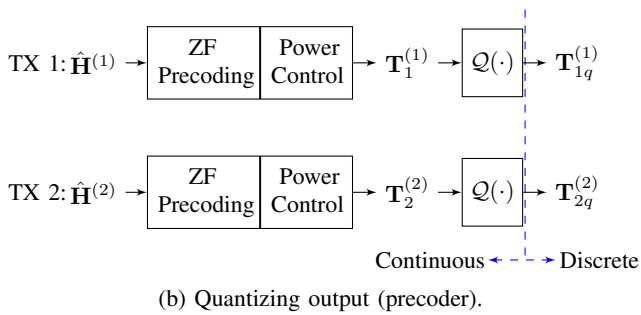
\begin{figure}[t]\centering%
				\begin{subfigure}[b]{0.47\textwidth} 
					    \begin{tikzpicture}[>=latex',scale=0.93, every node/.style={transform shape}]
        \tikzset{
					block/.style= {draw, rectangle, align=center,minimum width=2cm,minimum height=1cm},
        }
				
				
        \node   (tx2) {TX 2 :};
        \node [right =.25cm of tx2]  (h2) {$\bHH\expt$};
				
        \node [block, right = .25cm of h2, minimum width = 1ex]  (q2) {$\Qc(\cdot)$};

        \node [right = .3cm of q2]  (h2q) {$\bHH\expt_q$};

				\node [block, right = 2ex of h2q, minimum width=0.5ex] (prec2){
					\begin{tabular*}{1.3cm}{c} 
								\!\!\! ZF \\[-0.05ex] 
							\!\!\!\!\!\!\,	Precoding  \!\!\!\!\!\!
						\end{tabular*} 				
				};        
				\node [block, right = 0cm of prec2, minimum width=0.5ex] (pw2){
					\begin{tabular*}{1cm}{c} 
							  \!\!\!\!	 Power  \!\!\!\!\!\! \\[-0.05ex] 
							 \!\!\!\!\!	Control \!\!\!\!\!\! 
						\end{tabular*} 				
				};
        \node [right = 2.5ex of pw2]  (qout2) {$\bT\expt_2$};
				




				
        \node [above = 9ex of tx2]  (tx1) {TX 1 :};
        \node [right = .25cm of tx1]  (h1) {$\bHH\expo$};

        \node [block, right = .25cm of h1, minimum width = 1ex]  (q1) {$\Qc(\cdot)$};

        \node [right = .3cm of q1]  (h1q) {$\bHH\expo_q$};

        \node [block, right = 2ex of h1q, minimum width=0.5ex] (prec1){
					\begin{tabular*}{1.3cm}{c} 
								\!\!\! ZF \\[-0.05ex] 
							\!\!\!\!\!\!\,	Precoding  \!\!\!\!\!\!
						\end{tabular*} 				
				};        
				\node [block, right = 0cm of prec1, minimum width=0.5ex] (pw1){
					\begin{tabular*}{1cm}{c} 
							  \!\!\!\!	 Power  \!\!\!\!\!\! \\[-0.05ex] 
							 \!\!\!\!\!	Control \!\!\!\!\!\! 
						\end{tabular*} 				
				};
        \node [right = 2.5ex of pw1]  (qout1) {$\bT\expo_1$};

				

								

												
	

	\draw [dashed, blue] (3.05,2.7) -- (3.05,-1);

	\draw [dashed, blue, <->] (2.25,-1) -- (3.75,-1);
		
	\draw (1.4,-1) node {Continuous};
	\draw (4.5,-1) node {Discrete};


        \path[draw, ->]
            (h1) edge (q1) 	
            (q1) edge (h1q)
            (h1q) edge (prec1)
            (pw1) edge (qout1)
						
            (h2) edge (q2) 	
            (q2) edge (h2q)
            (h2q) edge (prec2)
            (pw2) edge (qout2);
												

										
    \end{tikzpicture}~~~~~\vspace{-3ex}
					\caption{Quantizing input (CSIT).} 
					\label{fig:diagram_input}
				\end{subfigure}\vspace{3ex}
				\hspace{1ex} 
				\begin{subfigure}[b]{0.47\textwidth} 
					    \begin{tikzpicture}[>=latex',scale=0.9, every node/.style={transform shape}]
        \tikzset{
					block/.style= {draw, rectangle, align=center,minimum width=2cm,minimum height=1cm},
        }
				
				
        \node   (tx2) {TX 2:};
        \node [right =-.2cm of tx2]  (h2) {$\bHH\expt$};				

				\node [block, right = 2ex of h2, minimum width=0.5ex] (prec2){
					\begin{tabular*}{1.3cm}{c} 
								\!\!\! ZF \\[-0.05ex] 
							\!\!\!\!\!\!\,	Precoding  \!\!\!\!\!\!
						\end{tabular*} 				
				};        
				\node [block, right = 0cm of prec2, minimum width=0.5ex] (pw2){
					\begin{tabular*}{1cm}{c} 
							  \!\!\!\!	 Power  \!\!\!\!\!\! \\[-0.05ex] 
							 \!\!\!\!\!	Control \!\!\!\!\!\! 
						\end{tabular*} 				
				};
        \node [right = .35cm of pw2]  (qout2) {$\bT\expt_2$};
				
        \node [block, right = .35cm of qout2, minimum width = 1ex]  (q2) {$\Qc(\cdot)$};

        \node [right = .35cm of q2]  (h2q) {$\bT\expt_{2q}$};





				
        \node [above = 9ex of tx2]  (tx1) {TX 1:};
        \node [right = -.2cm of tx1]  (h1) {$\bHH\expo$};

        \node [block, right = 2ex of h1, minimum width=0.5ex] (prec1){
					\begin{tabular*}{1.3cm}{c} 
								\!\!\! ZF \\[-0.05ex] 
							\!\!\!\!\!\!\,	Precoding  \!\!\!\!\!\!
						\end{tabular*} 				
				};        
				\node [block, right = 0cm of prec1, minimum width=0.5ex] (pw1){
					\begin{tabular*}{1cm}{c} 
							  \!\!\!\!	 Power  \!\!\!\!\!\! \\[-0.05ex] 
							 \!\!\!\!\!	Control \!\!\!\!\!\! 
						\end{tabular*} 				
				};
				
        \node [right = .35cm of pw1]  (qout1) {$\bT\expo_1$};

        \node [block, right = .35cm of qout1, minimum width = 1ex]  (q1) {$\Qc(\cdot)$};

        \node [right = .35cm of q1]  (h1q) {$\bT\expo_{1q}$};

				

								

												
	

	\draw [dashed, blue] (7.2,2.7) -- (7.2,-1);

	\draw [dashed, blue, <->] (6.65,-1) -- (7.65,-1);
		
	\draw (5.8,-1) node {Continuous};
	\draw (8.3,-1) node {Discrete};
			

        \path[draw, ->]
            (h1) edge (prec1) 	
            (pw1) edge (qout1)
            (qout1) edge (q1)
            (q1) edge (h1q)
						
            (h2) edge (prec2) 	
            (pw2) edge (qout2)
            (qout2) edge (q2)
            (q2) edge (h2q);
												

										
    \end{tikzpicture}~~~~~\vspace{-3ex}
					\caption{Quantizing output (precoder). } 
					\label{fig:diagram_output}							
				\end{subfigure}
				\caption{Two manners of discretizing decision space: At the input (information) or at the output (action).} 
				\label{fig:quantized_zf}
			\end{figure}

			\subsection{Proposed Transmission Scheme: Consistent Distributed ZF}
			\label{se:scheme_distr}
			We present in the following the Consistent Distributed ZF (CD-ZF) precoding scheme, where the CSIT of the TXs that do not have access to the most precise estimate is quantized to improve the consistency of the scheme.   
			Later, we analyze the feasibility of the proposed precoder. 
			The CD-ZF scheme presents an uneven structure, such that each TX applies a different strategy depending on who has higher average  precision. 
			Furthermore, the proposed scheme independently computes the precoder for the symbols of different RXs, except for the final power normalization.
			We recall that the precoding vector for the message intended at RX~$i$ is defined in~\eqref{eq:setting_def} as $\mu\wv_i\in \Cb^{N_T\times 1}$, and the segment of  $\mu\wv_i$ applied at TX~$j$ is given by $\mu\wv_{i,j}\in \Cb^{N_j\times 1} $.

				\subsubsection{Quantizing the CSIT}\label{subsubse:scheme_channel}%
			
				The block diagram of this precoding scheme is depicted in Fig.~\ref{fig:map_diagram}. 
				Because of the uneven structure of the precoder, we separately describe  the precoding strategy at the best-informed TX (TX 1) and at any  other TX (TX~$2$ to TX~$M$). 
				Let us consider first the later. 						
				The main limitation of the distributed precoding is not the error variance at the restricting TXs, but the impossibility at TX~1 
				of knowing what the other TXs are going to transmit. 
				In order to overcome this problem, all the TXs but TX~1 quantize their estimate of the channel matrix with a known quantizer~$\Qc$.  Hence, for any $j>1$, TX~$j$ does not directly use its CSIT $\bHH\expj$ to precode, but first pre-processes it. In other words, TX~$j$ applies	
					\eqm{
							\bHH\expj_q = \Qc(\bHH\expj).
					}
				The characteristics of the quantizer $\Qc$ will be detailed later. 
				Then, TX~$j$ naively applies a centralized ZF scheme as described in Section~\ref{subse:scheme_centr} but in a distributed manner (based on $\bHH\expj_q$).  
				Since the quantization transforms the continuous variable $\bHH\expj$ into a discrete one, it facilitates that the setting becomes a hierarchical setting in which the information available at other TXs is estimated without explicit communication. 
				
			\begin{figure}[t]\centering
					    \begin{tikzpicture}[>=latex',scale=.78, every node/.style={transform shape}]
        \tikzset{
					block/.style= {draw, rectangle, align=center,minimum width=2cm,minimum height=1cm},
        }

				
        \node   (tx2) {TX 2 :};
        \node [above = 1.2cm of tx2]  (tx1) {TX 1 :};


				\node [right = .2cm of tx1]  (h1) {$\bHH\expo$};								
				
        \node [coordinate, right = 0.25cm of h1] (h1tozfa){};	
        \node [coordinate, above = 0.4cm of h1tozfa] (h1tozfb2){};	
        \node [coordinate, below = 0.4cm of h1tozfa] (h1tozfb1){};	
								
        \node [block, right= 3ex of h1tozfb1, minimum width=0cm, minimum height=0cm]  (map1) {$\text{MAP}(\bHH\expt_q)$};

				\node [right = 3ex of map1]  (h1map) {$\bHH^{(2)\leftarrow(1)}_q$};								
				
        \node [block, right = 5.2cm of h1, minimum height=1.2cm] (prec1){ZF Precoding\\[-3ex] ~};

        \node [coordinate, above left = -0.2cm and 0cm of prec1] (h1tozfc2){};	
        \node [coordinate, below left = -0.2cm and 0cm of prec1] (h1tozfc1){};	
								
        \node [right=2ex of prec1]  (out_zf1) {${{\bT\expo_1}}$};
				
				
        \node [right = .2cm of tx2]  (h2) {$\bHH\expt$};

        \node [block, right= 8ex of h2, minimum width=0cm, minimum height=0cm]  (q2) {$\Qc(\cdot)$};
        \node [right= 6.5ex of q2]  (out_q2) {$\bHH\expt_q$};
				
        \node [block, right = 6.7ex of out_q2, minimum height=1.2cm] (prec2){ZF Precoding};

        \node [right=2ex of prec2]  (out_zf2) {${{\bT\expt_2}}$};				
	

	\draw [dashed, blue, <->] (5.2,-1) -- (4.65,-1) --(4.65,1.7) -- (6.8, 1.7) -- (8, .85)  -- (10.5,.85) -- (10.5,.3);

	\draw [dashed, blue, ->] (4.65,-1) -- (4,-1);
		
	\draw (3,-1) node {Continuous};
	\draw (5.9,-1) node {Discrete};


        \path[draw, ->]
            (h1) -- (h1tozfa) -- (h1tozfb1) -- (map1);
        \path[draw, ->]
            (h1tozfa) -- (h1tozfb2) -- (h1tozfc2);

        \path[draw, ->]							
            (h1map) edge (h1tozfc1)
            (map1) edge (h1map)
						
            (h2) edge (q2)						
            (q2) edge (out_q2)
            (out_q2) edge (prec2)
														
            (prec1) edge (out_zf1)
						
            (prec2) edge (out_zf2);

    \end{tikzpicture}
					\caption{Block diagram of Consistent Distributed ZF applied in the 2x2 Distributed CSIT scenario}%
					\label{fig:map_diagram}%
			\end{figure}
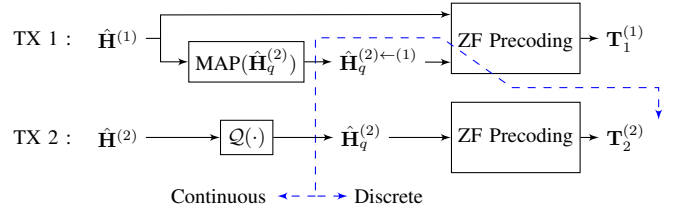	
								
				We focus now on the  precoder at the best-informed TX (TX~1), which attempts to correct the error of the previous TXs. 
				Let TX~$1$ estimate $\bHH\expj_q$ based on its own information $\bHH\expo$, e.g. by computing the Maximum A Posteriori estimator (MAP) of $\bHH\expj_q$: 
					\eqm{
							\bHH^{(j)\leftarrow(1)}_q = \argmax_{\bHH\expj_q\in\Qc(\Cb^{K\times N_T})} \Pr\LB\bHH\expj_q \mid \bHH\expo\RB.
					}
				It is important to notice that the quantized value $\bHH\expj_q$ is not intended to be transmitted, but it is aimed at helping TX~1 to estimate the CSIT used at TX~$j$, \emph{without any explicit communication between them}.
								
				\subsubsection{Canceling the interference}\label{subsubse:scheme_channel2b}
				The CD-ZF scheme naively assumes  that  TX~1 correctly estimates the CSIT at TX~$j$, for any $j$ in $\Nb_M$, such that  
					\eqm{
							\bHH^{(j)\leftarrow(1)}_q = \bHH\expj_q. \label{eq:prob_agree1}
					}
				The probability that~\eqref{eq:prob_agree1} holds true will be evaluated in the following section.  
				Therefore, the fact that~\eqref{eq:prob_agree1} is fulfilled implies that the setting is \emph{consistent}---although it does not guarantee feasibility, which will be ensured by the parameter $\mu$. 
				The goal of TX~$1$ is to imitate the interference cancellation performance that the centralized ZF scheme would achieve if every other TX also had access to $\bHH^{(1)}$.\footnote{We could develop more convoluted schemes in order to  increase the achievable rate at low or medium SNR regimes. 
				As a matter of example, we could analyze the minimum mean-square error (MMSE) precoding as in~\cite{Miretti2021_cellfree}, or design several layers of quantization such that, for every~$j$, TX~$(j-1)$ tries to correct the interference generated by TX~$j$, in a similar manner to the algorithm presented in~\cite{dekerret2016_asilomar}. Nevertheless, in this work we focus on the asymptotic regime, and thus we maintain the scheme in its simplest expression.} 
				In order to provide some insights, we first describe the single-antenna case with 2 TXs and 2 RXs.
				
					\paragraph{$2\times 2$ case}\label{par:distributed_scheme_tx1_2tx}					
						Let $\iop \triangleq i\pmod{2} + 1$. Mathematically, the goal is to have $\absn{\smash{\bhh\expo_{i}\wv_{\iop}}} = \absn{\smash{\bhh\expo_{i}\vv_{\iop}}}$, which can be rewritten as
							\eqm{
								\absn{\smash{\bhh\expo_{i,1}\wv_{\iop,1} + \bhh\expo_{i,2}\wv_{\iop,2}}} = \absn{\smash{\bhh\expo_{i,1}\vv_{\iop,1} + \bhh\expo_{i,2}\vv_{\iop,2}}}. \label{eq:def_prec_distr_1}
							}
						We remind the reader that $\wv$ stands for the distributed precoder whereas $\vv$ stands for the centralized precoder.  Under the assumption that TX~$1$ correctly estimates $\bHH\expt_q$, it knows $\wv_{\iop,2}$. Then, TX~1 computes its precoder such that
							\eqm{
								\wv_{\iop,1}  
										& =  \vv_{\iop,1} + \underbrace{(\bhh\expo_{i,1})^\dagger\bhh\expo_{i,2}(\vv_{\iop,2}-\wv_{\iop,2})}_{\bphi_{\iop}}, \label{eq:eq_1}
							}
						where $(\xv)^\dagger$ denotes the pseudo-inverse\footnote{$(\xv)^\dagger$ could also represent the regularized pseudo-inverse.} of $\xv$, which is known to have minimal Frobenius norm among all the generalized inverses \cite{Wiesel2008}. The term $\bphi_{i}$ represents the correction term that TX 1 has to apply in order to compensate the error introduced by TX 2; note that~\eqref{eq:eq_1} satisfies~\eqref{eq:def_prec_distr_1}.

					\paragraph{$M\times K$ case}\label{par:distributed_scheme_tx1_Mtx} 
					The generalization from the $2\times 2$ case to the general $M\times K$ case  needs one more step. 
					Let TX~1 have $N_1\geq K-1$ antennas. 
					The goal is again to obtain the same interference cancellation as for the centralized precoder, such that, for any RX~$i\in \Nb_K$,  
						\eqm{
							\sum_{\ell\in\Nb_K\backslash i} \absn{\smash{\bhh\expo_{i}\wv_{\ell}}}^2 = \sum_{\ell\in\Nb_K\backslash i}\absn{\smash{\bhh\expo_{i}\vv_{\ell}}}^2.\label{eq:firt_eq_intr_can_ew}
						}
					The equality in~\eqref{eq:firt_eq_intr_can_ew} is attained if (but not only if) $\bhh\expo_{i}\wv_{\ell} = \bhh\expo_{i}\vv_{\ell}$, for any $i,\ell\in\Nb_K$ such that $i\neq\ell$. 
					Let us split the precoding and channel vectors in two parts: $\vv_{\ell,{1}}$, $\wv_{\ell,{1}}$ and $\bhh\expo_{i,{1}}$ denote the sub-vector corresponding to the antennas of TX~1, whereas $\vv_{\ell,\bar{1}}$, $\wv_{\ell,\bar{1}}$ and $\bhh\expo_{i,\bar{1}}$ represent the  sub-vector corresponding to the antennas from TX~2 to TX~$M$.
					The sub-matrices $\bHH_{\ellop,1}$ and $\bHH_{\ellop,\bar{1}}$ are defined in the same manner, and they are illustrated in Fig.~\ref{fig:matrix_division_h} for the sake of better comprehension. 
					Note that we have introduced the notation  $N_{\bar{1}} \triangleq  N_T - N_1$. 
					We can expand the condition $\bhh\expo_{i}\wv_{\ell} = \bhh\expo_{i}\vv_{\ell}$ as a matrix equation in which $\wv_{\ell,1}$ has to satisfy 
						\eqm{
							\bHH_{\ellop,1}\wv_{\ell,1} 
									= \bHH_{\ellop,1}\vv_{\ell,1} 
											+ \bHH_{\ellop,\bar{1}}(\vv_{\ell,\bar{1}} - \wv_{\ell,\bar{1}}).
							\label{eq:matrix_exp}																			
						} 				
					The precoding vector at TX~$1$ is then selected as
						\eqm{
							 \wv_{\ell,1}  = \vv_{\ell,1} + \underbrace{\bHH_{\ellop,1}^\dagger\bHH_{\ellop,\bar{1}}(\vv_{\ell,\bar{1}} - \wv_{\ell,\bar{1}})}_{\bphi_{\ell}} \label{eq:eq_1b}.
						}					 
						\begin{remark}
							\cmtblu{We can see from the dimensionality of the linear system in~\eqref{eq:matrix_exp} that we need $N_1\geq K-1$ in order to make~\eqref{eq:matrix_exp}  feasible with high-probability. This is another way of understanding why we need the restriction that $N_1\geq K-1$. 
							However,  the main reason to have this constraint comes from the limitations of the Distributed CSIT setting. Specifically, it is known that the use of ZF precoding (or other linear precodings) in the Distributed CSIT setting can only achieve the \emph{DoF} of its centralized counterpart when this constraint is satisfied~\cite{Bazco2018_TIT}, and thus $\limpf {R^{\star}(\alpha^{(1)})}-	{R^{}(\balpha_M)} = \infty$ otherwise.							
							}
						\end{remark}
					In~\eqref{eq:matrix_exp}, it is ensured that the interference received is the same as for the centralized setting. 
					It is possible also to ensure that the receive signal  $\bhh\expo_{i}\wv_{i}$ is equal to the one of the centralized setting. 
					This would add an extra equation to the linear system that would require an extra antenna at TX~1. However, it is not necessary as the received intended signal turns out to be statistically equivalent in both distributed and centralized~settings.  
					
						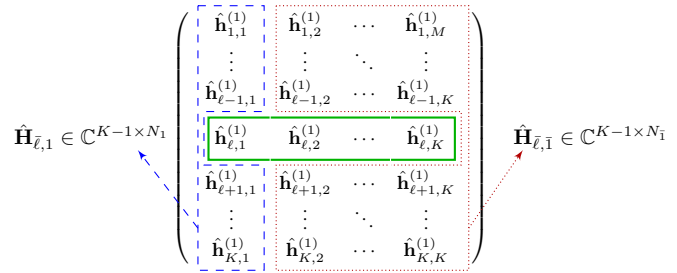
\begin{figure}[t]%
							\centering
								    \begin{tikzpicture}[>=latex',scale=0.8, every node/.style={transform shape}]
        \tikzset{
					block/.style= {draw=white, rectangle, align=center,minimum width=2cm,minimum height=1cm},
        }
				

				\node [block] {\small \renewcommand\arraystretch{0.0}
					$
						\left(\; \begin{array}{cccc}
							\bhh\expo_{1,1} & \bhh\expo_{1,2} & \cdots & \bhh\expo_{1,M}  \\[1ex]
							\vdots & \vdots & \ddots & \vdots  \\[1ex]
							\bhh\expo_{\ell-1,1} &  \bhh\expo_{\ell-1,2} & \cdots & \bhh\expo_{\ell-1,K}\\[3ex]
							\tikzmark{left3}
							\bhh\expo_{\ell,1} &  \bhh\expo_{\ell,2} &  \cdots & \bhh\expo_{\ell,K}\tikzmark{right3} \\[3ex]
							\bhh\expo_{\ell+1,1} &  \bhh\expo_{\ell+1,2} &  \cdots & \bhh\expo_{\ell+1,K}  \\[1ex]
							\vdots &  \vdots &  \ddots &  \vdots \\[1ex]
							\bhh\expo_{K,1} &  \bhh\expo_{K,2} &  \cdots &  \bhh\expo_{K,K} 
						\end{array}\; \right)
					$
				};
					
		\path[draw=blue, dashed] (-2.2,2.2) -- (-1.1,2.2)  -- (-1.1, .44) -- (-2.1, .44) -- (-2.1, -.45) -- (-1.1, -.45) -- (-1.1, -2.2)-- (-2.2, -2.2) -- cycle;
		\path[draw=red!70!black, densely dotted] (-.9, 2.2)  -- (-.9, .44) -- (2.175, .44) -- (2.175, -.45) -- (-.9, -.45) -- (-.9, -2.2) -- (2.3, -2.2) -- (2.3, 2.2) -- cycle;
								
		\node [draw=none,inner sep=0pt,minimum size=0.3cm] at (-4.0, 0) { $\,\bHH_{\bar{\ell},{1}}\in\Cb^{K-1\times N_1}$};
		\node [draw=none,inner sep=0pt,minimum size=0.3cm] at (4.3, 0) { $\,\bHH_{\bar{\ell},\bar{1}}\in\Cb^{K-1\times N_{\bar{1}}}$};
		\path[->, draw=blue, dashed] (-2.2,-1.5) -- (-3.2, -.2);
		\path[->, draw=red!70!black, densely dotted] (2.3,-1.5) -- (3.2, -.2);
										
		\tikzdrawbox[(0pt,0pt)]{3}{thick,green!70!black}
										
				
        \node [block]  (prec2) {						
					\begin{tabular}{c} 
								~\\~\\~\\~\\~\\~\\~\\
						\end{tabular} 
				};			
						
    \end{tikzpicture}\vspace{3ex}
								\caption{Illustration of the channel sub-matrices $\bHH_{\ellop,1}$ and $\bHH_{\ellop,\bar{1}}$. }\label{fig:matrix_division_h}%
						\end{figure}
								
				\subsubsection{Feasibility and consistency}\label{subsubse:distributed_scheme_tx1}
				In the previous description of the scheme, it has been assumed that TX~1 obtains a feasible consistent precoder, such that it correctly estimates the CSI at the other TXs, and that the obtained precoding vector can be used for transmission. 
				In fact, the CD-ZF scheme naively considers that the precoder is \emph{always} feasible and consistent.

				However, the transmission scheme will suffer from the two main issues described in Section~\ref{subse:centr_2_distr}: 
				Power outage, since it has to satisfy that $\norm{\bT_{1,n}} \leq 1$ for any~$n\in\Nb_{N_1}$, and \emph{Consistency}, as the quantization of the CSI at TX~$j$ allows TX~1 to obtain the quantized CSI only with a certain probability. 
				Further, 
								as explained in Section~\ref{se:scheme_distr_intro}, we also need to demonstrate that the interference cancellation is precise enough in the feasible consistent cases.  
				In the following we present some properties that will prove instrumental in dealing with those limitations. 
				For that, we introduce a set of quantizers that are essential in the proof of Theorem~\ref{theo:rate_convergence}. We first tackle the probability of inconsistency.			
					\begin{definition}[Asymptotically Consistent Quantizers] 
							A quantizer $\Qc$ is said to be  \emph{Asymptotically Consistent} if the probability of correct estimation of the MAP estimator at TX~1 satisfies
								\eqm{
										 \Pr\left(\bHH^{(j)\leftarrow(1)}_q \neq\bHH\expj_q\right)=o\LB\frac{1}{\log_2(P)} \RB\!\!\,, \quad \forall j\in\Nb_M.  \label{eq:cond_agreement} \tag{P1}
								}
					\end{definition}			
				Property~\eqref{eq:cond_agreement} implies that it is possible to induce that the probability of having \emph{inconsistent} precoding among TXs vanishes faster than $1/{\log_2(P)}$. 
				This fact implies that the rate impact of inconsistent precoding events vanishes asymptotically, as mentioned in Section~\ref{se:scheme_distr_intro}.  
				Clearly, it remains to prove that there exists some quantizer $\Qc_c$ that satisfies~\eqref{eq:cond_agreement} while also ensuring both feasibility and an adequate precision for the consistent cases. 
				Surprisingly, very simple quantizers as the scalar uniform quantizer satisfy the requirements. 
				To prove this statement, we first show that this quantizer is an Asymptotically Consistent Quantizer. 
					\begin{lemma}\label{lem:uniform_quantizer} 
							Let $\Qc_u(\bX)$ be a scalar uniform quantizer with quantization step $q = \Pb^{-\alpha_q}$, where $\alpha_q$ is such that $\alpha\expj>\alpha_q>0$, for all~$j\in\Nb_M$. 
							Then, $\Qc_u$ is an \emph{Asymptotically Consistent Quantizer} and 
								\eqm{
										\Pr\LB\bHH^{(j)\leftarrow(1)}_q  \neq \bHH\expj_q \RB = o\LB\frac{1}{\log_2(P)}\RB,\quad \forall j\in\Nb_M.
								}
					\end{lemma}
					\begin{proof}
							The proof is relegated to Appendix~\ref{app:proof_consistent}. 
					\end{proof}				
				Note that $\Qc_u$ is a scalar quantizer. 
				Thus, for a matrix $\bA$, the notation $\bA_q = \Qc_u(\bA)$  denotes (with an abuse of notation) that $\bA_q$ is composed of the independent scalar quantization of the real and imaginary part of each element in $\bA$.  
				Obviously, using vector quantization would improve the performance. 
				However, as the proof of Theorem~\ref{theo:rate_convergence} is constructive, we are interested in an intuitive example, and thus we use $\Qc_u$ for the sake of simplicity. The analysis at medium and low SNR, which requires an optimization on the applied quantizer, is an interesting research topic that is however out of the scope of this work. 
				
				We focus now on the probability of \emph{power outage}.  
				Let us denote the event of power outage as~$\pset$ and recall that the precoder at the $n$-th antenna of TX~$j$ is expressed by $\bT_{j,n}$. 
				In that case, 	
					\eqm{
							\pset\triangleq \bigcup_{\substack{n\in\Nb_{N_j}\\ j\in\Nb_M}} \big\{\norm{\bT_{j,n}} > 1 \big\}.
					}   
				The power outage probability is handled by the parameter $\mu$ of the precoding vector $\mu\wv_i\in \Cb^{N_T\times 1}$. 
				The value of $\mu$ acts as power back-off that can be tuned  to achieve the required scaling, as stated in the following lemma. 
				\begin{lemma}\label{lem:outage} 
						Let $\mu = 1 - \varpi$, where $\varpi > 0,\ \varpi = \Theta(\Pb^{-\alpha_{\mu}})$, and $\alpha_{\mu} < \alpha_q$. Then, 
							\eqm{\label{eq:lem2_p2}
									\Pr\LB\pset\RB = o\LB\frac{1}{\log_2(P)}\RB.\tag{P2}
							}
				\end{lemma}
						\begin{proof}
								See Appendix~\ref{app:proof_outage}.
						\end{proof}	
				Similarly to  property~\eqref{eq:cond_agreement}, property~\eqref{eq:lem2_p2} in Lemma~\ref{lem:outage} implies that  \emph{power outage}  events are negligible in terms of asymptotic rate. 
				The only TX that may incur in  power outage is TX~$1$, as the other TXs apply the naive centralized precoder and hence they will always satisfy the power constraint. 
				
				To this extent, we have shown that the uniform scalar quantizer $\Qc_u$ enables us to reach the requirements regarding the probability of the cases that are not \emph{feasible consistent}. Hence, it remains to prove that it also attains high enough precision in the feasible consistent cases. 
				This will be proven in Section~\ref{se:proof_theo}. 
				Hereinafter, we assume that the uniform quantizer of Lemma~\ref{lem:uniform_quantizer} is applied in the CD-ZF. 				
				\subsection{Hierarchical CSIT Setting}\label{subse:hierar}
				Theorem~\ref{theo:rate_convergence} shows that it is possible to asymptotically attain  the rate of the centralized setting. 
				Its performance at low-to-medium SNR is however limited by the probability of obtaining a \emph{feasible consistent} precoder.  
				This probability depends on the quantizer applied, the power back-off considered, and the values of $\alpha\expj$, and hence it is challenging to obtain. 
				As shown in Section~\ref{se:scheme_distr}, the precoder is computed assuming a correct estimation of the CSI at the other TXs. 
				Consequently, if the probability of \emph{consistency} is low, the scheme does not perform properly. 
				Moreover, this probability decreases as the network size increases, since TX~1 needs to correctly estimate more parameters.
				
				This limitation is inherent to the D-CSIT setting here assumed, in which each TX only knows its own CSI. 
				However, there exist another practical setting with distributed CSI but in which there is more structure in the network CSI: 
				The Hierarchical CSIT setting (H-CSIT).  
				In this setting, each TX is endowed with its own multi-user CSI $\bHH\expj$, as in the D-CSIT setting, but it is also endowed with the CSI of the TXs having less precision than itself. 
				Namely, in the sorted CSI scenario with $\alpha\expo>\dots\geq\alpha\expM$, TX~$j$ has access to $\{\bHH\expj,\bHH^{(j+1)},\dots,\bHH\expM\}$.
				
				This scenario, although it may seem less practical, may arise in many heterogeneous networks. 
				Fig.~\ref{fig:scenario_case_2} depicts an example:  
				Suppose that the RXs are all connected to the same main TX (e.g. TX~1), and the other TXs are remote radio-heads that receive a coarse version of the CSI by means of a wireless link from TX~1. 
				In this use case, TX~1 will know the CSI available at each other TX. 
				If the CSI sharing is done through dedicated links for each TX, each TX would receive CSI with precision proportional to its own link. 
				If the CSI is broadcast, they may obtain an estimate with different precision if  layered encoding \cite{Ng2009} or analog feedback \cite{ElAyach2012a} is used. 
				
				\begin{corollary}
					Theorem~\ref{theo:rate_convergence} also holds in the Hierarchical CSIT setting, and hence $\limpf {R^{\star}(\alpha^{(1)})}-	{R^{}(\balpha_M)} = 0$.
				\end{corollary}
				\begin{proof}The proof follows directly from the proof of Theorem~\ref{theo:rate_convergence} in Section~\ref{se:proof_theo}.
				\end{proof}
				
				In this setting, TX 1 already knows $\bHH\expj$ for any~$j \in \Nb_M$. 
				Hence, the discretization of the variables at the other TXs is not needed, and the precoders are \emph{consistent} with probability 1. 
				In fact, the idea of quantizing the CSI in the D-CSIT setting boils down to making the CSIT setting asymptotically hierarchical with a high enough probability. 
				The only effect that may restrain TX 1 to achieve the centralized performance is the power outage. 
				Therefore, the performance at medium SNR will improve with respect to the general Distributed CSIT case, and moreover, this performance is not affected by the size of the network, as we will see in the numerical examples of Section~\ref{se:numerical}.

			\section{Proof of Theorem~\ref{theo:rate_convergence}}\label{se:proof_theo}
				In order to prove Theorem~\ref{theo:rate_convergence}, we need to demonstrate that the user rate gap $\DeltaR_i = R^{\star}_i(\alpha^{(1)})-	R^{}_i(\balpha_M)$ vanishes, which directly yields that $\DeltaR = \sum_{i\in\Nb_K}\DeltaR_i$ will also vanish. 
				The proof is divided in several steps: First, armed with Lemma~\ref{lem:uniform_quantizer} and Lemma~\ref{lem:outage}, we show that both \emph{power outage} and \emph{inconsistent precoding} events can be made negligible in terms of rate loss. 
				Then, we prove that the rate gap also vanishes  in the feasible consistent cases thanks to the fine precision of the CD-ZF scheme in these cases. 
				We demonstrate this by showing that both the interference received and the total power received in the distributed setting converge to their counterparts of the centralized setting. 
				
				\subsection{Neglecting Non-consistent Events}\label{subsubse:neglect_issues}
		%
			The proof of Theorem~\ref{theo:rate_convergence} builds on Lemma~\ref{lem:uniform_quantizer} and Lemma~\ref{lem:outage}. 
			Indeed, the proposed scheme will perform poorly in the cases in which the precoder is not feasible consistent, as the scheme is built on the naive assumption that it is \emph{always} feasible  and consistent. 
			Nevertheless, both Lemma~\ref{lem:uniform_quantizer} and Lemma~\ref{lem:outage} illustrate that those events can be made very unlikely (in particular, the probability of these events is $o(\frac{1}{\log(P)})$).  
			Let $\Hc_{\neq}$ denote the set of \emph{inconsistent} events, i.e., $\Hc_{\neq} \triangleq  \bigcup_{2\leq j\leq M }\big\{\bHH^{(j)\leftarrow(1)}_q  \neq \bHH\expj_q \big\}$. 
			Hence, the probability of having \emph{feasible consistent} precoding is $\Pr\big(\pset^\setcomp\cap\Hc_{\neq}^\setcomp\big)$. 
			By means of the law of total expectation, we can split the expected rate gap for RX~$i$ ($\DeltaR_i$) as\footnote{Given a certain feasible event $A$ and its complementary event $A^\setcomp$, the law of total expectation states that $\Exp[X] = \Pr(A)\Exp[X| A] + \Pr(A^\setcomp)\Exp[X|A^\setcomp]$. 
			Furthermore, For any two events $A_1$, $A_2$, and union event $A=A_1\cup A_2$, it follows that $A^\setcomp = (A_1\cup A_2)^\setcomp = (A_1^\setcomp\cap A_2^\setcomp)$.}
				\eqmo{
						\DeltaR_i & = \Pr\LB\pset\cup\Hc_{\neq}\RB \DeltaR_{i\mid\pset\cup\Hc_{\neq}} \\
						 & \quad {}+ \Pr\LB\pset^\setcomp\cap\Hc_{\neq}^\setcomp\RB \DeltaR_{i\mid\pset^\setcomp\cap \Hc_{\neq}^\setcomp}. 
				}
			Let us focus on $\DeltaR_{i\mid\pset\cup\Hc_{\neq}}$. 
			The rate gap can be upper-bounded by setting the rate of the D-CSIT setting to 0, such that $\DeltaR_{i\mid\pset\cup\Hc_{\neq}} \leq R^\star_{i\mid\pset\cup\Hc_{\neq}}(\alpha\expo)$. 
			By neglecting the received interference, we can write that 
				\eqmo{
					&R^\star_{i\mid\pset\cup\Hc_{\neq}}(\alpha\expo)
						 \leq \Exp_{\mid\pset\cup\Hc_{\neq}}\left[\log_2\LB 1 + \PfracK\abs{\hv_i\vv_i}^2\RB\right]\\ 
						&\qquad\quad \leq \log_2( \PfracK) + \Exp_{\mid\pset\cup\Hc_{\neq}}\left[\log_2\LB1+\abs{\hv_i\vv_i}^2\RB\right]\!.\label{eq:bound_rate_notfeascons}
				}				
			Moreover, the set $\pset\cup\Hc_{\neq}$ depends on the different estimation noise at each TX, which is absent in the centralized setting. 
			Accordingly,~\eqref{eq:bound_rate_notfeascons} implies that $R^\star_{i\mid\pset\cup\Hc_{\neq}}(\alpha\expo) = \Theta(\log_2(P))$. 
			Hence, it follows from Lemma~\ref{lem:uniform_quantizer} and Lemma~\ref{lem:outage} that
				\eqmo{
						\Pr\LB\pset\cup\Hc_{\neq}\RB \DeltaR_{i\mid\pset\cup\Hc_{\neq}} &= o\LB\frac{1}{\log_2(P)}\RB \Theta(\log_2(P))\notag
				}
			and consequently
				\eqm{
					\DeltaR_i & = \DeltaR_{i\mid\pset^\setcomp\cap \Hc_{\neq}^\setcomp } + o(1) \label{eq:first_simplif}.
				}					
			Thus, in the remainder of the proof we assume w.l.o.g. that TX~1 knows $\bHH\expj_q$ for any~$j\in\Nb_M$, and that there is no power outage, as both cases become negligible at high SNR. 
			This assumption implies that the setting becomes hierarchical, because TX~$1$ correctly estimates the quantized CSI of the other TXs.   
			It is important to remark that this simplification is only possible because of the proposed scheme, in which we apply a correct power back-off and quantization step. 
			Indeed, the more important outcome of this work---and the main purpose of the careful design of the scheme---is not~\eqref{eq:first_simplif} but the fact that $\DeltaR_{i\mid\pset^\setcomp\cap \Hc_{\neq}^\setcomp}$ converges to the centralized setting rate. 
										
				\subsection{Reformulating the Rate Gap}\label{subsubse:reformul_rategap}
				We can rewrite the rate gap for RX~$i$ as				
					\eqmo{
							\DeltaR_i 
									& = \Exp\left[\log_2\LB 1 + \frac{\PfracK\abs{\hv_i\vv_i}^2}{1 + \PfracK\sum_{\ell\neq i}\abs{\hv_i\vv_{\ell}}^2 }\RB\right]\\
											&\quad\ - \Exp\left[\log_2\LB 1 + \frac{\PfracK\abs{\mu\hv_i\wv_i}^2}{1 + \PfracK\sum_{\ell\neq i}\abs{\mu\hv_i\wv_{\ell}}^2 }\RB\right] \\
									& = \Exp\bigg[\log_2\Big( \underbrace{\frac{1 + \PfracK \sum_{\ell\in\Nb_K\!}\abs{\hv_i\vv_{\ell}}^2 }{1 + \PfracK\sum_{\ell\in\Nb_K\!}\abs{\mu\hv_i\wv_{\ell}}^2 }}_{\Fc_\Dc}\Big)\bigg] \\
											&\quad\ + \Exp\bigg[\log_2\Big( \underbrace{\frac{1 + \PfracK \sum_{\ell\neq i}\abs{\mu\hv_i\wv_{\ell}}^2 }{1 + \PfracK\sum_{\ell\neq i}\abs{\hv_i\vv_{\ell}}^2 }}_{\Fc_\Ic}\Big)\bigg]. \label{eq:eq_proof_1}
					}
				This rewriting of $\DeltaR_i$ allows us to separate the ratio of received interference power ($\Fc_\Ic$) and the ratio of total received power ($\Fc_\Dc$). 
				In the following, we will prove that $\limpf \DeltaR_i = 0$ by showing that $\limpf\ExpB{\log_2(\Fc_i)}=0$ for both $\Fc_\Dc$ and $\Fc_\Ic$. We start with $\Fc_\Ic$ for simplicity, and later we apply a similar argument to $\Fc_\Dc$.
							
				\subsection{Analysis of the Interference Ratio ($\Fc_\Ic$)}\label{subsubse:interf_can_ratio}
				We prove the convergence by upper and lower-bounding $\Fc_\Ic$, and then showing that both bounds converge to 0. We recall that we assume that TX~$1$ is able to transmit the desired precoding vector of~\eqref{eq:eq_1} since the opposite case only yields an $o(1)$ rate contribution. 
					\subsubsection{Upper-bounding $\Exp\left[\log_2\LB \Fc_\Ic\RB\right]$}\label{subsubse:interf_can_ratioa}
					Note that, since  $\mu \leq 1$,   
						\eqmo{
							& \Exp\left[\log_2\LB \frac{1 + \PfracK \sum_{\ell\neq i}\abs{\mu\hv_i\wv_{\ell}}^2 }{1 + \PfracK\sum_{\ell\neq i}\abs{\hv_i\vv_{\ell}}^2 }\RB\right] \\
									&~\hspace{10ex} \leq \Exp\bigg[\log_2\Big( \underbrace{\frac{1 + \PfracK \sum_{\ell\neq i}\abs{\hv_i\wv_{\ell}}^2 }{1 + \PfracK\sum_{\ell\neq i}\abs{\hv_i\vv_{\ell}}^2 }}_{\Fc'_\Ic}\Big)\bigg],    
						}	
					where we have introduced the notation $\Fc'_\Ic$ for the sake of readability. 
					Let $\eta$ be a scalar satisfying $0\leq \eta \leq 1$.  
					We can split the expectation based on whether the term  $\Fc'_\Ic$ is smaller than $1+\eta$ or not. 
					Therefore, 
						\eqmo{
								\!\!\!\!\Exp\left[\log_2\LB \Fc'_\Ic\RB\right] 
										& \!= \Pr\LB\Fc'_\Ic<1+\eta\RB\Exp_{\Fc'_\Ic<1+\eta}\left[\log_2\LB \Fc'_\Ic \RB\right] \\
										& ~~ + \Pr\LB\Fc'_\Ic\geq1+\eta\RB\Exp_{\Fc'_\Ic\geq1+\eta}\left[\log_2\LB \Fc'_\Ic\RB\right]\!. \label{eq:split_expectation1}
						}					
					Now we present a useful lemma. 
						\begin{lemma}\label{lem:eta_convergence}
							Let $\eta = \Pb^{-\varepsilon}$, with $\alpha_q> \varepsilon>0$  and $\varepsilon$ arbitrarily small. Then,
								\eqm{
										\Pr\LB\Fc'_\Ic\geq 1+\eta	\RB = o\LB\frac{1}{\log_2(P)}\RB
								}
							and
								\eqm{
										\Pr\LB\frac{1}{\Fc'_\Ic} \geq 1+\eta	\RB = o\LB\frac{1}{\log_2(P)}\RB.  
								}
						\end{lemma}
						\begin{proof}
							The proof is relegated to Appendix~\ref{subse:prooflemma_eta1}.
						\end{proof}
					Let $\eta = \Pb^{-\varepsilon}$, with $\alpha_q> \varepsilon>0$ and $\varepsilon$ arbitrarily small. 
					Then,~\eqref{eq:split_expectation1} becomes
						\eqmo{
								\Exp\left[\log_2\LB \Fc'_\Ic\RB\right] 
										& \leq \Exp_{\Fc'_\Ic<1+\eta}\left[\log_2\LB \Fc'_\Ic \RB\right] \\
										& \quad\ + o\LB\frac{1}{\log_2(P)}\RB \Exp_{\Fc'_\Ic\geq1+\eta}\left[\log_2\LB \Fc'_\Ic\RB\right] \\ 
										& \leq \log_2(1+\eta) \ + \ o\LB 1\RB \label{eq:exp_bound_2}  
						}					
					since $\Exp_{\Fc'_\Ic\geq1+\eta}\left[\log_2\LB \Fc'_\Ic\RB\right]  = \Oc(\log_2(P))$. In order to prove that  $\Exp_{\Fc'_\Ic\geq1+\eta}\left[\log_2\LB \Fc'_\Ic\RB\right]  = \Oc(\log_2(P))$, note that
						\eqm{
						  &\!\!\!\Exp_{\Fc'_\Ic\geq1+\eta}\left[\log_2\LB \Fc'_\Ic\RB\right]  
								= \frac{1}{\Pr\LB\Fc'_\Ic\geq 1+\eta	\RB}\Big(\!\Exp_{}\left[\log_2\LB \Fc'_\Ic\RB\right] \Big. \notag\\
							&\hspace{8ex}\quad \Big. - \Pr\LB\Fc'_\Ic< 1+\eta	\RB\Exp_{\Fc'_\Ic<1+\eta}\left[\log_2\LB \Fc'_\Ic\RB\right]\Big) \label{eq:fffss} \\
								&\qquad\leq \frac{1}{\Pr\LB\Fc'_\Ic\geq 1+\eta	\RB}\Exp_{}\left[\log_2\LB \Fc'_\Ic\RB\right].\notag 
						}
					Furthermore, 
						\eqmo{
							\!\!\!\!\Exp_{}\left[\log_2\LB \Fc'_\Ic\RB\right] 
								& \leq \Exp\bigg[\log_2\Big( {1 + \PfracK \sum_{\ell\neq i}\abs{\hv_i\wv_{\ell}}^2 }\Big)\bigg]\\
								& \leq \log_2( \PfracK) + \Exp\!\bigg[\!\log_2\!\Big( {1 + \sum_{\ell\neq i}\abs{\hv_i\wv_{\ell}}^2 }\Big)\bigg]\!.\!\! \label{eq:fff}
						}			
					From \eqref{eq:fffss}, \eqref{eq:fff}, and the fact that $\Pr\LB\Fc'_\Ic\geq 1+\eta	\RB = \Theta(1)$, we obtain that $\Exp_{\Fc'_\Ic\geq1+\eta}\left[\log_2\LB \Fc'_\Ic\RB\right]  = \Oc(\log_2(P))$.
					
					\subsubsection{Lower-bounding $\Exp\left[\log_2\LB \Fc_\Ic\RB\right] $}\label{subsubse:interf_can_ratiob}
					Let us now lower-bound the expectation. Note that
						\eqm{
								\Exp\left[\log_2\LB \Fc_\Ic\RB\right] 
										& \geq \log_2(\mu^2) + \Exp\big[\log_2\big( \Fc'_\Ic\big)\big]. \label{eq:eq_ref_lem1}
						}	
					Furthermore, lower-bounding~\eqref{eq:eq_ref_lem1} is equivalent to upper-bounding $\Exp\Big[\log_2\Big( {\frac{1}{\Fc'_\Ic}}\Big)\Big]$.
					By applying Lemma~\ref{lem:eta_convergence} and in a similar way as in~\eqref{eq:exp_bound_2}, we obtain that $\Exp\left[\log_2\LB \tfrac{1}{\Fc'_\Ic}\RB\right] \leq \log_2(1+\eta) \ + \ o\LB 1\RB$ and hence								
						\eqm{
							\Exp\left[\log_2\LB \Fc_\Ic\RB\right] 										
									& \geq \log_2(\mu^2) - \log_2(1+\eta) \ + \ o\LB 1\RB.
						}					
					Consequently, the term $\Exp\left[\log_2\LB \Fc_\Ic\RB\right]$ can be bounded as
						\eqmo{
								& \log_2(\mu^2) - \log_2(1+\eta)  + o\LB 1\RB \\
								& \qquad\qquad\qquad \leq \Exp\left[\log_2\LB \Fc_\Ic\RB\right] 	\leq \log_2(1+\eta)  +  o\LB 1\RB.
						}
					Since $\limpf \mu = 1$ and $\limpf \eta = 0$, 
					it follows that
						\eqm{
								\limpf \Exp\left[\log_2\LB \Fc_\Ic\RB\right] = 0.
						}						
						
				\subsection{Analysis of the Received Signal Ratio ($\Fc_\Dc$)}\label{subsubse:desired_can_ratio}
				It remains to prove that the first expectation in~\eqref{eq:eq_proof_1} also converges to zero. As for $\Fc_\Ic$, we can write
					\eqmo{
						\ExpB{\log_2(\Fc_\Dc)} 
							& \leq \log_2\LB\frac{1}{\mu^2}\RB \\
							& \quad + \Exp\bigg[\log_2\Big( \underbrace{\frac{1 + \PfracK \sum_{\ell\in\Nb_K\!}\abs{\hv_i\vv_{\ell}}^2 }{1 + \PfracK\sum_{\ell\in\Nb_K\!}\abs{\hv_i\wv_{\ell}}^2  }}_{\Fc'_{\Dc}}\Big)\bigg].
					}		
				Moreover, the equivalent to Lemma~\ref{lem:eta_convergence} also holds for $\Fc'_\Dc$. 
					\begin{lemma}\label{lem:eta_convergence2}
						Let $\eta = \Pb^{-\varepsilon}$, with $\alpha_q> \varepsilon>0$  and $\varepsilon$ arbitrarily small. Then,
							\eqm{
								\Pr\LB\Fc'_\Dc\geq 1+\eta	\RB = o\LB\frac{1}{\log_2(P)}\RB
							}
						and 
							\eqm{
								\Pr\LB\frac{1}{\Fc'_\Dc} \geq 1+\eta	\RB = o\LB\frac{1}{\log_2(P)}\RB.  
							}
					\end{lemma}
					\begin{proof}
						The proof is relegated to Appendix~\ref{subse:prooflemma_eta2}.
					\end{proof}
										
				Thus, applying the same step as in~\eqref{eq:exp_bound_2} yields
					\eqm{
						\Exp\big[\log_2\big( \Fc'_{\Dc}\big)\big] 
							& \leq \log_2(1+\eta) + o(1).
					}
				We can similarly lower-bound $\ExpB{\log_2(\Fc_\Dc)}$  to obtain that												
					\eqmo{
							 & - \log_2(1+\eta)  +  o\LB 1\RB \leq \ExpB{\log_2(\Fc_\Dc)} 	\\
							 & \quad\qquad\qquad\qquad\leq \log_2\LB\nicefrac{1}{\mu^2}\RB +  \log_2(1+\eta)  +  o\LB 1\RB.
					}													
				The fact that $\limpf \mu = 1$ and $\limpf \eta = 0$ leads to 
					\eqm{
							\limpf \ExpB{\log_2(\Fc_\Dc)} = 0.
					}
					
				\subsection{Merging Previous Sections}\label{subsubse:end_proof}
				Given that $\limpf \DeltaR = \limpf \sum_{i=1}^K \DeltaR_i$, we obtain that 
					\eqmo{
							\limpf \DeltaR 
								& = \limpf K(\ExpB{\log_2(\Fc_\Dc)}  + \ExpB{\log_2(\Fc_\Ic)}) \\
								& =  0,\label{eq:eq_proof_end}
					}				
				which concludes the proof of Theorem~\ref{theo:rate_convergence}.\qed

			\section{Numerical Results}\label{se:numerical}
			In this section, we provide some performance analysis for the previous asymptotic results. 
			We consider a scenario in which the most-informed TX has a CSI precision scaling parameter $\alpha\expo = 1$ for the whole channel matrix, and the rest of TXs have a CSI precision scaling parameter $\alpha\expj = 0.6$, for any $j>1$.  
			Intuitively, this configuration can model a setting in which a main TX receives a quantized CSI feedback from all the RXs, and then it shares a compressed version of the CSI to the other auxiliary transmit antennas. 
			We present the performance of several schemes:  
				\begin{itemize}[noitemsep] \itemsep0em 
						\item  The ideal centralized CSIT setting, in which all the TXs are endowed with the CSI of TX~1. 
						\item  The CD-ZF scheme with Hierarchical CSIT (TX~1 knows the other TXs' CSI).
						\item  The CD-ZF, AP-ZF, and Naive ZF schemes when the CSIT is non hierarchical (general D-CSIT setting).
						\item  The performance of transmitting only from TX~1 and turning off the other TXs.
				\end{itemize}			
			In Fig.~\ref{fig:M2_11_K2}, we show the rate performance for a setting with 2~single-antenna TXs and 2~RXs under the assumption of instantaneous power constraint for the precoder. 
			Several insights emerge from the figure. 
			
			First, we observe how the proposed CD-ZF scheme performs almost as good as the ideal centralized CSIT setting for the Hierarchical CSIT configuration. 
			This fact holds for any setting configuration and size, yet considering that $N_1\geq K-1$. 
			Besides this, the CD-ZF scheme is shown to tend towards the centralized rate also for the general D-CSIT setting, where the CSI at TX~2 is not available at TX~1. 
			However, we can see how the convergence is slow, and at low SNR  the CD-ZF scheme outperforms the single-TX transmission or the Naive ZF only by a slight gap. 
			This is an aftermath of the scheme definition, as it is tailored for the asymptotic high-SNR regime. 
			Indeed, the CD-ZF scheme performs in an almost optimal manner if TX~1 correctly estimates the CSI at the other TXs; however, the probability of correct estimation increases slowly. 
			Thus, the performance at medium SNR is limited. 
			
			It is important to note that the CD-ZF scheme here presented is not optimized, as our objective was to show the asymptotic behavior. 
			For example, we assume a scalar quantizer that independently quantizes every real and imaginary part of each channel coefficient. 
			Considerably higher probabilities of consistency would be obtained if the quantization phase is optimized, as it can be seen in~\cite{Bazco2019_GRETSI}. We could, for example, use vector quantization or  consider more complex schemes such as the ones proposed in\cite{dekerret2016_asilomar,Atzeni2018}. 
			Nevertheless, the aforementioned points show how important it is to provide the CSI with structure (or hierarchy), as it has been proven indispensable to boost the performance. 
			Moreover, this CSI structure is sometimes given by the network configuration, such that it does not imply an extra aspect to develop. 
			
			Another point to be considered is that CD-ZF allows to obtain centralized performance with one informed antenna less than the single-TX transmission. This consideration can be seen in Fig.~\ref{fig:M2_11_K2}, as the single-TX transmission does not even achieve  the centralized DoF. 
			\begin{figure}[t]\centering
						\centering
	\definecolor{mycolor1}{rgb}{0.00000,0.44700,0.74100}%
	\definecolor{mycolor2}{rgb}{0.85000,0.32500,0.09800}%
	\definecolor{mycolor3}{rgb}{0.92900,0.69400,0.12500}%
	\definecolor{mycolor4}{rgb}{0.49400,0.18400,0.55600}%
	\definecolor{mycolor5}{rgb}{0.300000,0.64314,0.00000}
\begin{tikzpicture}
\begin{axis}[%
width=1.65*0.5\columnwidth,
height=1.65*.45\columnwidth,
at={(0.0\textwidth,0.0\textwidth)},
scale only axis,
xmin=0.000,
xmax=80.000,
xlabel style={font=\color{white!15!black}},
xlabel={P [dB]},
ymin=0.000,
ymax=25.000,
ylabel={\small Rate [bits/s/Hz]},
y label style={at={(axis description cs:0.02,.51)},rotate=0,anchor=south,font=\color{white!15!black}},
ylabel near ticks,
axis background/.style={fill=white},
axis x line*=bottom,
axis y line*=left,
xmajorgrids,
ymajorgrids,
legend style={at={(0.025,0.6)}, anchor=south west, legend cell align=left, align=left, draw=white!15!black}
]
\addplot [color=black, dashed,line width=0.75pt]
  table[row sep=crcr]{%
0.000	0.539\\
11.429	2.702\\
22.857	5.730\\
34.286	9.393\\
45.714	13.216\\
57.143	17.003\\
68.571	20.930\\
80.000	24.617\\
};
\addlegendentry{Centralized CSIT}

\addplot [color=mycolor1, mark=triangle, mark options={scale=1.2, solid, rotate=180, mycolor1},line width=0.65pt]
  table[row sep=crcr]{%
0.000	0.494\\
11.429	2.540\\
22.857	5.443\\
34.286	9.020\\
45.714	12.942\\
57.143	16.804\\
68.571	20.782\\
80.000	24.551\\
};
\addlegendentry{CD-ZF H-CSIT}

\addplot [color=mycolor1, mark=o, mark options={scale=1.2,solid, mycolor1},line width=0.65pt]
  table[row sep=crcr]{%
0.000	0.519\\
11.429	2.302\\
22.857	4.488\\
34.286	7.373\\
45.714	10.550\\
57.143	14.749\\
68.571	19.180\\
80.000	23.280\\
};
\addlegendentry{CD-ZF D-CSIT}

\addplot [color=mycolor2, mark=square, mark options={scale=1.2,solid, mycolor2},line width=0.65pt]
  table[row sep=crcr]{%
0.000	0.411\\
11.429	2.138\\
22.857	4.523\\
34.286	7.339\\
45.714	10.483\\
57.143	13.882\\
68.571	17.391\\
80.000	20.773\\
};
\addlegendentry{AP-ZF}

\addplot [color=mycolor4, dashdotted,line width=0.65pt]
  table[row sep=crcr]{%
0.000	0.520\\
11.429	2.303\\
22.857	4.489\\
34.286	7.074\\
45.714	9.581\\
57.143	11.944\\
68.571	14.517\\
80.000	16.742\\
};
\addlegendentry{Naive ZF}

\addplot [color=mycolor5, densely dotted,line width=0.75pt]
  table[row sep=crcr]{%
0.000	0.435\\
11.429	1.626\\
22.857	3.361\\
34.286	5.294\\
45.714	7.168\\
57.143	9.086\\
68.571	10.955\\
80.000	12.877\\
};
\addlegendentry{Single TX (TX 1)}

\end{axis}

\end{tikzpicture}%
					\caption{Performance of different precoding schemes for a setting with 2 single-antenna TXs and 2 RXs with instantaneous power constraint and $\alpha\expo=1$, $\alpha\expt=0.6$.} 
					\label{fig:M2_11_K2}%
			\end{figure}
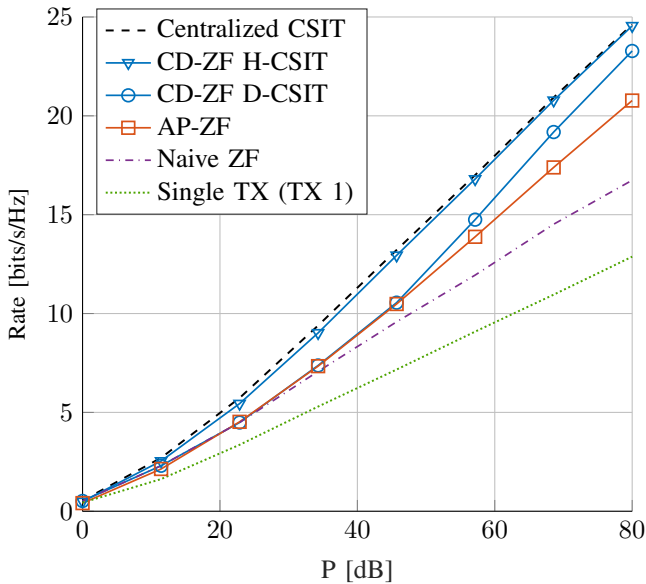

		\section{Conclusions}\label{se:conclusions}
		We have analyzed the achievable rate through linear precoding in a Network MISO setting with Distributed CSIT. 
		We have developed an achievable scheme that asymptotically attains the achievable rate of an ideal centralized setting where every TX is endowed with the best estimate among all the TXs. 
		\cmtblu{For the case in which the CSIT is distributed following a hierarchical structure among TXs}, the previous asymptotic insight is shown to be valid also at the moderate-SNR regime, where the performance obtained at both distributed and centralized settings are alike. 
		This result implies that we are not only able to reach the DoF performance of the ideal centralized setting but also its \emph{beamforming gain} when there is precise CSIT only at a subset of the transmit antennas.   
		Thus, the performance degradation on account of the CSI mismatches between TXs can be overcome by a properly designed precoding scheme which is aware of the distributed nature of the setting.  
		Besides this, it has been shown that the quantization of the information available at certain nodes is helpful as it facilitates the consistency of the decision at all the transmitters. 
		This last result could be applied to a broad set of distributed problems, in which the trade-off between global consistency and local precision has not been deeply analyzed yet. \cmtblu{Furthermore, this work fosters the question of whether we should \emph{enforce} heterogeneous CSIT allocation when designing the CSI sharing mechanism. }

\begin{appendices}

		\section{A useful Lemma}\label{app:precoder_useful_lemmas}
		We present in the following a lemma that is key for the next proofs. 
		The impact of the channel quantization on the precoder design can be asymptotically computed as shown below.
			\begin{lemma}\label{lem:convergence_error_tx2}
					Let TX~$j$, $2\leq j\leq M$, quantize its CSIT with a scalar uniform quantizer with quantization step $q = \Pb^{-\alpha_q}$, $\alpha\expM>\alpha_q>0$. The naive precoder of Section~\ref{se:scheme_distr} at TX~j for any $i\in\Nb_K$ satisfies 
						\eqm{
								\Exp\left[ \norm{\vv_{i,j}-\wv_{i,j}}\ \,\right]   & = \Oc\LB \Pb^{-\alpha_q}\RB \label{eq:lemma_diff}\\
								\Exp\left[\norm{\vv_{i,j}-\wv_{i,j}}^2\right] & = \Oc\LB P^{-\alpha_q}\RB.\label{eq:lemma_diff_2}
						}
			\end{lemma}
			\begin{proof} 
					The proof is provided in Appendix~\ref{app:lemma_precoder_tx2}.
			\end{proof}	
		Lemma~\ref{lem:convergence_error_tx2} is based on error propagation properties of linear systems. 
		Thus, it is expected to hold for a broad set of noisy estimation models whose error variance scales as $P^{-a}$ for any $a>0$. 
		For example, it holds for the quantized feedback model of\cite{Jindal2006}, in which random vector quantization is assumed and the number of quantization bits scales with~$P$, as shown in \cite{Bazco2019_ISIT}. 
		 Furthermore,  Lemma~\ref{lem:convergence_error_tx2}  leads to the following corollary. 
			\begin{corollary}\label{lem:convergence_error_global}
					Let TX~$j$, $2\leq j\leq M$, quantize its CSIT with a scalar uniform quantizer with quantization step $q = \Pb^{-\alpha_q}$, $\alpha\expM>\alpha_q>0$. The global precoder  of Section~\ref{se:scheme_distr} satisfies 
						\eqm{
								\Exp\left[\norm{\vv_{i}-\wv_{i}}\ \,\right] 	&= \Oc\LB \Pb^{-\alpha_q}\RB \label{eq:cor_diff_app_1}\\ 
								\Exp\left[\norm{\vv_{i}-\wv_{i}}^2\right] &= \Oc\LB P^{-\alpha_q}\RB. \label{eq:cor_diff_app_2}
						}
			\end{corollary}
			\begin{proof} 
				The proof is relegated to Appendix~\ref{app:lemma_precoder_tx2}.
			\end{proof}

		\section{Proof of Lemma~\ref{lem:eta_convergence} and  Lemma~\ref{lem:eta_convergence2}}\label{se:proofs}
		In this section we prove  Lemma~\ref{lem:eta_convergence} and  Lemma~\ref{lem:eta_convergence2}, which are instrumental for the proof of Theorem~\ref{theo:rate_convergence}. 
			\subsection{Proof of Lemma~\ref{lem:eta_convergence}}\label{subse:prooflemma_eta1}
			We aim to prove that, for any $\eta = \Pb^{-\varepsilon}$, with $\alpha_q> \varepsilon>0$  and $\varepsilon$ arbitrarily small, it holds that
				\eqm{
						\Pr\LB\Fc'_\Ic \geq 1+\eta	\RB  &= o\LB\frac{1}{\log_2(P)}\RB
				}
			and%
				\eqm{
					\Pr\LB\frac{1}{\Fc'_\Ic} \geq 1+\eta	\RB &= o\LB\frac{1}{\log_2(P)}\RB.
				}							
			We start by noting that $\hv_i$ can be written as $\hv_i = \frac{1}{\zop\expj}(\bhh\expj_i - z\expj\ \bdelta\expj_i)$ from the definition of the estimate in~\eqref{eq:dcsit}. Let us introduce the notations $\zop\expj_{\inv}=\frac{1}{\zop\expj}$ and $z\expj_{n}=\frac{z\expj}{\zop\expj}$. Hence, for any $i\neq\ell$, it follows that
				\eqm{
					\abs{\hv_i\wv_{\ell}} 
						&\overset{\mathclap{(a)}}{=} \abs{\smash{\zopi\bhh\expo_i\wv_{\ell} - \zzop\bdelta\expo_i\wv_{\ell}}} \nonumber \\
						&\overset{\mathclap{(b)}}{=} \abs{\smash{\zopi\bhh\expo_i\vv_{\ell} - \zzop\bdelta\expo_i\wv_{\ell} + \zzop\bdelta\expo_i\vv_{\ell}  -  \zzop\bdelta\expo_i\vv_{\ell}}}  \nonumber\\ 
						&\overset{\mathclap{(c)}}{=} \abs{\smash{\hv_i\vv_{\ell} - \zzop\bdelta\expo_i(\wv_{\ell} - \vv_{\ell})}},\!\!\!\!\!\!\!\!\!\!\!\!\!\!\!\!\!\label{eq:change_vw}
				}
			where $(a)$ and $(c)$ come from the D-CSIT model of~\eqref{eq:dcsit} and $(b)$ from the precoder definition in~\eqref{eq:def_prec_distr_1} since $\bhh\expo_i\vv_{\ell} = \bhh\expo_i\wv_{\ell}$. 
			Let us define $\dv^{\wv,\vv}_{\ell} \triangleq \wv_{\ell}-\vv_{\ell}$ for the sake of readability and space. 
			Hence, from the triangular inequality it follows that
				\eqmo{
					&\frac{1 +  \PfracK\sum_{\ell\neq i}\abs{\hv_i\wv_{\ell}}^2 }{1 + \PfracK\sum_{\ell\neq i}\abs{\hv_i\vv_{\ell}}^2} \\ 
					&\ \leq 1 + \frac{\PfracK}{1 + \PfracK\sum_{\ell\neq i}\abs{\hv_i\vv_{\ell}}^2}
						\sum_{\ell\neq i} \Big(\absn{\zzop\bdelta\expo_i\dv^{\wv,\vv}_{\ell}}^2 \Big. \\
						& \hspace{22ex} \Big. + 2\PfracK\absn{\hv_i\vv_{\ell}}\absn{\zzop\bdelta\expo_i\dv^{\wv,\vv}_{\ell}}\Big). 
					\label{eq:eq_min_1}
				}			
			Let us recall that
				\eqmo{
					\Pr\Big(\sum_{k=1}^{K} A_k \geq c\Big) \leq \sum_{k=1}^{K}\Pr\Big(A_k\geq \frac{c}{K}\Big). \label{eq:sum_prob_ineq}
				}  
			From~\eqref{eq:eq_min_1} and~\eqref{eq:sum_prob_ineq}, it follows that
				\eqm{
					&\Pr\LB\Fc'_\Ic \geq 1 + \eta\RB \notag\\
							& \leq\!\sum_{\ell\neq i}\!\Pr\!\Big(\!\frac{\PfracK  \big(\absn{\zzop\bdelta\expo_i\dv^{\wv,\vv}_{\ell}}^2\! +\! 2\PfracK\absn{\hv_i\vv_{\ell}}\absn{\zzop\bdelta\expo_i\dv^{\wv,\vv}_{\ell}}\big)}{1 + \PfracK\sum_{\ell\neq i}\abs{\hv_i\vv_{\ell}}^2}\!\geq\! \frac{\eta}{K}\!\Big) \notag\\			
							& \leq   \sum_{\ell\neq i}\Pr\LB \frac{\absn{\zzop\bdelta\expo_i\dv^{\wv,\vv}_{\ell}}^2 + 2\absn{\hv_i\vv_{\ell}}\absn{\zzop\bdelta\expo_i\dv^{\wv,\vv}_{\ell}}}{\abs{\hv_i\vv_{\ell}}^2}  \geq  \frac{\eta}{K} \RB \notag\\
							&\overset{(a)}{=}\!\!(K-1)\!\Pr\!\Big( \frac{\absn{\zzop\bdelta\expo_i\dv^{\wv,\vv}_{\ell}}^2\!+\!2\absn{\hv_i\vv_{\ell}}\absn{\zzop\bdelta\expo_i\dv^{\wv,\vv}_{\ell}}}{\abs{\hv_i\vv_{\ell}}^2}\!\geq\!\frac{\eta}{K} \!\Big)  \notag\\
							& \overset{(b)}{\leq} (K-1)\Bigg(\Pr\bigg( \frac{\absn{\zzop\bdelta\expo_i\dv^{\wv,\vv}_{\ell}}^2}{\abs{\hv_i\vv_{\ell}}^2}  \geq  \frac{\eta}{2K} \bigg) \notag\\
							&\quad \hspace{12ex}		+ \Pr\bigg( \frac{ 2\absn{\hv_i\vv_{\ell}}\absn{\zzop\bdelta\expo_i\dv^{\wv,\vv}_{\ell}}}{\abs{\hv_i\vv_{\ell}}^2}  \geq  \frac{\eta}{2K}\bigg) \Bigg)\notag \\
							& \overset{(c)}{\leq} 2(K-1) \Pr\bigg( \frac{\absn{\zzop\bdelta\expo_i(\wv_{\ell}-\vv_{\ell})}}{\abs{\hv_i\vv_{\ell}}}  \geq  \frac{\eta}{4K} \bigg) \label{eq:proof_lemma_eta_Z}
				}
			where $(a)$ comes from symmetry, $(b)$ from~\eqref{eq:sum_prob_ineq}, and $(c)$ because $\eta<1$ and $\dv^{\wv,\vv}_{\ell} \triangleq \wv_{\ell}-\vv_{\ell}$. 
			Let us now introduce  a parameter $\gamma\in\Rb$. 
			We can continue as  
				\eqmo{
					& \Pr\bigg( \frac{\absn{\zzop\bdelta\expo_i(\wv_{\ell}-\vv_{\ell})}}{\abs{\hv_i\vv_{\ell}}}  \geq  \frac{\eta}{4K} \bigg)\\
						&\qquad  = \Pr\LB  {\absn{\bdelta\expo_i(\wv_{\ell}-\vv_{\ell})} } \geq \frac{\eta}{4K}{\absn{\bdelta\expo_i\vv_{\ell}}}\RB \\	
						&\qquad \leq \Pr\LB \absn{\bdelta\expo_i\vv_{\ell}} < \Pb^{-\gamma}\RB \\
						&\qquad \quad +  \int_{\absn{\bdelta\expo_i\vv_{\ell}}\geq\Pb^{-\gamma}}\!\!\!\!  \frac{\ExpB{{\absn{\bdelta\expo_i(\wv_{\ell}-\vv_{\ell})} }}}{\frac{\eta}{4K}{y}}f_{\absn{\bdelta\expo_i\vv_{\ell}}}(y)\dd y\!\!\!\!\!\label{eq:eq_bound_int_1}
				}							
			where the first equality comes from the fact that  ${\absn{\hv_i\vv_{\ell}}} = \zzop{\absn{\bdelta\expo_i\vv_{\ell}}}$, and the last inequality from the Law of Total Probability and Markov's Inequality. $f_{\absn{\bdelta\expo_i\vv_{\ell}}}$ stands for the probability density function of $\absn{\bdelta\expo_i\vv_{\ell}}$. 
			Let us focus first on the first term of~\eqref{eq:eq_bound_int_1}, $\Pr\LB \absn{\bdelta\expo_i\vv_{\ell}} < \Pb^{-\gamma}\RB$, which satisfies the following proposition. 
				\begin{proposition}\label{lem:prob_zero_order}
					Let $\gamma>0$. Then, 
						\eqm{
								\Pr\LB \absn{\bdelta\expo_i\vv_{\ell}} < \Pb^{-\gamma}\RB  =  o\LB\frac{1}{\log_2(P)}\RB. \label{eq:eq_lem2}  
						}
				\end{proposition}		
				\begin{proof} 
					The proof is relegated to Appendix~\ref{app:proof_propositions}.
				\end{proof} 
			On the other hand, the integral term of~\eqref{eq:eq_bound_int_1} can be rewritten~as
				\eqmo{
					&\!\!\!\! \int_{{\absn{\bdelta\expo_i\vv_{\ell}}\geq\Pb^{-\gamma}}}^{}  \ExpB{{\absn{\bdelta\expo_i(\wv_{\ell}-\vv_{\ell})}}}\frac{f_{\absn{\bdelta\expo_i\vv_{\ell}}}(y)}{\frac{\eta}{4K}{y}}\dd y\\
							&\  = \frac{4K}{\eta}\ExpB{{\absn{\bdelta\expo_i(\wv_{\ell}-\vv_{\ell})} }} \ExpBs{\mid \absn{\bdelta\expo_i\vv_{\ell}}\geq\Pb^{-\gamma}}{\frac{1}{\absn{\bdelta\expo_i\vv_{\ell}}}} \\
							&\   \leq \frac{4K}{\eta}\ExpB{{\absn{\bdelta\expo_i(\wv_{\ell}-\vv_{\ell})} }} \Pb^{\gamma}.
				}							
			Now, we introduce another useful proposition, whose proof is also relegated to Appendix~\ref{app:proof_propositions}. 
				\begin{proposition}\label{lem:exp_noise_order}
					It holds that
						\eqm{
							\ExpB{{\absn{\bdelta\expo_i(\wv_{\ell}-\vv_{\ell})} }} 
									& = \Oc(\Pb^{-\alpha_q}).
						}
				\end{proposition}								
			By applying Proposition~\ref{lem:prob_zero_order} and Proposition~\ref{lem:exp_noise_order} into~\eqref{eq:eq_bound_int_1}, it is straightforward to see that 
				\eqmo{
					\Pr\LB \Fc'_\Ic \geq 1 + \eta\RB 
							& \leq o\LB\frac{1}{\log_2(P)}\RB \\
							& \quad + \frac{8K(K-1)}{\eta} \Oc(\Pb^{-\alpha_q}) \Pb^{\gamma}.
				}	
			Since $\eta = \Pb^{-\varepsilon}$, with $\alpha_q> \varepsilon>0$, 
				\eqm{
					\Pr\LB \Fc'_\Ic  \geq 1 + \eta\RB 
							& \leq o\LB\frac{1}{\log_2(P)}\RB + \Pb^{\varepsilon} \Oc(\Pb^{-\alpha_q}) \Pb^{\gamma}. \label{eq:eq_bound_int_5}
				}			
			Let us select $\gamma$ such that $\gamma > 0$ and $\varepsilon+\gamma-\alpha_q<0$. 
			Then, it follows that
				\eqmo{
					\Pr\LB \Fc'_\Ic \geq 1 + \eta\RB 
							& = o\LB\frac{1}{\log_2(P)}\RB,
				}		
			which concludes the proof of the first statement of Lemma~\ref{lem:eta_convergence}. 
			We prove in the following the second statement, i.e.,					
				\eqm{
						\Pr\LB\frac{1}{\Fc'_\Ic}\geq 1+\eta	\RB = o\LB\frac{1}{\log_2(P)}\RB.
				}
			This is obtained by switching the vectors $\vv_{\ell}$ and  $\wv_{\ell}$ and applying the same steps as in the proof of the first statement. 
			To begin with, by following the steps in~\eqref{eq:proof_lemma_eta_Z} we can  easily obtain that 
				\eqm{
						&\Pr\LB\frac{1}{\Fc'_\Ic}\geq 1+\eta	\RB
							 = \Pr\LB\frac{1 + \PfracK\sum_{\ell\neq i}\abs{\hv_i\vv_{\ell}}^2}{1 +  \PfracK\sum_{\ell\neq i}\abs{\hv_i\wv_{\ell}}^2 }  \geq 1 + \eta\RB \notag\\  
							&\qquad\  \leq 2(K-1) \Pr\bigg( \frac{\absn{\zzop\bdelta\expo_i(\wv_{\ell}-\vv_{\ell})}}{\abs{\hv_i\wv_{\ell}}}  \geq  \frac{\eta}{4K} \bigg).\!  \label{eq:proof_lemma_eta2_Z} 
				}
			Furthermore, 
			the final expression in~\eqref{eq:proof_lemma_eta2_Z} is equal to the one in~\eqref{eq:proof_lemma_eta_Z} except for the fact that the denominator is~$\abs{\hv_i\wv_{\ell}}$ instead of $\abs{\hv_i\vv_{\ell}}$. 
			Hence, continuing as in \eqref{eq:eq_bound_int_1}-\eqref{eq:eq_bound_int_5}, we obtain that
				\eqm{
					 \Pr\LB \frac{1}{\Fc'_\Ic}  \geq 1 + \eta\RB 
							& = o\LB\frac{1}{\log_2(P)}\RB,
				}		
			which concludes the proof of Lemma~\ref{lem:eta_convergence}.\qed

			\subsection{Proof of Lemma~\ref{lem:eta_convergence2}}\label{subse:prooflemma_eta2}
			We aim to prove that, for any $\eta = \Pb^{-\varepsilon}$, with $\alpha_q> \varepsilon>0$  and $\varepsilon$ arbitrarily small, it holds that
				\eqm{
						\Pr\LB\Fc'_\Dc\geq 1+\eta	\RB = o\LB\frac{1}{\log_2(P)}\RB\label{eq:lemma4_first}
				}
			and 
				\eqm{
					\Pr\LB\frac{1}{\Fc'_\Dc} \geq 1+\eta	\RB = o\LB\frac{1}{\log_2(P)}\RB.  
				}					
			Firstly, we focus on~\eqref{eq:lemma4_first}. Note that, applying similar steps as in~\eqref{eq:change_vw}, it holds that
				\eqmo{
					\abs{\hv_i\vv_{\ell}}^2 & \leq \absn{\hv_i\wv_{\ell}}^2 + \absn{\zzop\bdelta\expo_i(\wv_{\ell}-\vv_{\ell})}^2 \\
							&\qquad + 2\absn{\hv_i\wv_{\ell}}\absn{\zzop\bdelta\expo_i(\wv_{\ell}-\vv_{\ell})}, \label{eq:ineq_3}
					}
				\eqmo{
					\abs{\hv_i\vv_{{i}}}^2 		 & \leq \abs{\hv_i\wv_{{i}}}^2 		+ \absn{\hv_i(\wv_{{i}}-\vv_{{i}})}^2  \\
						  &\qquad + 2\abs{\hv_i\wv_{{i}}}\abs{\hv_i(\wv_{{i}}-\vv_{{i}})}.~~~~~~ \label{eq:ineq_4}
				}
			Hence, following the steps applied in~\eqref{eq:eq_min_1}-\eqref{eq:proof_lemma_eta_Z}, we can write that
				\eqmo{
					&\Pr\LB\frac{1 + \PfracK \sum_{\ell\in\Nb_K\!}\abs{\hv_i\vv_{\ell}}^2}{1 + \PfracK\sum_{\ell\in\Nb_K\!}\abs{\hv_i\wv_{\ell}}^2}\geq 1 + \eta\RB\\
							&\qquad\qquad\qquad\qquad\leq \Pr\LB \Dc_{1}+\Dc_{2}+\Dc_{3}+\Dc_{4}  \geq  {\eta} \RB  \\
							&\qquad\qquad\qquad\qquad\leq \sum_{i=1}^4\Pr\LB \Dc_{i}  \geq  \frac{\eta}{4}\RB, \label{eq:proof_fd_22}
				}				
			where we have introduced the notations
				\eqm{
					\Dc_{1} &\triangleq \sum_{\ell\neq i}\frac{ \absn{\zzop\bdelta\expo_i(\wv_{\ell}-\vv_{\ell})}^2}{\abs{\hv_i\wv_i}^2 } \\
					\Dc_{2} &\triangleq \sum_{\ell\neq i}\frac{2\abs{\hv_i\wv_{\ell}}\absn{\zzop\bdelta\expo_i(\wv_{{i}}-\vv_{\ell})}}{\abs{\hv_i\wv_i}^2 }\\
					\Dc_{3} &\triangleq \frac{\abs{\hv_i(\wv_{{i}}-\vv_{{i}})}^2}{\abs{\hv_i\wv_i}^2 } 		\\
					\Dc_{4} &\triangleq \frac{2\abs{\hv_i\wv_{{i}}}\abs{\hv_i(\wv_{{i}}-\vv_{{i}})}}{\abs{\hv_i\wv_i}^2}.
				}		
			The first inequality in~\eqref{eq:proof_fd_22} is obtained by applying~\eqref{eq:ineq_3}-\eqref{eq:ineq_4} and eliminating the term $1 + \PfracK\sum_{\ell\neq i}\abs{\hv_i\wv_{\ell}}^2$ from the denominator. 		 
			From the analysis of $\Fc_\Ic$ in the previous section (see~\eqref{eq:eq_bound_int_1}), it follows easily that, if $\eta = \Pb^{-\varepsilon}$, with $\alpha_q> \varepsilon>0$, then
				\eqm{
						\Pr\LB \Dc_{1} \geq \frac{\eta}{4}\RB 
							& \leq \sum_{\ell\neq i} \Pr\LB \zzop\frac{ \absn{\bdelta\expo_i(\wv_{\ell}-\vv_{\ell})}^2}{\abs{\hv_i\wv_i}^2 } \geq \frac{\eta}{4K}\RB  \notag \\
							& = o\bigg(\frac{1}{\log_2(P)}\bigg). 
				}
			Similarly, 
				\eqm{
						\Pr\LB \Dc_{2} \geq \frac{\eta}{4}\RB 
							& =\! \sum_{\ell\neq i}\Pr\!\Big(\! \zzop\frac{2\abs{\hv_i\wv_{\ell}}\absn{\bdelta\expo_i(\wv_{{i}}-\vv_{\ell})}}{\abs{\hv_i\wv_i}^2 } \geq \frac{\eta}{4K}\!\Big) \notag  \\
							& = o\bigg(\frac{1}{\log_2(P)}\bigg). 								
				}
			For the two remaining terms, $\Dc_3$ and $\Dc_4$, note that
				\eqmo{
					&\Pr\LB \Dc_{3} \geq \frac{\eta}{4}\RB + \Pr\LB \Dc_{4} \geq \frac{\eta}{4}\RB \\
						&\hspace{18ex} = \Pr\LB {\frac{\abs{\hv_i(\wv_{{i}}-\vv_{{i}})}^2}{\abs{\hv_i\wv_i}^2 }}\geq \frac{\eta}{4}\RB\\
							&\hspace{18ex} \quad\ +  \Pr\LB {\frac{2\abs{\hv_i(\wv_{{i}}-\vv_{{i}})}}{\abs{\hv_i\wv_i}}} 	\geq \frac{\eta}{4}\RB \\
						&\hspace{18ex} \leq 2 \Pr\LB {\frac{\abs{\hv_i(\wv_{{i}}-\vv_{{i}})}}{\abs{\hv_i\wv_i}}} 	\geq \frac{\eta}{16}\RB \\
						&\hspace{18ex} = 2 \Pr\LB {\frac{\absn{\thv_i(\wv_{{i}}-\vv_{{i}})}}{\absn{\thv_i\wv_i}}} 	\geq \frac{\eta}{16}\RB. \label{eq:proof_lemma_d_7}
				}	
			where $\thv = \frac{\hv}{\norm{\hv}}$  is unit-norm and it is isotropically distributed  on the $N_T$-dimensional unit-sphere\cite{Jindal2006}. We can continue as in~\eqref{eq:eq_bound_int_1} to write
				\eqmo{
						 &\Pr\LB {\frac{\absn{\thv_i(\wv_{{i}}-\vv_{{i}})}}{\absn{\thv_i\wv_i}}} 	\geq \frac{\eta}{16}\RB
							  \leq \Pr\LB \absn{\thv_i\wv_i} < \Pb^{-\gamma}\RB \\
							&\qquad\  \quad + \int_{\absn{\thv_i\wv_i}\geq\Pb^{-\gamma}}^{}  \frac{\ExpB{\smash{{\absn{\thv_i(\wv_{{i}}-\vv_{{i}})} }}}}{\frac{\eta}{16}{y}}f_{\absn{\thv_i\wv_i}}(y)\dd y \\
							&\qquad  \leq \Oc(\Pb^{-\gamma}) + 16 \Pb^{\varepsilon}\ExpB{\smash{{\absn{\thv_i(\wv_{{i}}-\vv_{{i}})} }}}\Pb^{\gamma}.\label{eq:bound_prob_fd_4}
				}   
			The fact that $\norm{\thv_i} = 1$ implies that $\ExpB{\smash{{\absn{\thv_i(\wv_{{i}}-\vv_{{i}})} }}}  \leq \ExpB{\smash{{\norm{\wv_{{i}}-\vv_{{i}}}}}}$. Moreover,    
			Corollary~\ref{lem:convergence_error_global} states that $\ExpB{\smash{\norm{\wv_{i}-\vv_{i}}}}= \Oc(\Pb^{-\alpha_q})$.  			
			Consequently, by selecting $\gamma$ such that $\gamma > 0$ and $\varepsilon+\gamma-\alpha_q<0$, 
			it follows from~\eqref{eq:bound_prob_fd_4} that
				\eqm{
					\Pr\LB {\frac{\absn{\thv_i(\wv_{{i}}-\vv_{{i}})}}{\absn{\thv_i\wv_i}}} 	\geq \frac{\eta}{16}\RB
							&  \leq \Oc(\Pb^{-\gamma}) + \Pb^{\varepsilon}\Oc(\Pb^{-\alpha_q})\Pb^{\gamma} \nonumber\\
							& = o\LB\frac{1}{\log_2(P)}\RB. \label{eq:proof_lemma_d_end} 
				}		
			We  introduce the result of~\eqref{eq:proof_lemma_d_end} into~\eqref{eq:proof_lemma_d_7} to obtain from~\eqref{eq:proof_fd_22}~that
				\eqm{
						\Pr\LB\Fc'_\Dc\geq 1+\eta	\RB = o\LB\frac{1}{\log_2(P)}\RB.
				}
			It would remain to prove that $\Pr\LB\tfrac{1}{\Fc'_\Dc}\geq 1+\eta	\RB = o\LB\frac{1}{\log_2(P)}\RB$. To do so, we just need to apply the same previous steps, but  interchanging $\wv$ and $\vv$. Following those steps and following a similar argument as in the proof for $\Pr\Big(\tfrac{1}{\Fc'_\Ic}\Big)$, we directly obtain the result. 
			For this reason, and for the sake of concision, we omit the derivation.  \qed

		\section{Proof of~Lemma~\ref{lem:outage} (Probability of Power Outage)}\label{app:proof_outage}
				
		  We denote the event of power outage as $\pset$. 
			Note that
				\eqm{
					\Pr\LB\pset\RB  \leq N_1\Pr\LB\norm{\bT_{1,1}}>1\RB,
				}
			and $\bT_{1,1} = \mu[\wv_{1,1,1},\wv_{2,1,1},\dots,\wv_{K,1,1}]$, where $\wv_{i,j,n}$ represents the $n$-th element of the precoding vector at TX~$j$ for the data symbols of RX~$i$.  
			Therefore, 
				\eqmo{
					\Pr\LB\pset\RB 
							& \leq N_1\Pr\LB\norm{\mu[\wv_{1,1,1},\wv_{2,1,1},\dots,\wv_{K,1,1}]} > 1\RB \\
							& \overset{(a)}{\leq} N_1\Pr\LB\bigcup_{i\in\Nb_K} \norm{\mu\wv_{i,1,1}} > \norm{\vv_{i,1,1}}\RB \\ 
							& \overset{(b)}{\leq} N_1K\Pr\LB\norm{\mu\wv_{1,1,1}} > \norm{\vv_{1,1,1}}\RB \\ 
							& \overset{(c)}{\leq} N_1K\Pr\LB\mu\norm{\vv_{1,1,1}} + \mu\norm{\bphi_{1}} > \norm{\vv_{1,1,1}}\RB \\
							& = N_1K\Pr\LB\norm{\bphi_{1}} > \frac{1-\mu}{\mu}\norm{\vv_{1,1,1}}\RB
				} 
			where $(a)$ is obtained from the precoder definition as $\norm{[\vv_{1,1,1}\ \dots\  \vv_{K,1,1}]}\leq 1$, $(b)$  follows because $\wv_{i,1,1}$ (resp. $\vv_{i,1,1}$) is equally distributed for any $i\in\Nb_K$, and $(c)$ from~\eqref{eq:eq_1b}.  
			Now, we  obtain the probability by conditioning on $\norm{\vv_{1,1}}$ and then averaging over the distribution of $\norm{\vv_{1,1,1}}$. Let us denote $\mu'\triangleq \frac{1-\mu}{\mu}$. Hence, 
				\eqm{
					\Pr\LB\pset\RB 			
							& \leq N_1K\int_{-\infty}^{\infty} \Pr\LB\norm{\bphi_{1}} > \mu'\nu \RB f_{\norm{\vv_{1,1,1}}}(\nu)\dd \nu.  
				}
			Using Markov's inequality  we obtain that 
				\eqmo{
					\Pr\LB\pset\RB 			
							& \leq N_1K\int_{-\infty}^{\infty} \frac{\Exp[{\norm{\bphi_{1}}}]}{ \mu'\nu}  f_{\norm{\vv_{1,1,1}}}(\nu)\dd \nu \\ 
							& = N_1K\Exp[{\norm{\bphi_{1}}}]\frac{1}{ \mu'}\E\big[{\norm{\vv_{1,1,1}}^{-1}}\big], \label{eq:chi_int} 
				}
			where $\E\big[{\norm{\vv_{1,1,1}}^{-1}}\big]$ exists from property~\eqref{eq:cond_2}. 
			Let us focus on the first expectation term of~\eqref{eq:chi_int} ($\Exp[{\norm{\bphi_{1}}}]$). 
			Recalling~\eqref{eq:eq_1b}, $\bphi_{i}$ is defined as  
						\eqm{
							 \bphi_{i}  = \bHH_{\iop,1}^\dagger\bHH_{\iop,\bar{1}}(\vv_{i,\bar{1}} - \wv_{i,\bar{1}})  \label{eq:eq_1y}. 
						}	
			Then, 			
				\eqmo{
					\Exp\left[\norm{\bphi_{i}}\right] 
							& \overset{(a)}{\leq}  \Exp\left[\norm{\bHH_{\iop,1}^\dagger\bHH_{\iop,\bar{1}}}\norm{\vv_{i,\bar{1}} - \wv_{i,\bar{1}}}\right] \\
							& \overset{(b)}{\leq}  \sqrt{\Exp\left[\norm{\bHH_{\iop,1}^\dagger\bHH_{\iop,\bar{1}}}^2\right] \Exp\left[\norm{\vv_{i,\bar{1}} - \wv_{i,\bar{1}}}^2\right]}, \label{eq:ineq_sigma}
				}	
			where $(a)$ comes from the sub-multiplicative property of the Frobenius norm and $(b)$ from Cauchy–Schwarz inequality. 
			Let us denote $g_m \triangleq \sqrt{\smash[b]{\Exp\big[\norm{\bHH_{\iop,1}^\dagger\bHH_{\iop,\bar{1}}}^2\big]}}$, which is a value that does not depend on $P$ since the channel estimates are equally distributed for any estimation error variance. 
			Then, we have that
				\eqm{
						\Exp\left[\norm{\bphi_{i}}\right]
									& \leq  g_m\sqrt{\Exp\left[\norm{\vv_{i,\bar{1}} - \wv_{i,\bar{1}}}^2\right]}. \label{eq:error_scaling_tx} 
				}				
			Lemma~\ref{lem:convergence_error_tx2} and the fact that $\Exp\left[\norm{\vv_{i,\bar{1}} - \wv_{i,\bar{1}}}^2\right] = \sum_{j=2}^M\Exp\big[\norm{\vv_{i,j} - \wv_{i,j}}^2\big]$ yield 
				\eqm{
					\Exp\left[\norm{\bphi_{1}}\right] = \Oc\Big( \Pb^{-\alpha_q}\Big).
				}
			Since $\mu = 1 - \varpi$, with $\varpi = \Theta(\Pb^{-\alpha_{\mu}})$ and $\varpi>0$, the term $\frac{1}{ \mu'} = \frac{\mu}{1-\mu}$ satisfies 
						$\frac{1}{ \mu'} = \Theta(\Pb^{\alpha_{\mu}})$. 
			From~\eqref{eq:cond_2}, $\E\big[{\norm{\vv_{1,1,1}}^{-1}}\big]$ is $\Theta(1)$. 
			Hence, recalling~\eqref{eq:chi_int},
				\eqmo{
					\Pr\LB\pset\RB 
							& \leq N_1K\E\big[{\norm{\vv_{1,1,1}}^{-1}}\big]\Exp[{\norm{\bphi_{1}}}]\frac{1}{ \mu'} \\
							& =  \Theta(1) \Oc\LB \Pb^{-\alpha_q}\RB \Theta(\Pb^{\alpha_{\mu}}), 
				}
			which implies that $\Pr\LB\pset\RB = \Oc\LB \Pb^{\alpha_{\mu}-\alpha_q} \RB $. 
			By selecting $\alpha_\mu < \alpha_q$, the probability of power outage vanishes and it holds that
				\eqm{
					\Pr\LB\pset\RB =  
							o\LB \frac{1}{\log_2(P)} \RB, 
				}		
			which concludes the proof.  \qed

		\section{Proof of Proposition~\ref{lem:prob_zero_order}  and Proposition~\ref{lem:exp_noise_order}}\label{app:proof_propositions}
							
			\subsection{Proof of Proposition~\ref{lem:prob_zero_order}}\label{app:proof_lemma_gamma}
			We prove in the following that $\Pr\big(\absn{\bdelta\expo_i\vv_{\ell}} < \Pb^{-\gamma}\big)  = o\LB\frac{1}{\log_2(P)}\RB$ for any $\gamma>0$ and $i,\ell\in\Nb_K$ such that $\ell\neq i$. 
			Let us denote the precoder for RX~$\ell$ obtained with perfect knowledge of $\bH$ as $\uv_\ell$. Then,
				\eqm{
						&\!\!\!\!\Pr\big(\absn{\bdelta\expo_i\vv_{\ell}} < \Pb^{-\gamma}\big)  
							 \!=\! \Pr\big(\absn{\bdelta\expo_i\uv_{\ell} + \bdelta\expo_i(\vv_\ell-\uv_{\ell})}  < \Pb^{-\gamma}\big)\notag\\
							&\hspace{8ex}\leq \Pr\big(\big|{\absn{\bdelta\expo_i\uv_{\ell}} - \absn{\bdelta\expo_i(\vv_\ell-\uv_{\ell})}}\big|  < \Pb^{-\gamma}\big),\!\!\label{eq:first_step_prop1}
				}
			where we have applied the inverse triangle inequality. 
			In order to prove Proposition~\ref{lem:prob_zero_order}, we capitalize the intuition that the term $\absn{\bdelta\expo_i\uv_{\ell}}$ is independent of the quality of the estimate and $P$, but the value of $\absn{\bdelta\expo_i(\vv_\ell-\uv_{\ell})}$ is directly proportional to   $P^{-\alpha\expo}$. Before applying this intuition to~\eqref{eq:first_step_prop1},
			we first analyze the term $\absn{\bdelta\expo_i(\vv_\ell-\uv_{\ell})}$ and the probability $\Pr\big(\absn{\bdelta\expo_i(\vv_\ell-\uv_{\ell})} > \Pb^{-\beta}\big)$ for any $\beta<\alpha\expo$.  
			\cmtblu{	
				\begin{proposition}\label{prop:proof_appendix_prop}
					For any $\beta<\alpha\expo$, it holds that 
					 \eqm{\label{eq:proof_appendix_prop}
							\Pr\big(\absn{\bdelta\expo_i(\vv_\ell-\uv_{\ell})} > \Pb^{-\beta}\big)
								=o\LB\frac{1}{\log_2(P)}\RB.
					 }						
				\end{proposition}
			}
				\begin{proof}
					\cmtblu{	
						It follows by the Cauchy-Schwarz inequality that $\absn{\bdelta\expo_i(\vv_\ell-\uv_{\ell})} < \norm{\bdelta\expo_i}\norm{\vv_\ell-\uv_{\ell}}$, which implies that $\Pr\big(\absn{\bdelta\expo_i(\vv_\ell-\uv_{\ell})} > \Pb^{-\beta}\big) \leq \Pr\big(\norm{\bdelta\expo_i}\norm{\vv_\ell-\uv_{\ell}} > \Pb^{-\beta}\big)$. 
						Let us define the scalar $\varepsilon>0$ such that $\beta<\beta+\varepsilon<\alpha\expo$. By applying the law of total probability, we obtain that
							\eqm{
								&\Pr\big(\absn{\bdelta\expo_i(\vv_\ell-\uv_{\ell})} > \Pb^{-\beta}\big)
												 \leq \Pr\big(\norm{\bdelta\expo_i}\norm{\vv_\ell-\uv_{\ell}} > \Pb^{-\beta}\big) \notag\\
									& \quad\leq \Pr\big(\norm{\bdelta\expo_i}\norm{\vv_\ell-\uv_{\ell}} > \Pb^{-\beta}\mid\norm{\vv_\ell-\uv_{\ell}}>\Pb^{-\beta-\varepsilon}\big)	\notag\\
											&\qquad\qquad\times\Pr\big(\norm{\vv_\ell-\uv_{\ell}}>\Pb^{-\beta-\varepsilon}\big) \notag\\
										&\ \qquad + \Pr\big(\norm{\bdelta\expo_i}\norm{\vv_\ell-\uv_{\ell}} > \Pb^{-\beta}\mid\norm{\vv_\ell-\uv_{\ell}}\leq\Pb^{-\beta-\varepsilon}\big)	\notag\\
											&\qquad\qquad\times\Pr\big(\norm{\vv_\ell-\uv_{\ell}}\leq\Pb^{-\beta-\varepsilon}\big)	\notag\\
										&\quad\leq \Pr\big(\norm{\vv_\ell-\uv_{\ell}}>\Pb^{-\beta-\varepsilon}\big)	\notag\\
										&\ \qquad + \Pr\big(\norm{\bdelta\expo_i} > \Pb^{\varepsilon}\mid\norm{\vv_\ell-\uv_{\ell}}\leq\Pb^{-\beta-\varepsilon}\big). \label{eq:app_err1}
							}
					}
					The term $\Pr\big(\norm{\vv_\ell-\uv_{\ell}}>\Pb^{-\beta-\varepsilon}\big)$ can be upper-bounded by means of the Markov's inequality, such that
						\eqmo{
								&\Pr\big(\norm{\vv_\ell-\uv_{\ell}}>\Pb^{-\beta-\varepsilon}\big) 
									 \leq~ \Pb^{\beta+\varepsilon}\ {\Exp[\norm{\vv_\ell-\uv_{\ell}}]}\\
								&\qquad\ \ \qquad\qquad \overset{(a)}{=} \Oc(\Pb^{\beta+\varepsilon-\alpha\expo}) =o\LB\frac{1}{\log_2(P)}\RB, \label{eq:order_difuv1}
						}
					where $(a)$  follows  after applying Lemma~\ref{lem:convergence_error_tx2} to vectors whose respective input estimates differ by a $\Oc(\Pb^{-\alpha\expo})$ additive error term.
					For the last term in~\eqref{eq:app_err1}, $\Pr\big(\norm{\bdelta\expo_i} > \Pb^{\varepsilon}\mid\norm{\vv_\ell-\uv_{\ell}}\leq\Pb^{-\beta-\varepsilon}\big)$, it follows that
						\eqmo{
								&\Pr\big(\norm{\bdelta\expo_i} > \Pb^{\varepsilon}\mid\norm{\vv_\ell-\uv_{\ell}}\leq\Pb^{-\beta-\varepsilon}\big)\\
									&\qquad\qquad\qquad\qquad\leq \frac{\Pr\big(\norm{\bdelta\expo_i} > \Pb^{\varepsilon}\big)   }{   \Pr\big(\norm{\vv_\ell-\uv_{\ell}}\leq\Pb^{-\beta-\varepsilon}\big)}. \label{eq:prob_12e}
						}
					From~\eqref{eq:order_difuv1}, it holds that $\Pr\big(\norm{\vv_\ell-\uv_{\ell}}\leq\Pb^{-\beta-\varepsilon}\big) = 1 - \Oc(\Pb^{\beta+\varepsilon-\alpha\expo})$. 
					Besides this, $\norm{\bdelta\expo_i}^2 = \sum_{n=1}^{N_T} \absn{\delta\expo_{i,n}}^2$, where  $\delta\expo_{i,n}$ are i.i.d. as $\CN(0,1)$. Consequently, $\absn{\delta\expo_{i,n}}^2$ is distributed following a Rayleigh distribution and 
						\eqmo{
								\norm{\bdelta\expo_i}^2 \sim \Gamma_d(N_T,1), 
						}
					where $\Gamma_d(N_T,1)$ denotes the Gamma distribution. 
					Moreover,  $\Gamma_d(N_T,1)$ is also called the Erlang distribution, and it satisfies that
						\eqm{
							\Pr\big(X\sim\Gamma_d(N_T,1) <x\big) = 1 - \sum_{n=0}^{N_T-1} \frac{1}{n!}e^{-x}x^n.
						}
					Hence,
						\eqmo{
								\Pr\big(\norm{\bdelta\expo_i} > \Pb^{\varepsilon}\big)
									& = \Pr\big(\norm{\bdelta\expo_i}^2 > \Pb^{2\varepsilon}\big) \\
									& = \sum_{n=0}^{N_T-1} \frac{1}{n!}e^{-\Pb^{2\varepsilon}}\Pb^{2n\varepsilon}.
						}
					Since $\varepsilon>0$, $\Pr\big(\norm{\bdelta\expo_i} > \Pb^{\varepsilon}\big) = o(\Pb^{x})$, for any~$x\in\Rb$, and hence it is $o(1/\log_2(P))$. This implies that 
						\eqmo{
							\frac{\Pr\big(\norm{\bdelta\expo_i} > \Pb^{\varepsilon}\big)   }{   \Pr\big(\norm{\vv_\ell-\uv_{\ell}}\leq\Pb^{-\beta-\varepsilon}\big)}
									& = \frac{o\LB\frac{1}{\log_2(P)}\RB}{1 - \Oc(\Pb^{\beta+\varepsilon-\alpha\expo})},
						}	
					which, together with~\eqref{eq:order_difuv1} and~\eqref{eq:app_err1}, leads to 
						\eqm{
								\Pr\big(\absn{\bdelta\expo_i(\vv_\ell-\uv_{\ell})} > \Pb^{-\beta}\big)
									=o\LB\frac{1}{\log_2(P)}\RB 
						}						
					for any $\beta<\alpha\expo$.
				\end{proof}
			Equipped with this result from Proposition~\ref{prop:proof_appendix_prop}, now we focus back on~\eqref{eq:first_step_prop1}.
			Before, let us introduce the notation $\dv^{\vv,\uv}_{\ell}\triangleq \vv_\ell-\uv_{\ell}$ for the sake of readability. 
			Then,~\eqref{eq:first_step_prop1} can be expanded 
			by means of the Law of total probability as	
				\eqmo{\notag
					&\Pr\big(\big|{\absn{\bdelta\expo_i\uv_{\ell}} - \absn{\bdelta\expo_i\dv^{\vv,\uv}_{\ell}}}\big|  < \Pb^{-\gamma}\big)\\
						&  = \Pr\Big(\big|{\absn{\bdelta\expo_i\uv_{\ell}} - \absn{\bdelta\expo_i\dv^{\vv,\uv}_{\ell}}}\big|  < \Pb^{-\gamma} ~\Big|~ \absn{\bdelta\expo_i\dv^{\vv,\uv}_{\ell}} \leq \Pb^{-\beta} \Big)\\
								&\qquad \times\Pr\big(\absn{\bdelta\expo_i\dv^{\vv,\uv}_{\ell}} \leq \Pb^{-\beta}\big) \\
							&\quad + \Pr\Big(\!\big|{\absn{\bdelta\expo_i\uv_{\ell}} - \absn{\bdelta\expo_i\dv^{\vv,\uv}_{\ell}}}\big| \! <\! \Pb^{-\gamma} ~\Big|~ \absn{\bdelta\expo_i\dv^{\vv,\uv}_{\ell}} > \Pb^{-\beta}\Big)\\
							& \qquad \times	\Pr\big(\absn{\bdelta\expo_i\dv^{\vv,\uv}_{\ell}} > \Pb^{-\beta}\big) \\
						&  = \Pr\Big(\big|{\absn{\bdelta\expo_i\uv_{\ell}} - \absn{\bdelta\expo_i\dv^{\vv,\uv}_{\ell}}}\big|  < \Pb^{-\gamma} 	~\Big|~ \absn{\bdelta\expo_i\dv^{\vv,\uv}_{\ell}} \leq \Pb^{-\beta} \Big) \\
							& \quad + o\LB\frac{1}{\log_2(P)}\RB \\
						&  \leq \Pr\big(\absn{\bdelta\expo_i\uv_{\ell}} < \Pb^{-\gamma} + \Pb^{-\beta} \big) + o\LB\frac{1}{\log_2(P)}\RB .
				}	
			Let us assume w.l.o.g. that $\beta<\gamma$, such that $\Pr\big(\absn{\bdelta\expo_i\uv_{\ell}} < \Pb^{-\gamma} + \Pb^{-\beta} \big) \leq \Pr\big(\absn{\bdelta\expo_i\uv_{\ell}} < 2\Pb^{-\beta} \big)$.  
			Therefore, it remains to prove that  $\Pr\big(\absn{\bdelta\expo_i\uv_{\ell}} < 2\Pb^{-\beta} \big) = o\LB\frac{1}{\log_2(P)}\RB$. 
			Let  $\epsilon_\beta$ be a scalar such that $0<\epsilon_\beta<\beta$ and let us define $\psi$ as the angle satisfying 		
				\eqmo{
						\cos(\psi) \triangleq \frac{\absn{\bdelta\expo_i\uv_{\ell}}}{\norm{\bdelta\expo_i}\norm{\uv_{\ell}}}.
				}
			Then, we use again the Law of Total Probability to obtain 
				\eqm{
					& \Pr\LB \absn{\bdelta\expo_i\uv_{\ell}} < 2\Pb^{-\beta}\RB	\notag\\  
						& = \Pr\LB \norm{\bdelta\expo_i}\norm{\uv_{\ell}} \cos(\psi) < 2\Pb^{-\beta} \mid \norm{\uv_{\ell}} \leq \Pb^{-\epsilon_\beta}\RB \notag \\
								& \qquad \times \Pr\LB \norm{\uv_{\ell}} \leq \Pb^{-\epsilon_\beta}\RB  \notag \\
							& \quad +  \Pr\LB \norm{\bdelta\expo_i}\norm{\uv_{\ell}} \cos(\psi) < 2\Pb^{-\beta} \mid\norm{\uv_{\ell}} > \Pb^{-\epsilon_\beta}\RB \notag \\
								& \qquad\times \Pr\LB \norm{\uv_{\ell}} > \Pb^{-\epsilon_\beta}\RB  \label{eq:proof_lemma_noise_5}\\
						& \leq \Pr\LB \norm{\uv_{\ell}} \leq \Pb^{-\epsilon_\beta}\RB\notag \\
							& \quad +  \Pr\LB \norm{\bdelta\expo_i}\norm{\uv_{\ell}} \cos(\psi) < 2\Pb^{-\beta} \mid\norm{\uv_{\ell}} > \Pb^{-\epsilon_\beta}\RB \notag\\
						& \leq \Pr\LB \norm{\uv_{\ell}} \leq \Pb^{-\epsilon_\beta}\RB \notag \\
						& \quad +  \Pr\LB \norm{\bdelta\expo_i}\cos(\psi) < 2\Pb^{-\beta}\Pb^{\epsilon_\beta} \mid\norm{\uv_{\ell}} > \Pb^{-\epsilon_\beta}\RB\notag.
				}		
			Importantly, $\bdelta\expo_i$ is isotropically distributed (i.e., the normalized value $\bdelta\expo_i/\norm{\bdelta\expo_i}$ is uniformly distributed in the sphere surface). Besides this,  $\uv_{\ell}$ is a function of $\bH$. Since $\bH$ and $\bdelta\expo_i$ are mutually independent, so  $\bdelta\expo_i$ and $\uv_{\ell}$ are. 
			Hence, from isotropy of $\bdelta\expo_i$, $\cos(\psi) $ is independent of $\uv_{\ell}$.						
			On this basis, we can select $\uv_{\ell} = [1, \bZero_{1\times N_T - 1}]$ to obtain that 
				\eqmo{
					&\Pr\LB \norm{\bdelta\expo_i}\cos(\psi) < 2\Pb^{\epsilon_\beta-\beta} \mid\norm{\uv_{\ell}} > \Pb^{-\epsilon_\beta} \RB~\\
						& \hspace{22ex}= \Pr\LB \absn{\bdelta\expo_{i,1,1}} < 2\Pb^{\epsilon_\beta-\beta}\RB,
				}		
			where $\bdelta\expo_{i,1,1}$ denotes the first element of the vector $\bdelta\expo_{i}$, and it is distributed as $\CN(0,1)$. 
			Then,
				\eqmo{
					\Pr\LB \absn{\bdelta\expo_{i,1,1}} < 2\Pb^{\epsilon_\beta-\beta}\RB 
							& = \frac{2}{\sqrt{2\pi}}\int_{0}^{2\Pb^{\epsilon_\beta-\beta}} \!\!\!\! \!\! e^{-x^2/2} \dd x \\
					 & \leq  \frac{4}{\sqrt{2\pi}}\Pb^{\epsilon_\beta-\beta}. \label{eq:proof_lemma_noise_enda}
				}				
			On the other hand, the term $\Pr\LB \norm{\uv_{\ell}} \leq \Pb^{-\epsilon_\beta}\RB$ is  bounded by 
				\eqmo{
					\Pr\LB \norm{\uv_{\ell}} \leq \Pb^{-\epsilon_\beta}\RB  
							& =  \int_{0}^{\Pb^{-\epsilon_\beta}} \!\!\!\! f^{}_{\norm{\uv_{i}}}(x) \dd x \\
							& \leq  f^{\max}_{\norm{\uv_{i}}} \Pb^{-\epsilon_\beta}, \label{eq:proof_lemma_noise_endb}
				}
			which  follows from~\eqref{eq:cond_3}. By introducing~\eqref{eq:proof_lemma_noise_enda} and \eqref{eq:proof_lemma_noise_endb} in~\eqref{eq:proof_lemma_noise_5} we obtain that
				\eqmo{
					\Pr\LB \absn{\bdelta\expo_i\uv_{\ell}} < 2\Pb^{-\beta}\RB 	 \leq \Oc(\Pb^{\max(-\epsilon_\beta,\ \epsilon_\beta-\beta)}).
				}		
			Note that $\epsilon_\beta$ satisfies $0<\epsilon_\beta < \beta$. Hence,
				\eqmo{
					&\Pr\big(\absn{\bdelta\expo_i\vv_{\ell}} < \Pb^{-\gamma}\big)\\  
						& \hspace{5ex}\leq \Pr\big(\big|{\absn{\bdelta\expo_i\uv_{\ell}} - \absn{\bdelta\expo_i(\vv_\ell-\uv_{\ell})}}\big|  < \Pb^{-\gamma}\big)\\
					& \hspace{5ex}\leq \Pr\big(\absn{\bdelta\expo_i\uv_{\ell}} < 2\Pb^{-\beta} \big) + o\LB\frac{1}{\log_2(P)}\RB \\
						& \hspace{5ex}= o\LB\frac{1}{\log_2(P)}\RB,			
				}
			which concludes the proof of Proposition~\ref{lem:prob_zero_order}. \qed 
			
			\subsection{Proof of Proposition~\ref{lem:exp_noise_order}}\label{app:proof_lemma_exp}
			We prove in the following that $\ExpB{\smash{\absn{\bdelta\expo_i(\wv_{\ell}-\vv_{\ell})} }} = \Oc(\Pb^{-\alpha_q})$ for any  $i,\ell\in\Nb_K : \ell\neq i$. 	
			It follows that
				\eqmo{\nonumber
					&\E\big[ \absn{\bdelta\expo_i(\wv_{\ell}-\vv_{\ell})}\big] \\ 
							&\ \ \ 	\leq \E\big[ {\norm{\bdelta\expo_i}\norm{\wv_{\ell}-\vv_{\ell}}} \big]	\\
							&\ \ \   =  \cov\LB{\norm{\bdelta\expo_i},\norm{\wv_{\ell}-\vv_{\ell}}} \RB + \ExpB{\norm{\bdelta\expo_i}}\ExpB{\norm{\wv_{\ell}-\vv_{\ell}}} \\
							& \ \ \  \leq  \sqrt{\ExpB{\norm{\bdelta\expo_i}^2}}\sigma_{\norm{\wv_{\ell}-\vv_{\ell}}} + \ExpB{\norm{\bdelta\expo_i}}\ExpB{\norm{\wv_{\ell}-\vv_{\ell}}},  
				}
			where $\cov(X,Y)\triangleq \Exp[(X-\Exp(X))(Y-\Exp(Y))]$ is the covariance between $X$ and $Y$ and $\sigma^2_X$ represents the variance of the random variable $X$. The last inequality comes from the fact that $\cov(x,y)\leq \sigma_x\sigma_y$ and $\sigma^2_x\leq \ExpB{x^2}$. 
			Besides this, it holds from $\norm{\bdelta\expo_i}^2 \sim \Gamma_d(N_T,1)$ that ${\ExpB{\norm{\bdelta\expo_i}^2}} = {N_T}$.   
			From this point and the fact that $\ExpB{x} \leq \sqrt{\ExpB{x^2}}$,  we can write
				\eqmo{\nonumber
						\E\big[ {\absn{\bdelta\expo_i(\wv_{\ell}-\vv_{\ell})}}\big] 
							& \leq \sqrt{N_T} \LB\sigma_{\norm{\wv_{\ell}-\vv_{\ell}}} + \ExpB{\norm{\wv_{\ell}-\vv_{\ell}}}\RB \\%
							& \overset{(a)}{\leq}  \sqrt{N_T}\ 2\sqrt{\ExpB{\norm{\wv_{\ell}-\vv_{\ell}}^2}}		\\ %
							& \overset{(b)}{=}   \Oc(\Pb^{-\alpha_q}), 	 %
				}				
			where $(a)$ comes from the fact that $\sigma_x + \E[x] \leq 2\sqrt{\E[x^2]}$ and $(b)$ from Corollary~\ref{lem:convergence_error_global}.   \qed
							
		\section{Proof of Lemma~\ref{lem:uniform_quantizer} (Quantizer Consistency)}\label{app:proof_consistent} 
		Let $q \triangleq \Pb^{-\alpha_q}$ be the quantization step size of the quantizer $\Qc_u$. Then, $\Qc_u$ is defined such that, for a scalar value $x\in\Rb$,
			\eqm{
				\Qc_u(x) \triangleq  q\left\lfloor \frac{x}{q}+\frac{1}{2}\right\rfloor.
			}			
		We extend the notation for any complex matrix $\bA\in\Cb^{n\times m}$ such that $\bA_q = \Qc_u(\bA)$ denotes the element-wise quantization, i.e.,
			\eqm{
					(\bA_q)_{i,k} \triangleq \Qc_u\big(\Real(\bA_{i,k})\big) + \mathrm{i} \Qc_u\big(\Imag(\bA_{i,k})\big), \label{eq:quantizer_element}
			}
		where $\Real(x)$ and $\Imag(x)$ stand for the real imaginary part of $x\in\Cb$, and $\iin\triangleq\sqrt{-1}$.  
		In this appendix we prove that, for a scalar uniform quantizer $\Qc_u$ with $q = \Pb^{-\alpha_q}$ and  $\alpha\expj >\alpha_q>0$, $\forall j\in \Nb_M$, it follows that
			\eqm{
					\Pr\LB\bHH^{(j)\leftarrow(1)}_q  \neq \bHH\expj_q \RB = o\LB\frac{1}{\log_2(P)}\RB,
			}		
		where $\bHH\expj_q = \Qc_u(\bHH\expj)$ and $\bHH^{(j)\leftarrow(1)}_q$ is the MAP estimator of $\bHH\expj_q$ given $\bHH\expo$. We start by noting that, by definition of the MAP estimator,
			\eqm{
				\Pr\LB\bHH^{(j)\leftarrow(1)}_q  \neq \bHH\expj_q \RB 
					& \leq \Pr\LB\Qc_u(\bHH^{(1)})  \neq \bHH\expj_q \RB.
			}
		Since $\Real(\bHH^{(1)}_{i,k})$ and $\Imag(\bHH^{(1)}_{i,k})$ are i.i.d. for any $i,k$, it follows that
			\eqmo{
				\!\! &\Pr\LB\Qc_u(\bHH^{(1)})  \neq \bHH\expj_q \RB  \\
					&\ \quad \leq 2KN_T \Pr\LB\Qc_u(\Real(\bHH^{(1)}_{1,1}))  \neq \Qc_u(\Real(\bHH^{(j)}_{1,1}))\RB\!, 
			}		
		where we have selected w.l.o.g. the real part of the (1,1) channel element. 
		Hence, it is sufficient to obtain the probability of disagreement for $\Real(\bHH^{(j)}_{1,1})$.  
		For that purpose, we split each reconstruction level of the quantizer in two parts: 
		The \emph{edge} of the cell and the \emph{center} of the cell. 
		This is done in order to show that, as $P$ increases, the probability of disagreement vanishes if $\bHH^{(1)}_{1,1}$ is in the \emph{center} of the quantization level and, besides this, the probability that $\bHH^{(1)}_{1,1}$ is in the \emph{edge} area also vanishes. 
		We rigorously show it in the following. Before starting, we introduce the simplified notation $\h\expj \triangleq \Real(\bHH^{(j)}_{1,1})$ to ease the readability. 
		Accordingly,  we also introduce the notation $\h \triangleq \Real(\bH_{1,1})$ and $\delta \triangleq \Real(\bdelta\expj_{1,1})$ such that $\h\expj = \zop\expj\h + z\expj\delta\expj$, with  $z\expj = \Pb^{-\alpha\expj}$ and $\zop\expj \triangleq \sqrt{1-(z\expj)^2}$.
		Furthermore, we recall the notations $\zop\expj_{\inv} \triangleq \frac{1}{\zop\expj}$ and $z\expj_{n} \triangleq \frac{z\expj}{\zop\expj}$, introduced in Appendix~\ref{se:proofs}. 
		
			\subsection{Egde and center of the reconstruction level}\label{subse:quant_partition}
		
		Let $\ell_n$ be the $n$-th quantization level of $\Qc_u$, $n\in\Zb$, with $\ell_0 = 0$. Let us define $L_n$ as the input interval that outputs $\ell_n$, i.e., 
			\eqm{
					L_n \triangleq \{x\mid \Qc_u(x) = \ell_n\}.
			}
		$L_n$ has a range $[L^{\min}_n, L^{\max}_n)$ such that	$|L_n| \triangleq L^{\max}_n - L^{\min}_n = \Pb^{-\alpha_q}$. 
		We split $L_n$ in two areas, the \emph{edge} area $E_n$ and the \emph{center} area $C_n$, depicted in~Fig.~\ref{fig:quantizer}. 
		The \emph{edge} area is defined as the part of $L_n$ that is at most at distance $\Pb^{{-c_e\alpha_q}}$ of the boundary of the cell, with $c_e>1$. 
			\eqm{
				E_n \triangleq \left\{ x\!\in\! L_n \mid x - L^{\min}_n < \Pb^{-c_e\alpha_q}  \lor  L^{\max}_n - x  < \Pb^{{-c_e\alpha_q}}\!\right\}\!.\nonumber
			}			
    The \emph{center} area is given by
			\eqm{
				C_n \triangleq  \left\{ x\in L_n \backslash E_n \right\}. 
			}				
		Intuitively, the probability of disagreement is very high if $\h^{(1)}$ lies in the edge area $E_n$, whereas this probability vanishes in the central area $C_n$.
		Mathematically, we have that
			\eqmo{
					&\!\!\!\!\!\Pr\bigg(\Qc_{u}(\h\expo)\neq \Qc_{u}(\h\expj) \bigg) \leq \Pr\bigg(\h\expo\in \bigcup_{n\in\Zb} E_n \bigg) \\
					& \qquad~ + \Pr\bigg(\Qc_{u}(\h\expo)\neq \Qc_{u}(\h\expj) \mid \h\expo\in \bigcup_{n\in\Zb} C_n\bigg).	\label{eq:IEEEproof_vanish_dis_1}
			}
		Let us analyze separately the two probabilities	in the right-hand side of~\eqref{eq:IEEEproof_vanish_dis_1}.

				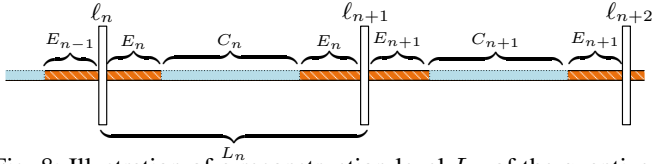
\begin{figure}[t]
					\centering
						\begin{overpic}[width=0.99\columnwidth]{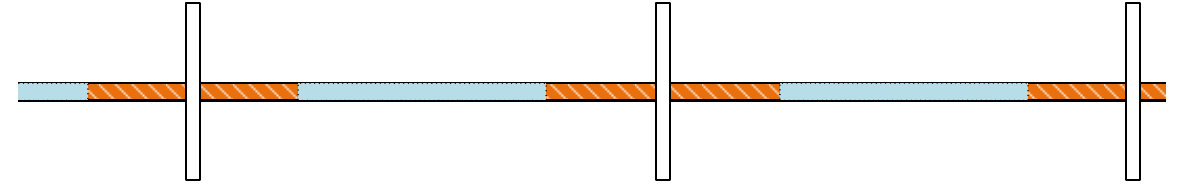}
								{\put(16, -2.25){\parbox{\linewidth}{\small $\underbrace{\phantom{This is a really blank space}}_{L_n}$ }}}
								{\put(7.40, 9.5){\parbox{\linewidth}{\small $\overbrace{\phantom{ccici\,}}^{E_{n-1}}$ }}}
								{\put(17.0, 9.5){\parbox{\linewidth}{\small $\overbrace{\phantom{ccici\,}}^{E_n}$ }}}
								{\put(46.50, 9.5){\parbox{\linewidth}{\small $\overbrace{\phantom{ccici\;}}^{E_n}$ }}}
								{\put(25.5, 9.2){\parbox{\linewidth}{\small $\overbrace{\phantom{ o make space\,}}^{C_n}$ }}}
								{\put(65.75, 9.2){\parbox{\linewidth}{\small $\overbrace{\phantom{ o make space\,}}^{C_{n+1}}$ }}}
								{\put(57.0, 9.75){\parbox{\linewidth}{\small $\overbrace{\phantom{ccici\,}}^{E_{n+1}}$ }}}							
								{\put(87.0, 9.75){\parbox{\linewidth}{\small $\overbrace{\phantom{ccici\,}}^{E_{n+1}}$ }}}							
								{\put(14.5, 16.25){\parbox{\linewidth}{\small $\ell_n$ }}}
								{\put(52.6, 16.25){\parbox{\linewidth}{\small $\ell_{n+1}$ }}}
								{\put(92.6, 16.25){\parbox{\linewidth}{\small $\ell_{n+2}$ }}}
						\end{overpic}~\\ \vspace{1ex}  
						\caption{Illustration of a reconstruction level $L_n$ of the quantizer and the two sub-areas in which we divide it: The central  area $C_n$ and the edge area  $E_n$.} \label{fig:quantizer}
				\end{figure}			
					
			\subsection{Probability of belonging to the edge area}\label{subse:prob_edge}
			Consider an arbitrary quantization level $\ell_n$. 
			Let $f^{\max}_{L_n}$ be the maximum value of the pdf of $\h\expo$ in $L_n = \{x\mid \Qc_u(x) = \ell_n\}$.
			It follows that the probability that $\h\expo$ is in $E_n$  is upper-bounded by
				\eqmo{
					\Pr\big(\h\expo\in E_n \big) 
							&\leq f^{\max}_{L_n} |E_n| \\ %
							& =2 f^{\max}_{L_n} {\Pb^{{-c_e\alpha_q}}},
				}	
			where $|E_n|$ denotes the length of $E_n$. The standard normal distribution has a derivative that is, at most, $1/\sqrt{2\pi e}$. Thus, the probability of  being in $L_n$ satisfies
				\eqmo{
					\Pr\big(\h\expo\in L_n \big) 
							&\geq \big(f^{\max}_{L_n}- 1/\sqrt{2\pi e}|L_n|\big)|L_n| \\ %
							& = \big(f^{\max}_{L_n}- 1/\sqrt{2\pi e}\Pb^{-\alpha_q}\big) {\Pb^{-\alpha_q}}.
				}				
			Hence, the probability that $\h\expo$ is in $E_n$,  given that it is in $L_n$, satisfies for any~$n$ that 
				\eqmo{
					&\Pr\big(\h\expo\in E_n \mid L_n \big)
						 = \frac{\Pr\LB\h\expo\in E_n \RB }{\Pr\LB\h\expo\in L_n \RB }\\
						&\qquad\qquad\quad \leq \frac{2 f^{\max}_{L_n}}{ \big(f^{\max}_{L_n}- 1/\sqrt{2\pi e}\Pb^{-\alpha_q}\big)} \Pb^{-(c_e-1)\alpha_q}.\label{eq:IEEEproof_vanish_dis_1bbedgeq_ch7}
				}		
			Let us define $g_{\max}$ as $g_{\max} \triangleq \max_{n\in\Zb}\frac{2 f^{\max}_{L_n}}{ \big(f^{\max}_{L_n}- 1/\sqrt{2\pi e}\Pb^{-\alpha_q}\big)}$. Note that $g_{\max} = \Theta(1)$. 
			Hence, from~\eqref{eq:IEEEproof_vanish_dis_1bbedgeq_ch7} and the fact that $\sum_{n\in\Zb} \Pr\big(\h\expo\in L_n \big)  = 1$, we can write
				\eqm{
					\Pr\big(\h\expo\in \bigcup_{n\in\Zb} E_n \big)  
						& = \sum_{n\in\Zb} \frac{\Pr\LB\h\expo\in E_n \RB }{\Pr\LB\h\expo\in L_n \RB}\Pr\big(\h\expo\in L_n \big)\nonumber \\
						& \leq  g_{\max} \Pb^{-(c_e-1)\alpha_q} \sum_{n\in\Zb} \Pr\big(\h\expo\in L_n \big)\nonumber \\
						& =  \Oc\big(\Pb^{-(c_e-1)\alpha_q}\big). 
						\label{eq:IEEEproof_vanish_dis_1bbedge_ch7a}
				}
			Consequently, it holds that
				\eqmo{
					\Pr\big(\h\expo\in \bigcup_{n\in\Zb} E_n \big)  
						& =  o\bigg(\frac{1}{\log_2(P)}\bigg). \label{eq:IEEEproof_vanish_dis_1bbedge_ch7}
				}				
			\subsection{Probability of disagreement in the center area}\label{subse:prob_center}
			From the fact that the minimum distance from any point of $C_n$ to the border of  $L_n$ is $\Pb^{-c_e\alpha_q}$, it holds that
				\eqmo{
					&\Pr\LB\mathcal{Q}_{u}(\h\expo)\neq \mathcal{Q}_{u}(\h\expj) \ \Big|\  \h\expo \in \bigcup_{n\in\Zb} C_n\RB\\  
								&\quad \leq \Pr\LB\big|{\h\expo-\h\expj}\big| \geq \Pb^{-c_e\alpha_q}\ \Big|\  \h\expo \in \bigcup_{n\in\Zb} C_n\RB\!.  \label{eq:IEEEproof_quant_1a}
				}	
			Given that, for two events $A,C$, $\Pr(A\mid C) \leq \Pr(A)/\Pr(C)$, it follows that
				\eqm{
					&\Pr\LB\big|{\h\expo-\h\expj}\big| \geq \Pb^{-c_e\alpha_q}\ \Big|\  \h\expo \in \bigcup_{n\in\Zb} C_n\RB\nonumber \\ 
					&\qquad \leq \frac{1}{\Pr\LB\h\expo \in \bigcup_{n\in\Zb} C_n\RB}  \Pr\LB\absn{\h\expo-\h\expj} \geq\Pb^{-c_e\alpha_q}\RB \nonumber\\
					&\qquad\leq \frac{1}{\Pr\LB\h\expo \in \bigcup_{n\in\Zb} C_n\RB} \ \frac{\Eb\left[\big|\h\expo-\h\expj\big|\right]}{\Pb^{-c_e\alpha_q}},  \label{eq:refref}
				}	
			which comes from Markov's Inequality. 
			In the following, we obtain the expectation  $\Eb\left[\big|\h\expo-\h\expj\big|\right]$. 
			Then, 
				\eqmo{
					\!\h\expo-\h\expj
						& = \h(\zop\expo-\zop\expj) + (z\expo\delta\expo - z\expj\delta\expj).\!
				}
			From the assumption of Gaussian variables, it holds that 
				\eqm{
					\h(\zop\expo-\zop\expj) \ & \ \sim\  \Nc\LB 0, \ (\zop\expo-\zop\expj)^2\RB,\\
					z\expo\delta\expo - z\expj\delta\expj \ &\  \sim\  \Nc\LB 0, \ (z\expo)^2 + (z\expj)^2\RB. \label{eq:distr_diff} 
				}
			Since $\h\expo$ is independent of $\delta\expo$ and $\delta\expj$, it follows that
				\eqm{
						\h\expo-\h\expj \ \sim \ \Nc\LB 0, \sigma^2_d\RB,\label{eq:distr_diff2as} 
				}
			where $\sigma^2_d$ is given by
				\eqm{
						\sigma^2_d = (\zop\expo-\zop\expj)^2 + (z\expo)^2 + (z\expj)^2. \label{eq:distr_diff2} 
				}
			Substituting the variables for their values yields
				\eqm{
						\!\sigma^2_d = 2\LB 1 - \sqrt{1 - P^{-\alpha\expo} - P^{-\alpha\expj} + P^{-\alpha\expo-\alpha\expj}}\RB\!. \label{eq:sigma_difference}
				}
			Furthermore, if $\h\expo-\h\expj$ is drawn from a zero-mean Normal distribution of variance $\sigma^2_d$, $\absn{\h\expo-\h\expj}$ is distributed as a half-normal distribution  of mean
				\eqm{
						\Eb\left[\big|\h\expo-\h\expj\big|\right]	&=  \sigma_d \sqrt{\frac{2}{\pi}}.
				}
			From~\eqref{eq:sigma_difference} and the fact that, for any $x$ such that $0\leq x \leq 1$, it holds that $\sqrt{1-x}\geq 1-x$, it follows that
				\eqm{
						\Eb\left[\big|\h\expo-\h\expj\big|\right]	
							& \leq \sqrt{\frac{4}{\pi} \LB {P^{-\alpha\expo} + P^{-\alpha\expj} - P^{-\alpha\expo-\alpha\expj}}\RB} \nonumber\\
							& = \Oc(\Pb^{-\alpha\expj}). \label{eq:proof_quantizer_order_m}
				}		
			Besides this, it holds from~\eqref{eq:IEEEproof_vanish_dis_1bbedge_ch7} that
				\eqmo{
					\Pr\Big(\h\expo \in \bigcup_{n\in\Zb} C_n\Big) 
						& = 1 - \Pr\Big(\h\expo\in \bigcup_{n\in\Zb} E_n\Big)\\
						& = 1 - \Oc(\Pb^{-\alpha_q}).			 \label{eq:proof_quantizer_order_n}
				}
			From~\eqref{eq:refref}, both~\eqref{eq:proof_quantizer_order_m} and~\eqref{eq:proof_quantizer_order_n} lead to
				\eqmo{
					&\Pr\LB\mathcal{Q}_{u}(\h\expo)\neq \mathcal{Q}_{u}(\h\expj) \ \Big|\  \h\expo \in \bigcup_{n\in\Zb} C_n\RB\\  
							&\qquad\quad \leq \frac{1}{\Pr\Big(\h\expo \in \bigcup_{n\in\Zb} C_n\Big) } \  \frac{\Eb\left[\big|\h\expo-\h\expj\big|\right]}{\Pb^{-c_e\alpha_q}} \\ 
							&\qquad\quad = \Oc\LB \Pb^{c_e\alpha_q-\,\alpha\expj}\RB \\
							&\qquad\quad = o\LB \frac{1}{\log_2(P)}\RB.  \label{eq:proof_quant_1aab}									
				}
			The last inequality is obtained only if  $c_e\alpha_q <\alpha\expj$ for any~$j\in\Nb_M$. 
			Thus, it follows from~\eqref{eq:proof_quant_1aab} that it is necessary to satisfy that $c_e\alpha_q <\alpha\expj$ for any~$j\in\Nb_M$.
			Since for any $\alpha_q <\alpha\expj$ we can find a $c_e>1$ such that $c_e\alpha_q <\alpha\expj$, any 	$\alpha_q <\alpha\expj$	 will satisfy~\eqref{eq:proof_quant_1aab} as long as a correct $c_e$ is selected.

			\subsection{Assembling probabilities}\label{subse:end_proof}
			We make use of~\eqref{eq:IEEEproof_vanish_dis_1bbedge_ch7} and~\eqref{eq:proof_quant_1aab} to show that~\eqref{eq:IEEEproof_vanish_dis_1} satisfies
				\eqmo{
					&\Pr\bigg(\!\Qc_{u}(\h\expo)\neq \Qc_{u}(\h\expj)\! \bigg) 
						 \leq  \Pr\bigg(\h\expo\in \bigcup_{n\in\Zb} E_n \bigg) \\
						& \qquad~~ {}+ \Pr\bigg(\Qc_{u}(\h\expo)\neq \Qc_{u}(\h\expj) \mid \h\expo\in\bigcup_{n\in\Zb} C_n\bigg)\!\\
						& \quad = o\LB \frac{1}{\log_2(P)}\RB,
				}
			which concludes the proof of Lemma~\ref{lem:uniform_quantizer}. \qed

		\section{Proof of Lemma~\ref{lem:convergence_error_tx2} (Error on Naive Precoder)}\label{app:lemma_precoder_tx2}
		In this appendix, we prove both Lemma~\ref{lem:convergence_error_tx2} and Corollary~\ref{lem:convergence_error_global}. First, we focus on demonstrating that 
			\eqm{
					\Exp\left[\norm{\vv_{i,2}-\wv_{i,2}}^2\right] = \Oc\LB P^{-\alpha_q}\RB. \label{eq:app_proof_err}
			} 
		The generalization of~\eqref{eq:app_proof_err} for any $\Exp\left[\norm{\vv_{i,k}-\wv_{i,k}}^2\right]$, $k\in\Nb_M\backslash 1$ is straightforward. Then, the proof of the other results in the lemma and the corollary will be shown to follow from~\eqref{eq:app_proof_err}.

		In order to prove~\eqref{eq:app_proof_err}, we make use of the fact that, as presented in Section~\ref{subse:scheme_centr}, we assume that the precoding vectors and matrices can be expressed as a combination of summations, products, and generalized inverses of the channel estimate. Note that,  with the previous operations, it is also possible to compute divisions and norms of the channel estimate coefficients. 

		First, note that both $\wv_{i,2}$ and  $\vv_{i,2}$ are obtained following the same algorithm but based on different information (input). 
		Specifically, $\wv_{i,2}$ is computed on the basis of $\bHH\expt_q = \Qc_u(\bHH\expt)$, where $\Qc_u$ is a scalar uniform quantizer with quantization step $q = \Pb^{-\alpha_q}$, and $\vv_{i,2}$ is computed on the basis of $\bHH\expo$. 
		As in the previous appendix, let $\h\expt_q$ (resp. $\h\expj$ and $\h$) denote the real or imaginary part of an arbitrary element of the matrix $\bHH\expt_q$ (resp. $\bHH\expj$ and $\bH$). 
		Let us define $\h\expt_\varsigma$ as
			\eqm{
				\h\expt_\varsigma \triangleq  \h\expt + \varsigma_q,
			}
		where $\varsigma_q$ is distributed as a binary symmetric distribution with points $[-q,\,q]$, independent of the other variables, such that $\sigma_{\varsigma_q}^2 = q^2$.
		Note that the error $\h\expt_q - \h\expo$ has smaller or equal variance than $\h\expt_\varsigma - \h\expo = \h\expt - \h\expo + \varsigma_q$. Hence, we can assume that $\wv_{i,2}$ is computed on the basis of the estimate $\h\expt_\varsigma$, since increasing the error variance can only hurt. Consequently, the error $\xi \triangleq \h\expt_\varsigma -\h\expo $ has a variance $\sigma_{\xi}^2$ given by
			\eqmo{
				\sigma_{\xi}^2 & \leq \sigma^2_d + \sigma_{\varsigma_q}^2 +  2\sigma_d\sigma_{\varsigma_q} \\
						& = \Oc(P^{-\alpha_q}),\label{eq:order_error_upper_4}
			}
		where $\sigma^2_d$ is given in~\eqref{eq:sigma_difference}.
		Therefore, we can write that
			\eqm{
				 \h\expt_\varsigma = \h\expo + \Pb^{-\alpha_q}\delta_\xi, \label{eq:distr_q_4}
			}
		where $\delta_\xi$ is a variable of variance $\Theta(1)$ and bounded density. 
		We continue by showing that the error-variance scaling remains being at most $\Theta(P^{-\alpha_q})$ after applying addition, product, inverse or pseudo-inverse operations. 
		Afterward, based on those results, we prove~\eqref{eq:app_proof_err}.
		
			\subsection{Error in the addition}\label{subse:error_add}
			Let $\an\expt_\xi$, $\bbn\expt_\xi$, be distributed as~\eqref{eq:distr_q_4}, i.e., $\an\expt_\xi \triangleq \an\expo + \Pb^{-\alpha_q}\delta^\an_\xi$, $\bbn\expt_\xi \triangleq \bbn\expo + \Pb^{-\alpha_q}\delta^\bbn_\xi$, where $\delta^\an_\xi$, $\delta^\bbn_\xi$, are variables with variance $\Theta(1)$ and bounded density. 
			It is easy to see that, for any $\an\expt_\xi$, $\bbn\expt_\xi$,  
				\eqmo{
					\an\expt_\xi + \bbn\expt_\xi  
						&= 	\an\expo + \Pb^{-\alpha_q}\delta^\an_\xi + \bbn\expo + \Pb^{-\alpha_q}\delta^\bbn_\xi \\
								& =  \an\expo + \bbn\expo  + \Pb^{-\alpha_q}(\delta^\an_\xi+\delta^\bbn_\xi).
				}
			This implies that the error variance of the sum is also~$\Oc(P^{-\alpha_q})$ as~\eqref{eq:order_error_upper_4}.  	
			
			\subsection{Error in the product}\label{subse:error_prod}
			In a similar way, it follows that
				\eqmo{
					\!\!\!\!&\!\an\expt_\xi\bbn\expt_\xi  = \LB\an\expo + \Pb^{-\alpha_q}\delta^\an_\xi\RB  \LB\bbn\expo + \Pb^{-\alpha_q}\delta^\bbn_\xi\RB \\
						&\   = \an\expo\bbn\expo  + \Pb^{-\alpha_q}\!\LB \an\expo\delta^\bbn_\xi  + \bbn\expo\delta^\an_\xi+ \Pb^{-\alpha_q}\delta^\an_\xi\delta^\bbn_\xi\RB\!, \label{eq:prod_var}
				}		
			which implies that the product also maintains the scaling of the variance as~$\Oc(P^{-\alpha_q})$.		 
			Moreover, as the sum and product of matrices is a composition of sums and products of its coefficients, the result extends to any two matrices of suitable dimension. 
			\subsection{Error in the inverse}\label{subse:error_inv}
			Let us first assume that $\bA\expt_\xi$ and $\bA\expo$ are square matrices with full rank with probability one, and with coefficients following~\eqref{eq:distr_q_4}. 
			We can then write that
				\eqmo{
						&(\bA\expt_\xi)^{-1} = \LB\bA\expo +  \Pb^{-\alpha_q}\bDelta^\bA_\xi\RB^{-1} \\
						&\ \ \ = (\bA\expo)^{-1}\\
						&\ \ \ \quad \  -  \Pb^{-\alpha_q}\,(\bA\expo)^{-1}\bDelta^\bA_\xi\LB\bA\expo+ \Pb^{-\alpha_q}\bDelta^\bA_\xi\RB^{-1} \! \label{eq:proof_inv_wood}
				}					
			which is obtained from the Woodbury matrix identity\cite{henderson1981}. 
			Hence, the error variance of the inverse is again~$\Oc(P^{-\alpha_q})$. 
			Once that it is proved that the inverse operation generates an error with variance~$\Oc(P^{-\alpha_q})$, we extend it for the Moore–Penrose inverse (pseudo-inverse) $(\bA\expt_\xi)^\dagger$.  
			We assume (as throughout the rest of the document) that each sub-matrix has maximum rank, i.e.,
				\eqm{
						\rank\LB\bA\expt_\xi\in\Cb^{N\times M}\RB = \min(N,M). 
				}
			Let us assume that $\bA\expt_\xi$ is full row-rank matrix, i.e.,  $N\leq M$. Under this assumption, the pseudo-inverse is given by
				\eqm{
						(\bA\expt_\xi)^\dagger = (\bA\expt_\xi)^\He\LB\bA\expt_\xi(\bA\expt_\xi)^\He\RB^{-1}.
				}
			The case in which $\bA\expt_\xi$ is full column-rank matrix ($N\geq M$) will follow the same steps and thus we omit it. 
			It follows from~\eqref{eq:prod_var} that $\bA\expt_\xi(\bA\expt_\xi)^\He = \bA\expo(\bA\expo)^\He +  \Pb^{-\alpha_q}\bDelta_{\text{eq}}$, where $\bDelta_{\text{eq}}$ has variance $\Theta(1)$. 
			This, together with~\eqref{eq:proof_inv_wood}, implies that 
				\eqm{
						\LB\bA\expt_\xi(\bA\expt_\xi)^\He\RB^{-1} = \LB\bA\expo(\bA\expo)^\He\RB^{-1} +   \Pb^{-\alpha_q}\bDelta'_{\text{eq}},
				}		
			and, by applying again~\eqref{eq:prod_var},  it holds that
				\eqm{
						(\bA\expt_\xi)^\dagger =  (\bA\expo)^\He\LB\bA\expo(\bA\expo)^\He\RB^{-1}  +   \Pb^{-\alpha_q}\bDelta^{''}_{\text{eq}},
				}				
			where $\bDelta'_{\text{eq}}$ and $\bDelta^{''}_{\text{eq}}$ have variance $\Theta(1)$. As explained in~\cite{Wiesel2008}, under the assumption that $\bXh$ is a full row-rank matrix, any generalized inverse may be expressed as $\bXh^{-} = \bXh^\dagger + \bP_{\bot}\bU$, \cmtblu{where $\bP_{\bot}=\bI - \bXh^\dagger\bXh$ is the orthogonal projection onto the null space of $\bXh$ and $\bU$ is an arbitrary matrix}. Hence, a similar result could be obtained for a broad set of generalized inverses.

			\subsection{Error variance of the difference of precoders}\label{subse:error_var}
									
			The centralized precoder $\vv_{i,2}$ is based on $\bHH\expo$, i.e., $\bV = \Vc(\bHH\expo)$. The distributed precoder at TX~2 is based on its own quantized CSIT $\bHH\expt_q$, and thus $\wv_{i,2}$ is obtained from $\bW = \Vc(\bHH\expt_q)$. Based on the previous results and the definition of linear precoders, it follows that we can write the distributed precoder based on the  CSIT of TX~$2$~($\bHH\expt_q$) as
				\eqm{\label{error_var_2_app}
						 \wv_{i,2} = \vv_{i,2} + \Pb^{-\alpha_q}\ev_{\wv},
				}
			where $\Exp\big[\norm{\ev_\wv}^2\big] = \Oc(1)$. 
			Consequently,
				\eqm{\label{error_var_2_app2}
						\Exp\left[\norm{\vv_{i,2} - \wv_{i,2}}^2\right] 
									& = \Exp\left[\norm{\vv_{i,2} - (\vv_{i,2} + \Pb^{-\alpha_q}\ev_{\wv})}^2\right]~~~~ \nonumber \\
									& = P^{-\alpha_q}\Exp\left[\norm{\ev_{\wv}}^2\right]  \\
									& = \Oc\LB P^{-\alpha_q}\RB,\nonumber
				}	
			which concludes the proof of~\eqref{eq:app_proof_err} and thus of~\eqref{eq:lemma_diff_2}. 
			Moreover, since $\Exp[\norm{\xv}] \leq \sqrt{\Exp[\norm{\xv}^2]}$, it follows that
				\eqmo{
						\Exp\left[\norm{\vv_{i,j}-\wv_{i,j}}\right]
									& = \Oc\LB \Pb^{-\alpha_q}\RB,
				}				
			which concludes the proof of Lemma~\ref{lem:convergence_error_tx2}. \qed		
		
		\subsection{Proof of Corollary~\ref{lem:convergence_error_global}}		
			In order to prove that $\Exp\left[\norm{\vv_{i}-\wv_{i}}^2\right] = \Oc\LB P^{-\alpha_q}\RB$ for any $i\in\Nb_K$, let us recall that the vector $\wv_{{i}}-\vv_{{i}}$ can be written as
				\eqm{
						\wv_{{i}}-\vv_{{i}} = 
							\begin{bmatrix}
									\bphi_i \\ 
								\wv_{{i,\bar{1}}}-\vv_{{i,\bar{1}}}
							\end{bmatrix}, \label{eq:vec_split1}
				}
			where $\bphi_i$ (defined in~\eqref{eq:eq_1b} as $\bphi_{i}  = \bHH_{\iop,1}^\dagger\bHH_{\iop,\bar{1}}(\vv_{i,\bar{1}} - \wv_{i,\bar{1}})$) is the difference at  TX~$1$, and $\wv_{{i,\bar{1}}}-\vv_{{i,\bar{1}}}$ denotes the difference for the coefficients of all the TXs but TX~$1$, i.e., 
				\eqm{
						\!\wv_{{i,\bar{1}}}-\vv_{{i,\bar{1}}} = [(\wv_{{i,2}}-\vv_{{i,2}})^\Trans,\ \dots,\ (\wv_{{i,K}}-\vv_{{i,K}})^\Trans]^\Trans.\notag
				}
			Let us recall that $N_{\bar{1}}$ has been defined as~$N_{\bar{1}} \triangleq N_T - N_1$, 
			and let us further define $\bG_\bI$ as 
			\begin{align}
				\bG_\bI\triangleq \begin{bmatrix} 
							 \bHH_{\iop,1}^\dagger\bHH_{\iop,\bar{1}}  \\ 
							\bI_{N_{\bar{1}}}
					\end{bmatrix}\!.
			\end{align}
			From the definition of $\bphi_i$, we can rewrite~\eqref{eq:vec_split1} as
				\eqm{
						\wv_{{i}}-\vv_{{i}} = 
							\bG_\bI (\wv_{{i,\bar{1}}}-\vv_{{i,\bar{1}}}). 
							\label{eq:vec_split3}
				}
			Generalizing~\eqref{error_var_2_app} for any $j\in\{\Nb_K\backslash 1\}$, it holds that
				\eqm{
						 \wv_{i,j} = \vv_{i,j} + \Pb^{-\alpha_q}\ev_{\wv,j} \quad\forall j\in\{\Nb_K\backslash 1\},
				}
			where $\Exp[\norm{\ev_{\wv,j}}^2] = \Oc(1)$. 		
			Then, upon defining $\ev_{\wv,\bar{1}} = [\ev_{\wv,2}^\Trans,\ \dots,\ \ev_{\wv,K}^\Trans]^\Trans$, it follows that 	
				\eqmo{
					\ExpB{\smash{{\norm{\wv_{{i}}-\vv_{{i}}} }}^2}
						& = \ExpB{\smash{\norm{\bG_\bI(\wv_{{i,\bar{1}}}-\vv_{{i,\bar{1}}})}}^2}\\
						& = P^{-\alpha_q}\ExpB{\smash{\norm{\bG_\bI\ev_{\wv,\bar{1}}}}^2}\\
						& =  \Oc\LB P^{-\alpha_q}\RB
						\label{eq:chi_int_term2}  
				}				
			since $\ExpB{\smash{\norm{\bG_\bI\ev_{\wv,\bar{1}}}}^2} =\Oc\LB 1\RB$, which proves~\eqref{eq:cor_diff_app_2}. In order to prove~\eqref{eq:cor_diff_app_1},  it follows from from $\Exp[\norm{\xv}] \leq \sqrt{\Exp[\norm{\xv}^2]}$ that
				\eqmo{
						\Exp\left[\norm{\vv_{i}-\wv_{i}}\right]
									& = \Oc\LB \Pb^{-\alpha_q}\RB,
				}				
			which concludes the proof of Corollary~\ref{lem:convergence_error_global}. \qed

\end{appendices}

			\bibliographystyle{IEEEtran}												
			\bibliography{IEEEabrv,Literature} 									

\end{document}